\documentclass[journal]{IEEEtran}
\ifCLASSINFOpdf
\else
\fi

\usepackage{amssymb} 
\usepackage[margin=1in]{geometry}
\usepackage{amsmath,amssymb,amsfonts}
\usepackage{extarrows}
\usepackage{mathrsfs}
\usepackage{cite}
\usepackage{graphicx}
\usepackage{tikz}
\usepackage[ruled]{algorithm2e}
\usepackage{hyperref}
\usepackage{cuted}
\usepackage{float}
\usepackage{scalerel}
\usepackage{changepage}
\usepackage{inputenc}
\usepackage{pifont}
\usepackage{macros}

\newcommand{\dist}{\text{dist}}

\newcommand{\Supp}{\text{Supp}}
\newcommand{\mdist}{\text{dist}_{\text{min}}}

\newcommand{\eq}[1]{\hyperref[eq:#1]{(\ref*{eq:#1})}}
\renewcommand{\sec}[1]{\hyperref[sec:#1]
    {Section~\ref*{sec:#1}}}
\newcommand{\thm}[1]{\hyperref[thm:#1]
    {Theorem~\ref*{thm:#1}}}
\newcommand{\lem}[1]{\hyperref[lem:#1]{Lemma~\ref*{lem:#1}}}
\newcommand{\clm}[1]{\hyperref[claim:#1]{Claim~\ref*{claim:#1}}}
\newcommand{\fct}[1]{\hyperref[fct:#1]{Fact~\ref*{fct:#1}}}

\newcommand{\prop}[1]{\hyperref[prop:#1]
    {Proposition~\ref*{prop:#1}}}
\newcommand{\prob}[1]{\hyperref[prob:#1]
    {Problem~\ref*{prob:#1}}}
\newcommand{\cor}[1]{\hyperref[cor:#1]
    {Corollary~\ref*{cor:#1}}}
\newcommand{\fig}[1]{\hyperref[fig:#1]{Figure~\ref*{fig:#1}}}
\newcommand{\tab}[1]{\hyperref[tab:#1]{Table~\ref*{tab:#1}}}
\newcommand{\alg}[1]{\hyperref[alg:#1]
    {Algorithm~\ref*{alg:#1}}}
\newcommand{\app}[1]{\hyperref[app:#1]
    {Appendix~\ref*{app:#1}}}
\newcommand{\conj}[1]{\hyperref[conj:#1]
    {Conjecture~\ref*{conj:#1}}}
    
\newcommand{\zoto}[1]{{\{0,1\}^{#1}}}
\newcommand{\zoton}{{\zoto{n}}}

\newenvironment{casepar}[1]{
    \begin{adjustwidth}{0.5em}{} \noindent\textbf{#1}.
}{
    \end{adjustwidth}
    \smallskip
}

\usetikzlibrary{shapes.geometric,arrows.meta}
\usetikzlibrary{decorations.pathreplacing,decorations.markings,decorations.pathmorphing,decorations.shapes}

\tikzstyle{zig} = [decorate,decoration={snake,amplitude=.4mm,segment length=2mm}]

\newcommand{\blackdot}[3]{\pgfpathcircle{\pgfpointxyz{#1}{#2}{#3}}{3pt} \color{black} \pgfusepath{fill}}
\newcommand{\whitedot}[3]{%
    \pgfpathcircle{\pgfpointxyz{#1}{#2}{#3}}{3pt} \color{white} \pgfusepath{fill}%
    \pgfpathcircle{\pgfpointxyz{#1}{#2}{#3}}{3pt} \color{black} \pgfusepath{stroke}}
\newcommand{\id}{\mathsf{id}}
\makeatletter
\newcommand{\removelatexerror}{\let\@latex@error\@gobble}
\makeatother

\begin{document}
\onecolumn

%

\newpage

\title{Structured decomposition for reversible Boolean functions}


\author{Jiaqing~Jiang,
        Xiaoming~Sun,
        Yuan~Sun, Kewen~Wu, and~Zhiyu~Xia
\thanks{J.~Jiang, ~X.~Sun, ~Y.~Sun, ~Z.~Xia are with the CAS Key Lab of Network Data Science and Technology, Institute of Computing Technology, Chinese Academy of Sciences, Beijing, China; University of Chinese Academy of Sciences, Beijing, China (e-mail: (\{jiangjiaqing,~sunxiaoming, ~sunyuan2016,~xiazhiyu\}@ict.ac.cn).}
\thanks{K.~Wu is with the School of Electronics Engineering and Computer Science, Peking University, Beijing, China (e-mail: shlw\_kevin@pku.edu.cn).}
}

\IEEEtitleabstractindextext{%

\begin{abstract}
Reversible Boolean function is a one-to-one function which maps $n$-bit input to $n$-bit output. Reversible logic synthesis has been widely studied due to its relationship with low-energy computation as well as quantum computation. 
In this work, we give a structured decomposition for even reversible Boolean functions (RBF). Specifically, for $n\geq 6$, any even $n$-bit RBF can be decomposed to $7$ blocks of $(n-1)$-bit RBF, where $7$ is a constant independent of $n$; and the positions of those blocks have large degree of freedom.
Moreover, if the $(n-1)$-bit RBFs are required to be even as well, we show for $n\geq 10$, $n$-bit RBF can be decomposed to $10$ even $(n-1)$-bit RBFs. For simplicity, we say our decomposition has block depth $7$ and even block depth $10$.

Our result improves Selinger's work in block depth model, by reducing the constant from $9$ to $7$; and from $13$ to $10$ when the blocks are limited to be even. We emphasize that our setting is a bit different from Selinger's. In Selinger's constructive proof, each block is one of two specific positions and thus the decomposition has an alternating structure. 
We relax this restriction and allow each block to act on arbitrary $(n-1)$ bits. 
This relaxation keeps the block structure and provides more candidates when choosing positions of blocks.
\end{abstract}

\begin{IEEEkeywords}
Reversible computation, reversible logic, synthesis method, quantum computation, logic gates, integrated circuits.
\end{IEEEkeywords}}

\maketitle

\IEEEdisplaynontitleabstractindextext

\ifCLASSOPTIONpeerreview
\begin{center} \bfseries EDICS Category: 6–EMRG 
\end{center}
Dear editor:\\

We would like to send this enclosed manuscript entitled \emph{Structured decomposition for reversible Boolean functions}, which we wish to be considered for publication in \emph{the IEEE Transactions on Computer-Aided Design of Integrated Circuits and Systems}. No conflict of interests exists in the submission of this manuscript, and it is approved by all authors for publication.

We deeply appreciate your consideration, and we look forward to receiving comments from the reviewers. If you have any queries, please do not hesitate to contact us at the address below.\\\

\noindent Thank you and best regards.\\

\noindent Yours sincerely,

\noindent Jiaqing Jiang, Xiaoming Sun, Yuan Sun, Kewen Wu, and Zhiyu Xia.

\fi
\IEEEpeerreviewmaketitle

\section{Introduction}
   
    \IEEEPARstart{R}{eversible} Boolean function is a one-to-one function which maps $n$-bit input to $n$-bit output. Combinatorially, it represents a permutation over $\{0,1\}^n$. 
    One historical motivation of studying reversible computation is to reduce the energy consumption caused by computation \cite{bennett1988notes,saeedi2013synthesis,arabzadeh2010rule}. According to Landauer's principle \cite{Landauer1961Irreversibility}, irreversible computation leads to energy dissipation of the order of $KT$ per bit, where $K$ refers to the Boltzmann constant and $T$ is the temperature of the environment. In contrast, if the computing process is reversible, we can in principle use no energy. A classic example of realization of reversible Boolean function --- the billiard ball computer where computation costs no energy --- can be found in Nielsen and Chuang's book \cite{book}. In addition, reversible Boolean functions are widely used in the quantum circuit such as in the modular exponentiation part of Shor's factoring algorithm \cite{shor1999polynomial}, or oracles in Grover's search algorithm \cite{Grover1996A,shende2003synthesis}. Any quantum circuit involving a Boolean function, which is generally irreversible and can not be implemented in quantum circuit directly, such as quantum arithmetic circuit \cite{maslov2008quantum,takahashi2009quantum}, may benefit from the study of reversible Boolean function. 
    
    When implementing an $n$-bit reversible Boolean function, the intuition is to use induction and divide the problem into smaller cases. That is, we try to decompose an $n$-bit reversible Boolean function into a product of several $(n-1)$-bit reversible Boolean functions. This decomposition is generally impossible, since if the $n$-bit reversible Boolean function represents an odd permutation over $\{0,1\}^n$, it can not be implemented by $(n-1)$-bit reversible Boolean functions, which are even when regarded as a permutation on $n$ bits. However, in 2017, Selinger \cite{selinger2018finite} found the decomposition does exist for even $n$-bit reversible Boolean functions and remarkably, the number of required $(n-1)$-bit functions is a constant independent of $n$. More precisely, he proved that an arbitrary even $n$-bit reversible Boolean function can be represented by $9$ $(n-1)$-bit reversible Boolean functions with an alternating structure shown in \fig{alternation}. He also proved that, if we limit the $(n-1)$-bit functions to be even as well, then the number of $(n-1)$-bit functions is at most $13$. For simplicity, in the following we use \textit{block} to refer to the $(n-1)$-bit reversible Boolean function, and \textit{even block} to refer to the even $(n-1)$-bit reversible Boolean function.
    
    \begin{figure}[ht]
        \centering
        \scalebox{0.6}{\begin{tikzpicture}[
    very thick,scale=0.7,
    rec/.style={rectangle,fill=white, draw=black,minimum width=0.8cm,minimum height=2.7cm,scale=0.7}
    ]
    \foreach \y in {0,1,2,3}
        \draw (0.3,\y) -- (13.6+1.7+1.4,\y);
    \foreach \z in {1.7}{
        \foreach \x in {1,2,3,4}{
            \node[rec] at (\z*2*\x-\z,1) {};
            \node[rec] at (\z*2*\x,2) {};
            \foreach \y in {1.25,1.5,1.75}{
                \node at (\z*2*\x-\z*0.5,\y) {\Large $\cdot$};
                \node at (\z*2*\x+\z*0.5,\y) {\Large $\cdot$};
            }
        }
        \node[rec] at (\z*2*5-\z,1) {};
        \foreach \y in {1.25,1.5,1.75}{
            \node at (\z*2*5-\z+\z*0.5,\y) {\Large $\cdot$};
            \node at (\z*2*1-\z-\z*0.5,\y) {\Large $\cdot$};
        }
    }
\end{tikzpicture}}
        \caption{Alternating structure in \cite{selinger2018finite}}\label{fig:alternation}
    \end{figure}
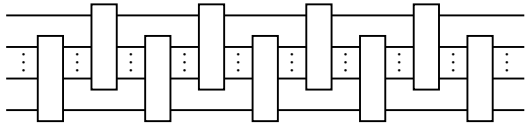
    
    Our main contributions are: we improve the constant from $9$ to $7$ for $n\geq 6$  and  $13$ to $10$ for $n\geq 9$ when limiting the blocks to be even. To be concise, our decomposition has block depth $7$ and even block depth $10$. We should emphasize that our setting is a bit different from Selinger's. In Selinger's work, the decomposition is restricted to an alternating structure.
    Instead of fixing two specific positions, we allow blocks to act on arbitrary $(n-1)$ bits. This relaxation keeps the block structure and provides more candidates when choosing the position of blocks. We believe this relaxation makes the model more flexible in application.
    
    For convenience, we abbreviate reversible Boolean function as RBF. 
    We further say a RBF is controlled RBF if it keeps a certain bit invariant (formal definition is in \sec{preliminary}). 
    Our construction consists of two steps. In the first stage, we prove that an arbitrary even $n$-bit RBF can be transformed into an even controlled RBF by $3$ $(n-1)$-bit blocks and the positions of those low-level blocks have a lot of freedom. It is worth mentioning that the number $3$ is also essentially tight.
    Then we prove that an arbitrary even controlled RBF can be substituted with $5$ blocks, where the third and fourth blocks have many choices as well. While putting it together, we can literally merge the last block in the first step with the first block in the second step, thus providing a $7$-depth full decomposition.
    As a partial result during the construction, we show that two different $(n-1)$-bit blocks are sufficient to formulate the cycle pattern of any even $n$-bit permutation free of 3/5-cycle. We believe this result has some individual interest. Here, cycle pattern is the list $\{c_k\}$, where $c_k$ is the number of cycles of length $k$; and free of 3/5-cycle means $c_3=c_5=0$.
    The limitation that cycle pattern is free of 3/5-cycle is indeed inevitable since we can also prove two $(n-1)$-bit blocks can not compose a single 3/5-cycle. The proof of even block depth $10$ is similar. Since all the proofs in this paper are constructive in essence, our decomposition can be programmed as an efficient algorithm.
    
    In 2003, Shende \emph{et al.} \cite{shende2003synthesis} proved that any even reversible Boolean function can be decomposed into  NOT gates, CNOT gates and Toffoli gates without using temporary storage. Besides, In 2010, Saeedi \emph{et al.} \cite{saeedi2010reversible}  gave an algorithm which synthesizes a given permutation by $7$ building blocks.
    These works focus on decomposing RBFs into smaller pieces, however,  their constructions can not be merged into 7 $(n-1)$-bit blocks, thus they are different from our work.
    There are also some related works about decomposing $n$-bit unitary operator to smaller ones.
    In 2010, Saeedi \emph{et al.} \cite{saeedi2010block} showed how to decompose an arbitrary $n$-bit unitary operator down into $\ell$-bit unitary operators ($\ell<n$) using quantum Shannon decomposition \cite{shende2006synthesis}.

    The structured decomposition may have some potential applications. Though not directly improving results in circuit synthesis, the structure of this decomposition implies some interesting results. For instance, in Selinger's construction in \fig{alternation}, long-distance CNOT, i.e., CNOT between the first and the last bit prohibited by today's quantum devices \cite{devitt2016performing,divincenzo2000physical}, shall be avoided. Although a similar effect can be realized with SWAP gates \cite{zulehner2018efficient}, this result actually indicates that such gate-costing alternatives will not happen frequently in a proper structure. In our setting, the positions of blocks have certain freedom to choose, which makes the construction even more flexible for different potential physical devices \cite{almudever2017engineering,veldhorst2017silicon}.
    
    {\bf Organization of the paper}\quad In \sec{preliminary}, we give formal definitions of the key elements required in expressing problem and formulating proof. Then in \sec{mainresult}, we list our main results and give a proof sketch. In \sec{colored graph} and \sec{new2}, we give detailed proofs to the result of block depth $7$. Specifically, in \sec{colored graph}, we transform an even $n$-bit RBF to an even controlled RBF by $3$ $(n-1)$-bit blocks. In \sec{new2}, we show how to recover an even controlled RBF by $5$ blocks. In addition, an explicit example of our algorithm is put in \sec{example}.
    In \sec{EvenBlock}, we give a proof sketch of the result of even block depth. This proof is similar to the proof of block depth but involves a much more sophisticated analysis. 
    At last, the paper is concluded in \sec{con}. Due to the page limit, the omitted proofs are deferred into the appendix.

    \section{Preliminary}\label{sec:preliminary}
    
    \begin{figure*}[ht]
        \centering
        \begin{tikzpicture}[->,>=stealth]
    \foreach \x in {0,1,2,3}{
        \node (sig\x) at (2.2*\x,0) {$\sigma^{(\x)}$};
    }
    \node at (2.2*3+1.15,-0.066) {$\in A_{\{0,1\}^n}^{(r_1)}$};
    \foreach \x in {4,5,6}{
        \node (sig\x) at (2.2*7-2.2*\x,-2) {$\sigma^{(\x)}$};
    }
    \node (sig7) at (2.2*7-2.2*7,-2) {$\id$};
    \node (sig33) at (2.2*7-2.2*3,-2) {$\widetilde\sigma^{(3)}$};
    \node at (2.2*7-2.2*3+1.15,-2-0.066) {$\in A_{\{0,1\}^n}^{(r_1)}$};

    {
    \small
    \draw (sig0) edge node[above] {$SC_{\{0,1\}^n}^{(r_2)}$} (sig1);
    \draw (sig1) edge node[above] {$SC_{\{0,1\}^n}^{(r_1)}$} (sig2);
    \draw (sig2) edge node[above] {$SC_{\{0,1\}^n}^{(r_2)}$} (sig3);
    \draw (sig3) edge node[right,align=center] {special\qquad\quad\hspace{0pt}\\$SC_{\{0,1\}^n}^{(r_2)}$} (sig33);
    \draw (sig33) edge node[below] {$SC_{\{0,1\}^n}^{(r_1)}$} (sig4);
    \draw (sig4) edge node[below] {$SC_{\{0,1\}^n}^{(r_3)}$} (sig5);
    \draw (sig5) edge node[below] {$SC_{\{0,1\}^n}^{(r_4)}$} (sig6);
    \draw (sig6) edge node[below] {$SC_{\{0,1\}^n}^{(r_1)}$} (sig7);
    }

    \draw[decorate,decoration={brace,amplitude=10pt,raise=20pt},yshift=15pt,-] (sig0) -- (sig3) node [black,midway,yshift=40pt] {\prop{new1}};
    \draw[-,dashed,rounded corners] (8.5,0.4) -- (6.7,0.4) -- (-0.3,-1.7) -- (-0.3,-2.8) -- (11.1,-2.8) -- (11.1,-2) -- cycle;
    \node at (5.4,-3.1) {\prop{new2}};
\end{tikzpicture}
        \caption{Process of the algorithm for \thm{7steps}.}
        \label{fig:process}
    \end{figure*}
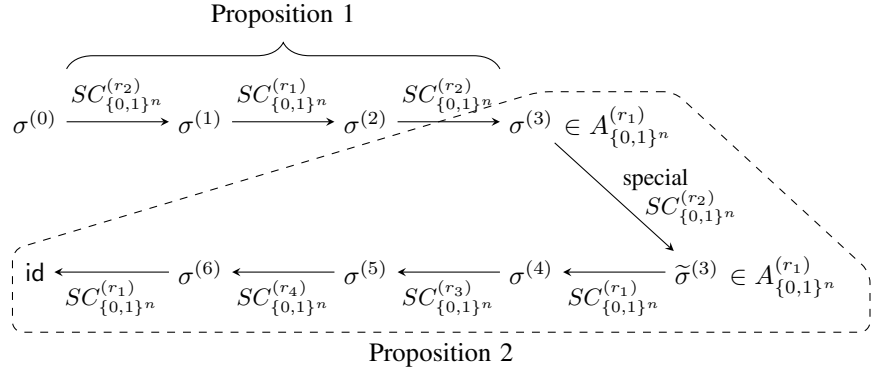

    In general, our work aims to implement an even $n$-bit reversible Boolean function using $(n-1)$-bit reversible Boolean function.
    In order to state our problems and theorems properly, formal definitions are required. 
    
    Denote $[n]$ as $\{1,2,\cdots,n\}$ and $\{0,1\}^n$ as the set of $n$-bit binary strings. Define $S_{\zoton}$ as the group of permutations over $\zoton$; and $A_{\zoton}$ as the group of even permutations over $\zoton$. For any $\sigma\in S_{\zoton}$ and $\bm x,\bm y\in\zoton$, define 
    $$
    \dist^{\sigma}(\bm x,\bm y) = \min\{k \in \mathbb N \mid \sigma^k(\bm x) = \bm y\}
    $$ 
    (if $\bm y$ is not reachable from $\bm x$ under $\sigma$, $\dist^\sigma(\bm x,\bm y)=+\infty$)
    and 
    $\mdist^{\sigma}(\bm x, \bm y) = \min\{\dist^\sigma(\bm x, \bm y), \dist^\sigma(\bm y, \bm x)\}$.
    We also define the support of $\sigma$ as $\Supp(\sigma)=\{\bm{x}| \sigma(\bm{x})\neq \bm{x}\}$. 
    
    Recall that every permutation has a unique cycle decomposition. We say $\sigma$ has a $k$-cycle if there is a cycle of length $k$ in the cycle decomposition. We say $\bm x\in \zoton$ is a fix-point if $\sigma(\bm x)=\bm x$ and a fix-point is a $1$-cycle as well. If $\sigma$ consists of $k_1$-cycle, ..., $k_t$-cycle, we say $\sigma$ is exactly $k_1,\ldots,k_t$-cycle. 
    We may omit $k_i$ if $k_i=1$. For example, we may abbreviate $1,3,4$-cycle as $3,4$-cycle.
    In addition, we say $\sigma$ is free of $l_1/l_2/.../l_s$-cycle if for any $i\in [s],j\in[t], l_i\neq k_j$.
    
    For simplicity, we abbreviate reversible Boolean function as \textit{RBF} and permutation over $\zoton$ as \textit{$n$-bit permutation}.
    Since any $n$-bit RBF can be viewed as a permutation over $\zoton$, thus the set of all $n$-bit RBFs is isomorphic to $S_{\zoton}$. Moreover, we say an $n$-bit RBF is \textit{even} if its corresponding permutation is even. 
    
    Given $\bm x\in \zoton$, write $\bm x_i$ for the value of its $i$-th bit; and
    $\bm x^{\oplus i}:=\bm x_1\cdots\bm x_{i-1}(1-\bm x_i)\bm x_{i+1}\cdots\bm x_n$,
    i.e., $\bm x^{\oplus i}$ is $\bm x$ flipped the $i$-th bit.
    Furthermore, define $\bm x^{\oplus i_1,i_2,\ldots,i_k}$ recursively as $\left(\bm x^{\oplus i_1}\right)^{\oplus i_2,\cdots,i_k}$.
    
    \begin{definition}[Controlled RBF (CRBF)]
        Given $n>0$ and $i\in [n]$, we say $\pi$ is an $n$-bit $i$-CRBF if $\pi\in S_{\zoton}^{(i)}$, where
        $$
            S_{\zoton}^{(i)}:=\big\{\sigma\in S_{\zoton}\ \big|\ \forall\bm x\in\{0,1\}^n,\sigma(\bm x)_i=\bm x_i \big\}.
        $$
        We also define 
        $$
            A_{\zoton}^{(i)}:=\big\{\sigma\in A_{\zoton}\ \big|\ \forall\bm x\in\{0,1\}^n,\sigma(\bm x)_i=\bm x_i \big\}.
        $$
    \end{definition}
    
    An $i$-CRBF keeps the $i$-th bit of any input invariant. For example, if $i=1$, then there exist $f_0,f_1\in S_{\zoto{n-1}}$ such that $\pi(0\bm y)=0f_0(\bm y),\pi(1\bm y)=1f_1(\bm y)$ for any $\bm y\in \{0,1\}^{n-1}$. 
    Moreover, we say $\pi$ is a concurrent controlled RBF (CCRBF) if $f_0=f_1$. Further, when $f_0$ is even, we say $\pi$ is concurrently even; and concurrently odd when $f_0$ is odd. 
    The formal definitions are shown below. 
    
    \begin{definition}[Concurrent Controlled RBF (CCRBF)]\label{concface}
        Given $n>0$ and $i\in [n]$, we say $\pi$ is an $n$-bit $i$-CCRBF if $\pi\in SC_{\zoton}^{(i)}$, where
        \begin{align*}
            SC_{\zoton}^{(i)}
            &:=\big\{\sigma\in S_{\zoton}^{(i)}\ \big|\ \forall\bm x\in\{0,1\}^n,\\
            &\forall k\in[n]\backslash\{i\},\sigma(\bm x)_k=\sigma(\bm x^{\oplus i})_k\big\}.
        \end{align*}
    \end{definition}

    \begin{definition}[Concurrently Even/Odd]
        An $n$-bit $i$-CCRBF $\pi$ can be regarded as an $(n-1)$-bit RBF $\sigma|_{-i}$ on bits $[n]/\{i\}$.
        We say that $\sigma$ is $i$-concurrently even/odd if $\sigma|_{-i}$ is even/odd. 
        Define $AC_{\zoton}^{(i)}$ as the set of $n$-bit concurrently even $i$-CRBF.
    \end{definition}
    
    When dimension $i$ is clear in the context, we simply use concurrently even/odd. 
    Note that no matter whether $\sigma|_{-i}\in S_{\{0,1\}^{n-1}} $ is odd or even, CCRBF $\sigma\in S_{\zoton}$ itself is always even.

    \begin{definition}[Block depth and even block depth]
        Given $n\geq 2$ and $\sigma\in S_{\zoton}$, we say $\sigma$ has block depth $d$ if there exist $\sigma_1,\sigma_2,\ldots,\sigma_d\in\bigcup_{j=1}^n SC_{\zoton}^{(j)}$
        such that $\sigma=\sigma_1\sigma_2\cdots\sigma_d$.
        
        Similarly, we say $\sigma$ has even block depth $d$ if those $\sigma_i\in\bigcup_{j=1}^nAC_{\zoton}^{(j)}$.
    \end{definition}
    
    Notice that the decomposition problem considered here is a bit different from Selinger's work\cite{selinger2018finite}. In Selinger's work, any $\sigma_i$ is in one of two specific positions, thus the decomposition forms an alternating structure as \fig{alternation}. Here we relax the restriction and allow blocks acting on arbitrary $(n-1)$ bits. Thus we consider the block depth instead of alternation depth used in \cite{selinger2018finite}.

    \section{Main results and proof sketch}\label{sec:mainresult}
    
    In the previous work, Selinger \cite{selinger2018finite} proved that an arbitrary even $n$-bit RBF has alternation depth $9$ and even alternation depth $13$.
    Our main contribution is to improve the constant $9$ to $7$ in block depth model and $13$ to $10$ in even block depth model. The main theorems are stated as follows.
    
    \begin{theorem}\label{thm:7steps}
        For $n\geq 6$, any $\sigma\in A_{\zoton}$ has \textit{block depth} $7$.
    \end{theorem}
    
    \begin{theorem}\label{thm:10steps}
        For $n\geq 10$, any $\sigma\in A_{\zoton}$ has even block depth $10$.
    \end{theorem}

    \begin{proof}[Proof sketch of  \thm{7steps}.] To prove \thm{7steps}, we first turn $\sigma$ into an even CRBF by \prop{new1}; then further break the even CRBF down into identity by \prop{new2}.
    We achieve these two steps with $3$ and $5$ blocks respectively.
    By a finer analysis, the last block of the first step and the first block of the second step can be merged. Thus a $7$-block implementation is obtained. The sketch of the whole process is depicted in \fig{process}. 
    \end{proof}
    
    The proof of \thm{10steps} is similar. Before \sec{EvenBlock}, we only focus on the proof of block depth $7$.
    
    \prop{new1} states that we can transform an even $n$-bit RBF to an even CRBF by $3$ CCRBFs with many choices. 
   
    \begin{proposition}\label{prop:new1}
        For $n\geq 4,r_1\in[n]$ and $\sigma\in A_{\zoton}$, there exist at leasts $(n-2)$ different $r_2\in[n]\backslash\{r_1\}$ such that $\sigma\pi_1\sigma_1\pi_2\in A_{\zoton}^{(r_1)} $ holds for some $\sigma_1\in SC_{\zoton}^{(r_1)},\pi_1,\pi_2\in SC_{\zoton}^{(r_2)}$.  
    \end{proposition}
     
    In addition, we also show the tightness of \prop{new1}  by \lem{new1tight} in \sec{colored graph}. It is also worth noting that the proof works for $\sigma\in S_{\zoton}$ (with $\sigma\pi_1\sigma_1\pi_2\in S_{\zoton}^{(r_1)}$) as well. For our purpose, it is more convenient to state it as \prop{new1}.
     
    \prop{new2} states that we can recover any even $n$-bit CRBF by $5$ CCRBFs. 
      
    \begin{proposition}\label{prop:new2}
        For $n\geq 6,r_1\in[n], r_2,r_3,r_4\in[n]\backslash\{r_1\},r_3\neq r_4$ and $\sigma\in A_{\zoton}^{(r_1)}$, there exist $\pi_1\in SC_{\zoton}^{(r_2)}$,
        $\sigma_1,\sigma_2\in SC_{\zoton}^{(r_1)}, \tau_1\in SC_{\zoton}^{(r_3)},\tau_2\in SC_{\zoton}^{(r_4)}$ such that $\sigma\pi_1\sigma_1\tau_1\tau_2\sigma_2=\id$.
    \end{proposition}
    
    The key to the proof of \prop{new2} is the following proposition, which states two $n$-bit CCRBFs can formulate the cycle pattern of any even $n$-bit permutation free of 3/5-cycle. We believe this proposition has some individual interest.
    
    \begin{proposition}\label{prop:new2sum}
        For $n\geq 4$, distinct $r_1,r_2\in[n]$ and $\sigma\in A_{\zoton}$ free of 3/5-cycle, there exist $\pi\in SC_{\zoton}^{(r_1)},\tau\in SC_{\zoton}^{(r_2)}$ such that $\pi\tau$ and $\sigma$ have the same cycle pattern, which is equivalent to that $h\sigma h^{-1}=\pi\tau$ holds for some $h\in S_{\zoton}$.
    \end{proposition}
    
    The proof of \prop{new1} is in \sec{colored graph} and the proof of \prop{new2} and \prop{new2sum} are in \sec{new2}.

    \section{Transforming even \texorpdfstring{$n$}{Lg}-bit RBF to controlled RBF}\label{sec:colored graph}
    
    In this section, we give proof of \prop{new1}. 
    That is, we transform an even $n$-bit RBF $\sigma$ to an even CRBF using $3$ CCRBFs. 
    $\sigma$ may involve $2^n$ elements and have a complicated pattern. However, to transform $\sigma$ to a controlled RBF, which keeps one bit invariant, the key point is whether the $i$-th bit of $\sigma(\bm x)$ equals the $i$-th bit of $\bm x$. 
    So we simplify the representation of a RBF by constructing a black-white cuboid, where the color indicates whether $\sigma(\bm x)_i=\bm x_i$. Then proving \prop{new1} is equivalent to transforming the colored cuboid to white. An explicit example of the whole process of \prop{new1} can be seen in \sec{example}. 
    
    Recall that $n$-bit RBF is in fact a permutation on $\{0,1\}^n$.
    Specifically, we visualize the permutation on a $2\times 2\times 2^{n-2}$ $3$-d cuboid. In \sec{visualizing}, we give the construction for the black-white $3$-d cuboid corresponding to $\sigma$.  After that, in \sec{proof1}, we give a constructive proof to transform the colored cuboid to a white cuboid.
    
    \begin{proof}[Proof sketch of \prop{new1}]
        First we choose arbitrary two different $r_1,r_2\in [n]$ and construct a black-white cuboid. Then we transform the colored cuboid to a canonical form by $SC_{\zoton}^{(r_2)}$ using \lem{canonical}. We also prove in most cases, by \lem{goodcase}, the canonical form can be transformed to a white cuboid by $SC_{\zoton}^{(r_2)}$, $SC_{\zoton}^{(r_1)}$, $SC_{\zoton}^{(r_2)}$. 
        Finally, if the canonical form falls into a bad case, we prove for any $r_3\in[n]\backslash\{r_1,r_2\}$, by checking the new canonical form based on $r_1,r_3$, this case can be tackled with $SC_{\zoton}^{(r_3)}$, $SC_{\zoton}^{(r_1)}$, $SC_{\zoton}^{(r_3)}$ using \lem{badcase}. 
    \end{proof}

    \subsection{Visualizing a permutation on a \texorpdfstring{$3$}{Lg}-d cuboid}\label{sec:visualizing}
   
    Given permutation $\sigma\in S_{\zoton}$, in this section we construct a $3$-d black-white cuboid for $\sigma$ and discuss the effect of transformation, that is the new colored cuboid for $\sigma\tau$, $\tau\in S_{\zoton}$.
    
    Recall that $\sigma$ is a permutation over $2^n$ elements. Fixing $r_1,r_2\in [n]$ and compressing the other $(n-2)$ dimensions, we get a $3$-d cuboid. For example, if $n=4,r_1=1,r_2=2$,  then we compress the remaining two dimensions into one by letting the coordinates to be $00,01,10,11$. 
    We visualize $\sigma$ in \fig{sigmapattern}, where
    \begin{align*}
        \sigma:=&(1001,1100,0101)(1110,0110,0111,\\
        &1111)(1010,0010,0011,1011).
    \end{align*}

    As an example, $1100$ is labelled on $(1,1,00)$, where $00$ represents the third coordinate. The arrows in the figure stand for permutation $\sigma$. In this case, $\sigma(1100)=0101$, so we draw an arrow from $1100$ to $0101$.
    
    \begin{figure}[ht]
        \centering
        \scalebox{0.6}{\begin{tikzpicture}[>=stealth]
    \Large
    \draw[dashed] (0,0,0) -- (6,0,0);
    \draw[->] (6,0,0) -- (7,0,0) node [right] {\huge$\substack{\it other\\ \it dims}$};
    \draw[dashed] (0,0,0) -- (0,2,0);
    \draw[->] (0,2,0) -- (0,3,0) node [above] {$r_1$};
    \draw[->] (0,0,2) -- (0,0,3) node [below left] {$r_2$};

    \draw[dashed] (0,2,2) -- (0,0,2) -- (6,0,2) (6,0,0) -- (6,2,0) -- (0,2,0)
               (0,2,2) -- (6,2,2) -- (6,2,0) (6,2,2) -- (6,0,2);
    \foreach \x in {0,2,4,6} {
        \draw[dashed] (\x,2,0) -- (\x,2,2);
        \draw[dashed] (\x,0,2) -- (\x,2,2) (\x,0,0) -- (\x,2,0);
    }
    \foreach \x in {0,2,4,6} {
        \draw[dashed] (\x,0,0) -- (\x,0,2);
    }

    \large
    \node at (0.5,0.2) {$0000$}; \node at (2.5,0.2) {$0001$};
    \node at (4.5,0.2) {$0010$}; \node at (6.5,0.2) {$0011$};
    \node at (-0.45,2.3) {$1000$}; \node at (2,2.3) {$1001$};
    \node at (4,2.3) {$1010$}; \node at (6,2.3) {$1011$};
    \node at (-1.3,1.1) {$1100$}; \node at (-1+2,1.5) {$1101$};
    \node at (-1+4,1.5) {$1110$}; \node at (-1+6,1.5) {$1111$};
    \node at (-1.3,-0.7) {$0100$}; \node at (1.2,-1.1) {$0101$};
    \node at (1.2+2,-1.1) {$0110$}; \node at (1.2+4,-1.1) {$0111$};

    \draw[-Latex,shorten >= 5pt, shorten <= 5pt] (2,0,2) to (2,2,0);
    \draw[-Latex,shorten >= 5pt, shorten <= 5pt] (2,2,0) to (0,2,2);
    \draw[-Latex,shorten >= 5pt, shorten <= 5pt] (0,2,2) to (2,0,2);
    \foreach \z in {0,2}{
        \draw[-Latex,shorten >= 5pt, shorten <= 5pt] (4,0,\z) to [bend right=10] (6,0,\z);
        \draw[-Latex,shorten >= 5pt, shorten <= 5pt] (6,0,\z) to [bend right=15] (6,2,\z);
        \draw[-Latex,shorten >= 5pt, shorten <= 5pt] (6,2,\z) to [bend right=10] (4,2,\z);
        \draw[-Latex,shorten >= 5pt, shorten <= 5pt] (4,2,\z) to [bend right=15] (4,0,\z);
    }

\end{tikzpicture}}
        \caption{Visualize $\sigma$ on a 3-d cuboid}\label{fig:sigmapattern}
    \end{figure}
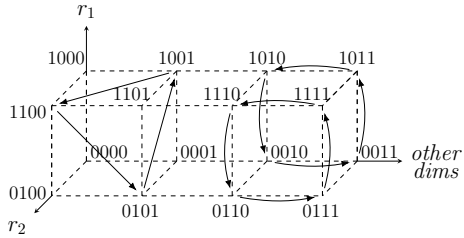
    
    The graph reflects both pattern and structure of the permutation. If we exert a CCRBF 
    \begin{align*}
        \tau=&(1010,1011,0011,0010)(1110,\\
         &1111,0111,0110)\in SC_{\zoto{4}}^{(2)}
    \end{align*}
    on $\sigma$, it will have the same effect on the front and back face of the cuboid, eliminating the two $4$-cycles. That is, the $3$-d cuboid corresponding to $\sigma\tau$ will only have a $3$-cycle.

    Back to \prop{new1}, here we aim to eliminate cycles which have overlap with both top and bottom face.
    To further simplify the notation, we transform the cuboid with arrow pattern into a cuboid with black-white colored nodes. That is, we paint coordinate $\bm x\in\{0,1\}^n$ black if $\sigma(\bm x)_{r_1}\not=\bm x_{r_1}$ as shown in \fig{sigmacolor}. Intuitively, the black node means that $\sigma(\bm x)$ is in a wrong face.
    
    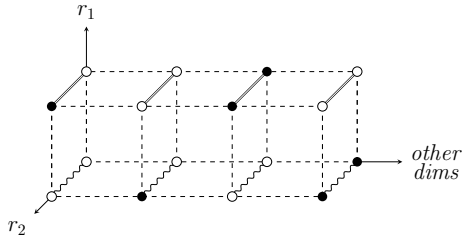
\begin{figure}[ht]
        \centering
        \scalebox{0.6}{\begin{tikzpicture}[>=stealth]
    \Large
    \draw[dashed] (0,0,0) -- (6,0,0);
    \draw[->] (6,0,0) -- (7,0,0) node [right] {\huge$\substack{\it other\\ \it dims}$};
    \draw[dashed] (0,0,0) -- (0,2,0);
    \draw[->] (0,2,0) -- (0,3,0) node [above] {$r_1$};
    \draw[->] (0,0,2) -- (0,0,3) node [below left] {$r_2$};

    \draw[dashed] (0,2,2) -- (0,0,2) -- (6,0,2) (6,0,0) -- (6,2,0) -- (0,2,0)
               (0,2,2) -- (6,2,2) -- (6,2,0) (6,2,2) -- (6,0,2);
    \foreach \x in {0,2,4,6} {
        \draw[double] (\x,2,0) -- (\x,2,2);
        \draw[dashed] (\x,0,2) -- (\x,2,2) (\x,0,0) -- (\x,2,0);
    }
    \foreach \x in {0,2,4,6} {
        \draw[zig] (\x,0,0) -- (\x,0,2);
    }
    \foreach \x/\y/\z in {0/0/0,2/0/0,4/0/0,0/2/0,2/2/0,6/2/0,0/0/2,4/0/2,2/2/2,6/2/2}{
        \whitedot{\x}{\y}{\z};
    }
    \foreach \x/\y/\z in {6/0/0,4/2/0,0/2/2,2/0/2,6/0/2,4/2/2}{
        \blackdot{\x}{\y}{\z};
    }
\end{tikzpicture}}
        \caption{Visualize $\sigma$ on a colored cuboid} \label{fig:sigmacolor}
    \end{figure}

    Now we consider the cuboid of $\sigma\pi$ with some permutation $\pi$. For example, if $\pi$ pushes $\bm x$ to the opposite face, the color of $\bm x$ in cuboid for $\sigma\pi$ will be the opposite of original $\bm \pi(x)$'s in cuboid for $\sigma$. 
    That is, assuming $\pi(\bm x)=\bm x'$ and $\bm x_{r_1}\neq \bm x'_{r_1}$, if $\sigma(\bm x')_{r_1}\neq \bm x'_{r_1}$, then $\sigma(\bm x')_{r_1}=\bm x_{r_1}$ (i.e., $\sigma(\pi(\bm x))_{r_1}=\bm x_{r_1}$), vice versa. An example is in \fig{sigmabeforepi} and \fig{sigmapi} for $\pi=(1100,0101)(1000,0001)\in SC_{\zoto{4}}^{(2)}$.

    \begin{figure}[ht]
        \centering
        \scalebox{0.6}{\begin{tikzpicture}[>=stealth]
    \Large
    \draw[dashed] (0,0,0) -- (6,0,0);
    \draw[->] (6,0,0) -- (7,0,0) node [right] {\huge$\substack{\it other\\ \it dims}$};
    \draw[dashed] (0,0,0) -- (0,2,0);
    \draw[->] (0,2,0) -- (0,3,0) node [above] {$r_1$};
    \draw[->] (0,0,2) -- (0,0,3) node [below left] {$r_2$};

    \draw[dashed] (0,2,2) -- (0,0,2) -- (6,0,2) (6,0,0) -- (6,2,0) -- (0,2,0)
               (0,2,2) -- (6,2,2) -- (6,2,0) (6,2,2) -- (6,0,2);
    \foreach \x in {0,2,4,6} {
        \draw[double] (\x,2,0) -- (\x,2,2);
        \draw[dashed] (\x,0,2) -- (\x,2,2) (\x,0,0) -- (\x,2,0);
    }
    \foreach \x in {0,2,4,6} {
        \draw[zig] (\x,0,0) -- (\x,0,2);
    }
    \foreach \x/\y/\z in {0/0/0,2/0/0,4/0/0,0/2/0,2/2/0,6/2/0,0/0/2,4/0/2,2/2/2,6/2/2}{
        \whitedot{\x}{\y}{\z};
    }
    \foreach \x/\y/\z in {6/0/0,4/2/0,0/2/2,2/0/2,6/0/2,4/2/2}{
        \blackdot{\x}{\y}{\z};
    }

    \draw[Latex-Latex,shorten >= 5pt, shorten <= 5pt] (2,0,0) to (0,2,0);
    \draw[Latex-Latex,shorten >= 5pt, shorten <= 5pt] (2,0,2) to (0,2,2);

\end{tikzpicture}}
        \caption{Colored cuboid for $\sigma$. Arrows refer to $\pi$.}\label{fig:sigmabeforepi}
        \scalebox{0.6}{\begin{tikzpicture}[>=stealth]
    \Large
    \draw[dashed] (0,0,0) -- (6,0,0);
    \draw[->] (6,0,0) -- (7,0,0) node [right] {\huge$\substack{\it other\\ \it dims}$};
    \draw[dashed] (0,0,0) -- (0,2,0);
    \draw[->] (0,2,0) -- (0,3,0) node [above] {$r_1$};
    \draw[->] (0,0,2) -- (0,0,3) node [below left] {$r_2$};

    \draw[dashed] (0,2,2) -- (0,0,2) -- (6,0,2) (6,0,0) -- (6,2,0) -- (0,2,0)
               (0,2,2) -- (6,2,2) -- (6,2,0) (6,2,2) -- (6,0,2);
    \foreach \x in {0,2,4,6} {
        \draw[double] (\x,2,0) -- (\x,2,2);
        \draw[dashed] (\x,0,2) -- (\x,2,2) (\x,0,0) -- (\x,2,0);
    }
    \foreach \x in {0,2,4,6} {
        \draw[zig] (\x,0,0) -- (\x,0,2);
    }
    \foreach \x/\y/\z in {0/0/0,2/0/2,4/0/0,0/2/2,2/2/0,6/2/0,0/0/2,4/0/2,2/2/2,6/2/2}{
        \whitedot{\x}{\y}{\z};
    }
    \foreach \x/\y/\z in {6/0/0,4/2/0,0/2/0,2/0/0,6/0/2,4/2/2}{
        \blackdot{\x}{\y}{\z};
    }
\end{tikzpicture}}
        \caption{Colored cuboid for $\sigma\pi$}\label{fig:sigmapi}
    \end{figure}

    Using colored cuboid, for some $\pi'$, the cuboid for $\sigma\pi'$ is white if and only if $\sigma\pi'\in S_{\zoton}^{(r_1)}$. To prove \prop{new1}, it suffices to show that we can transform any black-white cuboid into a white cuboid, using CCRBFs. 

    For simplicity, as shown in \fig{sigmapi}, we use a double line to connect $\bm x$ and $\bm x^{\oplus r_2}$ for all $\bm x$ with $\bm x_{r_1}=1$; and zigzag line to connect $\bm x$ and $\bm x^{\oplus r_2}$ for all $\bm x$ with $\bm x_{r_1}=0$. Let $a_1,a_2,a_3,a_4$ be the number of \tikz[thick]{\draw[double] (0,1,0) -- (0,1,1); \whitedot{0}{1}{0}; \blackdot{0}{1}{1};},
\tikz[thick]{\draw[double] (0,1,0) -- (0,1,1); \blackdot{0}{1}{0}; \whitedot{0}{1}{1};},
\tikz[thick]{\draw[double] (0,1,0) -- (0,1,1); \blackdot{0}{1}{0}; \blackdot{0}{1}{1};},
\tikz[thick]{\draw[double] (0,1,0) -- (0,1,1); \whitedot{0}{1}{0}; \whitedot{0}{1}{1};}
 and $b_1,b_2,b_3,b_4$ be the number of \tikz{\draw[zig] (0,1,0) -- (0,1,1); \whitedot{0}{1}{0}; \blackdot{0}{1}{1};},
\tikz{\draw[zig] (0,1,0) -- (0,1,1); \blackdot{0}{1}{0}; \whitedot{0}{1}{1};},
\tikz{\draw[zig] (0,1,0) -- (0,1,1); \blackdot{0}{1}{0}; \blackdot{0}{1}{1};},
\tikz{\draw[zig] (0,1,0) -- (0,1,1); \whitedot{0}{1}{0}; \whitedot{0}{1}{1};}
 respectively.

    \subsection{Transforming \texorpdfstring{$\sigma$}{Lg} to controlled  permutation}\label{sec:proof1}

    In this section, we transform the given permutation to CRBF. Following previous section, we construct a colored cuboid for $\sigma\in A_{\zoton}$ and calculate corresponding $a_i$'s, $ b_i$'s. According to $a_i$'s, $ b_i$'s, we transform $\sigma$ to $A_{\zoton}^{(r_1)}$ using \lem{goodcase} or \lem{badcase}. We also show the tightness of 3 steps by \lem{new1tight}. 

    Firstly we prove \lem{goodcase} to show most cases are solvable by $SC_{\zoton}^{(r_2)}$, $SC_{\zoton}^{(r_1)}$, $SC_{\zoton}^{(r_2)}$. Since the number of black nodes in lower and upper faces is the same, it is easy to see $a_3+a_4+b_3+b_4$ is even. 
    
    \begin{lemma}\label{lem:goodcase}
        There exist $\sigma_1\in SC_{\zoton}^{(r_1)}$ and $\pi_1,\pi_2\in SC_{\zoton}^{(r_2)}$ such that $\sigma\pi_1\sigma_1\pi_2\in A_{\zoton}^{(r_1)}$ if
        \begin{enumerate}
            \item[1.] $a_3+a_4+b_3+b_4>2$ holds or, 
            \item[2.] $a_3+a_4+b_3+b_4=2$ and $\min\{b_1+a_2,a_1+b_2\}>0$ hold or, 
            \item[3.] $a_3+a_4+b_3+b_4=0$ and $b_1+a_2$ is even (equivalently $a_1+b_2$ is even) hold. 
        \end{enumerate}
    \end{lemma}
  
    To give specific constructions, we first transform the colored cuboid to a \textit{canonical form} by \lem{canonical}. Then we classify them into different cases and solve case by case. 
    
    A \textit{canonical form} is a colored cuboid only containing $3$ kinds of matching pairs (``cards'') along the compressed dimensions, which are
        \tikz[thick,scale=0.7]{
            \draw[dashed] (0,2,0) -- (0,1,0) (0,2,1) -- (0,1,1);
            \draw[double] (0,2,0) -- (0,2,1); \whitedot{0}{2}{0}; \blackdot{0}{2}{1}; 
            \draw[zig] (0,1,0) -- (0,1,1); \blackdot{0}{1}{0}; \whitedot{0}{1}{1}; 
        },
        \tikz[thick,scale=0.7]{
            \draw[dashed] (0,2,0) -- (0,1,0) (0,2,1) -- (0,1,1);
            \draw[double] (0,2,0) -- (0,2,1); \whitedot{0}{2}{0}; \blackdot{0}{2}{1}; 
            \draw[zig] (0,1,0) -- (0,1,1); \whitedot{0}{1}{0}; \blackdot{0}{1}{1}; 
        },
        \tikz[thick,scale=0.7]{
            \draw[dashed] (0,2,0) -- (0,1,0) (0,2,1) -- (0,1,1);
            \draw[double] (0,2,0) -- (0,2,1); \whitedot{0}{2}{0}; \whitedot{0}{2}{1}; 
            \draw[zig] (0,1,0) -- (0,1,1); \whitedot{0}{1}{0}; \whitedot{0}{1}{1}; 
        }. 
    We call them A-card, B-card, and C-card; and the numbers of these three kinds are $\alpha,\beta,\gamma$ respectively.
    
    If $a_2+a_3\leq b_2+b_3$, we can use \lem{canonical} to transform the colored cuboid to a canonical form.
    
    \begin{lemma}\label{lem:canonical}
        If $a_2+a_3\leq b_2+b_3$, there exists $\pi\in SC_{\zoton}^{(r_2)}$ such that the colored cuboid for $\sigma\pi$ is of canonical form.
    \end{lemma}
    \begin{proof}
        Recall that the color of a node $\bm x$ refers to whether $\sigma(\bm x)$ is in the correct face. So if coordinate $\bm x'$ is black and $\bm x'_{r_1}\neq\bm x_{r_1}$, then coordinate $\bm x$ will be white after swapping $\bm x$ and $\bm x'$, vice versa. 
        See \fig{sigmapi} as an example.
        
        We first apply $\tau\in SC_{\zoton}^{(r_2)}$ such that the cuboid for $\sigma\tau$ satisfies $a_4'=b_4'$, $a_2'=a_3'=b_3'=0$. Then we use $\tau'\in SC_{\zoton}^{(r_2)}$ to rearrange the nodes, such that the cuboid for $\sigma\tau\tau'$ is a canonical form.
        $\tau$ is achieved by the following algorithm.
        
        \begin{center}
\begin{algorithm}[ht]
    \footnotesize
    \caption{\small Canonical form (\textsc{Canonical})}\label{alg:canonical}
    \LinesNumbered
    \DontPrintSemicolon
    \KwIn{$\sigma\in S_{\zoton}$ and its colored cuboid}
    \KwOut{Canonical form of the cuboid}
    Swap 
    \tikz[thick]{\draw[double] (0,1,0) -- (0,1,1); \blackdot{0}{1}{0}; \blackdot{0}{1}{1};} with
    \tikz[thick]{\draw[zig] (0,1,0) -- (0,1,1); \blackdot{0}{1}{0}; \blackdot{0}{1}{1};} 
    until $a_3$ or $b_3$ reaches zero\;
    \If {$a_3=0$}{
        Swap \tikz[thick]{\draw[double] (0,1,0) -- (0,1,1); \blackdot{0}{1}{0}; \whitedot{0}{1}{1};} 
        (or
        \tikz[thick]{\draw[double] (0,1,0) -- (0,1,1); \whitedot{0}{1}{0}; \blackdot{0}{1}{1};}  
        )
        with
        \tikz[thick]{\draw[zig] (0,1,0) -- (0,1,1); \blackdot{0}{1}{0}; \blackdot{0}{1}{1};} 
        until $b_3$ reaches zero\; 
    }\Else{
        Swap 
        \tikz[thick]{\draw[double] (0,1,0) -- (0,1,1); \blackdot{0}{1}{0}; \blackdot{0}{1}{1};}
        with
        \tikz[thick]{\draw[zig] (0,1,0) -- (0,1,1); \blackdot{0}{1}{0}; \whitedot{0}{1}{1};} 
        (or
        \tikz[thick]{\draw[zig] (0,1,0) -- (0,1,1); \whitedot{0}{1}{0}; \blackdot{0}{1}{1};} 
        )
        until $a_3$ reaches zero\;
    }
    Swap 
    \tikz[thick]{\draw[double] (0,1,0) -- (0,1,1); \blackdot{0}{1}{0}; \whitedot{0}{1}{1};} with
    \tikz[thick]{\draw[zig] (0,1,0) -- (0,1,1); \blackdot{0}{1}{0}; \whitedot{0}{1}{1};}
    until $a_2$ reaches zero\;
\end{algorithm}
\end{center} 
        \vspace{-25pt}
        
        The correctness comes from the following observation. 
        Since the number of black nodes is the same in the top and bottom face, if there is a black node in one face, the opposite face has one as well. 
        Therefore, in line $3$ the number of
        \tikz[thick]{\draw[double] (0,1,0) -- (0,1,1); \blackdot{0}{1}{0}; \whitedot{0}{1}{1};} 
        and 
        \tikz[thick]{\draw[double] (0,1,0) -- (0,1,1); \whitedot{0}{1}{0}; \blackdot{0}{1}{1};}  
        is no fewer than 
        \tikz[thick]{\draw[zig] (0,1,0) -- (0,1,1); \blackdot{0}{1}{0}; \blackdot{0}{1}{1};};
        in line $6$ the number of
        \tikz[thick]{\draw[zig] (0,1,0) -- (0,1,1); \blackdot{0}{1}{0}; \whitedot{0}{1}{1};} 
        and
        \tikz[thick]{\draw[zig] (0,1,0) -- (0,1,1); \whitedot{0}{1}{0}; \blackdot{0}{1}{1};}  
        is no fewer than
        \tikz[thick]{\draw[double] (0,1,0) -- (0,1,1); \blackdot{0}{1}{0}; \blackdot{0}{1}{1};}.
        Since $a_2+a_3\leq b_2+b_3$, it can be verified when algorithm executes in line $8$, the number of
        \tikz[thick]{\draw[double] (0,1,0) -- (0,1,1); \blackdot{0}{1}{0}; \whitedot{0}{1}{1};}
        is no more than
        \tikz[thick]{\draw[zig] (0,1,0) -- (0,1,1); \blackdot{0}{1}{0}; \whitedot{0}{1}{1};} .
        After performing this algorithm, we have $a_2'=a_3'=b_3'=0$ and $a_4'=b_4'$. 
        
        Then we rearrange the nodes to form A-, B-, C-cards. Since $a_4'=b_4'$, by some permutation $\tau'\in SC_{\zoton}^{(r_2)}$, we can assure that the colored cuboid corresponding to  $\sigma\tau\tau'$ only has these three kind of cards.
        Thus let $\pi=\tau\tau'$, then the colored cuboid for $\sigma\pi$ is of canonical form.

        Since the number of 
          \tikz[thick]{\draw[double] (0,1,0) -- (0,1,1); \blackdot{0}{1}{0}; \whitedot{0}{1}{1};} and
                \tikz[thick]{\draw[zig] (0,1,0) -- (0,1,1); \whitedot{0}{1}{0}; \blackdot{0}{1}{1};} 
           is invariant in \alg{canonical}, as well as \tikz[thick]{\draw[double] (0,1,0) -- (0,1,1); \whitedot{0}{1}{0}; \blackdot{0}{1}{1};} and
                \tikz[thick]{\draw[zig] (0,1,0) -- (0,1,1); \blackdot{0}{1}{0}; \whitedot{0}{1}{1};},
        we have $\alpha=\frac12(a_1-a_2+b_2-b_1),\beta=b_1+a_2,\gamma=\frac12(a_3+a_4+b_3+b_4)$.
 
    \end{proof}
   
    Now we give the proof of \lem{goodcase}.

    \begin{proof}[Proof of \lem{goodcase}]\label{proof:goodcase}
    W.l.o.g, assume $a_2+a_3\leq b_2+b_3$. Using \lem{canonical}, we transform the colored cuboid to a canonical form with $\pi'\in SC_{\{0,1\}^n}^{(r_2)}$. Record the number of the 3 kind of cards, i.e., $\alpha,\beta,\gamma$. 
    
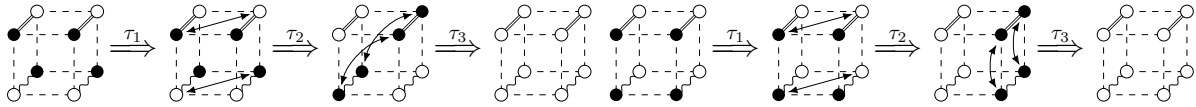
        \begin{figure}[ht]                  
            \centering
            \input{./pics/part1split00x.tex}  
            \input{./pics/part1split00x2.tex}
        \end{figure}

    First notice that if we pair two A-cards or two B-cards, the paired A-cards and B-cards can be transformed to C-cards by the following permutations where $\tau_1\in SC_{\zoton}^{(r_2)}$, $\tau_2\in SC_{\zoton}^{(r_1)}$, $\tau_3\in SC_{\zoton}^{(r_2)}$:

        This approach solves the $3^{\it rd}$ case directly
        and reduces the $1^{\it st}$ case to the following $3$ subcases.
        Since these card groups can be tackled in parallel, in final construction, $\pi_1=\pi'\tau_1$, $\sigma_1=\tau_2$, and $\pi_2=\tau_3$. 

        \begin{itemize}
            \item $\alpha=1,\beta=1,\gamma\geq 2$ : This graph shows how to tackle $1$ A-card and $1$ B-card with $2$ C-cards.
                \begin{center}
                    \scalebox{0.8}
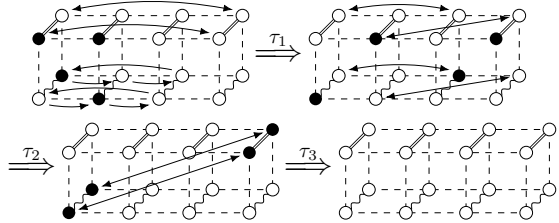
{\tikz{
    \Large
    \draw[dashed] (0,1,0) -- (0,0,0) (0,1,1) -- (0,0,1) (1,1,0) -- (1,0,0) (1,1,1) -- (1,0,1)
        (2,1,0) -- (2,0,0) (2,1,1) -- (2,0,1) (3,1,0) -- (3,0,0) (3,1,1) -- (3,0,1) 
        (0,1,0) -- (1,1,0) -- (2,1,0) -- (3,1,0) (0,1,1) -- (1,1,1) -- (2,1,1) -- (3,1,1) 
        (0,0,0) -- (1,0,0) -- (2,0,0) -- (3,0,0) (0,0,1) -- (1,0,1) -- (2,0,1) -- (3,0,1);
    \draw[double] (0,1,0) -- (0,1,1); \whitedot{0}{1}{0}; \blackdot{0}{1}{1}; 
    \draw[zig] (0,0,0) -- (0,0,1); \blackdot{0}{0}{0}; \whitedot{0}{0}{1};
    \draw[double] (1,1,0) -- (1,1,1); \whitedot{1}{1}{0}; \blackdot{1}{1}{1}; 
    \draw[zig] (1,0,0) -- (1,0,1); \whitedot{1}{0}{0}; \blackdot{1}{0}{1};
    \draw[double] (2,1,0) -- (2,1,1); \whitedot{2}{1}{0}; \whitedot{2}{1}{1}; 
    \draw[zig] (2,0,0) -- (2,0,1); \whitedot{2}{0}{0}; \whitedot{2}{0}{1};
    \draw[double] (3,1,0) -- (3,1,1); \whitedot{3}{1}{0}; \whitedot{3}{1}{1}; 
    \draw[zig] (3,0,0) -- (3,0,1); \whitedot{3}{0}{0}; \whitedot{3}{0}{1};

    \foreach \z in {0,1}{
        \draw[Latex-,shorten >= 5pt, shorten <= 5pt,very thin,yshift=0.7mm] (0,0,\z) to [bend left=10] (2,0,\z);
        \draw[Latex-,shorten >= 5pt, shorten <= 5pt,very thin,yshift=-0.7mm] (2,0,\z) to [bend left=10] (1,0,\z);
        \draw[Latex-,shorten >= 5pt, shorten <= 5pt,very thin,yshift=-0.7mm] (1,0,\z) to [bend left=10] (0,0,\z);
        \draw[Latex-Latex,shorten >= 5pt, shorten <= 5pt,very thin,yshift=0.7mm] (0,1,\z) to [bend left=10] (3,1,\z);
    }

    \node at (3.6,0.5,0) {$\xLongrightarrow{\tau_1}$};

    \draw[dashed] (4.6,1,0) -- (4.6,0,0) (4.6,1,1) -- (4.6,0,1) (5.6,1,0) -- (5.6,0,0) (5.6,1,1) -- (5.6,0,1)
        (6.6,1,0) -- (6.6,0,0) (6.6,1,1) -- (6.6,0,1) (7.6,1,0) -- (7.6,0,0) (7.6,1,1) -- (7.6,0,1) 
        (4.6,1,0) -- (5.6,1,0) -- (6.6,1,0) -- (7.6,1,0) (4.6,1,1) -- (5.6,1,1) -- (6.6,1,1) -- (7.6,1,1) 
        (4.6,0,0) -- (5.6,0,0) -- (6.6,0,0) -- (7.6,0,0) (4.6,0,1) -- (5.6,0,1) -- (6.6,0,1) -- (7.6,0,1);
    \draw[double] (4.6,1,0) -- (4.6,1,1); \whitedot{4.6}{1}{0}; \whitedot{4.6}{1}{1}; 
    \draw[zig] (4.6,0,0) -- (4.6,0,1); \whitedot{4.6}{0}{0}; \blackdot{4.6}{0}{1};
    \draw[double] (5.6,1,0) -- (5.6,1,1); \whitedot{5.6}{1}{0}; \blackdot{5.6}{1}{1}; 
    \draw[zig] (5.6,0,0) -- (5.6,0,1); \whitedot{5.6}{0}{0}; \whitedot{5.6}{0}{1};
    \draw[double] (6.6,1,0) -- (6.6,1,1); \whitedot{6.6}{1}{0}; \whitedot{6.6}{1}{1}; 
    \draw[zig] (6.6,0,0) -- (6.6,0,1); \blackdot{6.6}{0}{0}; \whitedot{6.6}{0}{1};
    \draw[double] (7.6,1,0) -- (7.6,1,1); \whitedot{7.6}{1}{0}; \blackdot{7.6}{1}{1}; 
    \draw[zig] (7.6,0,0) -- (7.6,0,1); \whitedot{7.6}{0}{0}; \whitedot{7.6}{0}{1};

    \foreach \y in {0,1}{
        \draw[Latex-Latex,shorten >= 5pt, shorten <= 5pt,very thin,yshift=0.7mm] (4.6,\y,0) to [bend left=10] (6.6,\y,0);
        \draw[Latex-Latex,shorten >= 5pt, shorten <= 5pt,very thin] (5.6,\y,1) to (7.6,\y,0);
    }
}}
\vspace{5pt}\\
\hspace{-5pt}\scalebox{0.8}{%
\tikz{
    \Large
    \node at (-1,0.5,0) {$\xLongrightarrow{\tau_2}$};

    \draw[dashed] (0,1,0) -- (0,0,0) (0,1,1) -- (0,0,1) (1,1,0) -- (1,0,0) (1,1,1) -- (1,0,1)
        (2,1,0) -- (2,0,0) (2,1,1) -- (2,0,1) (3,1,0) -- (3,0,0) (3,1,1) -- (3,0,1) 
        (0,1,0) -- (1,1,0) -- (2,1,0) -- (3,1,0) (0,1,1) -- (1,1,1) -- (2,1,1) -- (3,1,1) 
        (0,0,0) -- (1,0,0) -- (2,0,0) -- (3,0,0) (0,0,1) -- (1,0,1) -- (2,0,1) -- (3,0,1);
    \draw[double] (0,1,0) -- (0,1,1); \whitedot{0}{1}{0}; \whitedot{0}{1}{1}; 
    \draw[zig] (0,0,0) -- (0,0,1); \blackdot{0}{0}{0}; \blackdot{0}{0}{1};
    \draw[double] (1,1,0) -- (1,1,1); \whitedot{1}{1}{0}; \whitedot{1}{1}{1}; 
    \draw[zig] (1,0,0) -- (1,0,1); \whitedot{1}{0}{0}; \whitedot{1}{0}{1};
    \draw[double] (2,1,0) -- (2,1,1); \whitedot{2}{1}{0}; \whitedot{2}{1}{1}; 
    \draw[zig] (2,0,0) -- (2,0,1); \whitedot{2}{0}{0}; \whitedot{2}{0}{1};
    \draw[double] (3,1,0) -- (3,1,1); \blackdot{3}{1}{0}; \blackdot{3}{1}{1}; 
    \draw[zig] (3,0,0) -- (3,0,1); \whitedot{3}{0}{0}; \whitedot{3}{0}{1};

    \foreach \z in {0,1}{
        \draw[Latex-Latex,shorten >= 5pt, shorten <= 5pt,very thin] (0,0,\z) to (3,1,\z);
    }

    \node at (3.6,0.5,0) {$\xLongrightarrow{\tau_3}$};

    \draw[dashed] (4.6,1,0) -- (4.6,0,0) (4.6,1,1) -- (4.6,0,1) (5.6,1,0) -- (5.6,0,0) (5.6,1,1) -- (5.6,0,1)
        (6.6,1,0) -- (6.6,0,0) (6.6,1,1) -- (6.6,0,1) (7.6,1,0) -- (7.6,0,0) (7.6,1,1) -- (7.6,0,1) 
        (4.6,1,0) -- (5.6,1,0) -- (6.6,1,0) -- (7.6,1,0) (4.6,1,1) -- (5.6,1,1) -- (6.6,1,1) -- (7.6,1,1) 
        (4.6,0,0) -- (5.6,0,0) -- (6.6,0,0) -- (7.6,0,0) (4.6,0,1) -- (5.6,0,1) -- (6.6,0,1) -- (7.6,0,1);
    \draw[double] (4.6,1,0) -- (4.6,1,1); \whitedot{4.6}{1}{0}; \whitedot{4.6}{1}{1}; 
    \draw[zig] (4.6,0,0) -- (4.6,0,1); \whitedot{4.6}{0}{0}; \whitedot{4.6}{0}{1};
    \draw[double] (5.6,1,0) -- (5.6,1,1); \whitedot{5.6}{1}{0}; \whitedot{5.6}{1}{1}; 
    \draw[zig] (5.6,0,0) -- (5.6,0,1); \whitedot{5.6}{0}{0}; \whitedot{5.6}{0}{1};
    \draw[double] (6.6,1,0) -- (6.6,1,1); \whitedot{6.6}{1}{0}; \whitedot{6.6}{1}{1}; 
    \draw[zig] (6.6,0,0) -- (6.6,0,1); \whitedot{6.6}{0}{0}; \whitedot{6.6}{0}{1};
    \draw[double] (7.6,1,0) -- (7.6,1,1); \whitedot{7.6}{1}{0}; \whitedot{7.6}{1}{1}; 
    \draw[zig] (7.6,0,0) -- (7.6,0,1); \whitedot{7.6}{0}{0}; \whitedot{7.6}{0}{1};
}}

                \end{center}
            \item $\alpha=1,\beta=0,\gamma\geq 2$ : This graph shows how to tackle $1$ A-card with $2$ C-cards.
                \begin{center}
                    \scalebox{0.8}{\tikz{\Large
    \draw[dashed] (0,1,0) -- (0,0,0) (0,1,1) -- (0,0,1) (1,1,0) -- (1,0,0) (1,1,1) -- (1,0,1)
        (2,1,0) -- (2,0,0) (2,1,1) -- (2,0,1)
        (0,1,0) -- (1,1,0) -- (2,1,0) (0,1,1) -- (1,1,1) -- (2,1,1)
        (0,0,0) -- (1,0,0) -- (2,0,0) (0,0,1) -- (1,0,1) -- (2,0,1);
    \draw[double] (0,1,0) -- (0,1,1); \whitedot{0}{1}{0}; \blackdot{0}{1}{1}; 
    \draw[zig] (0,0,0) -- (0,0,1); \blackdot{0}{0}{0}; \whitedot{0}{0}{1};
    \draw[double] (1,1,0) -- (1,1,1); \whitedot{1}{1}{0}; \whitedot{1}{1}{1}; 
    \draw[zig] (1,0,0) -- (1,0,1); \whitedot{1}{0}{0}; \whitedot{1}{0}{1};
    \draw[double] (2,1,0) -- (2,1,1); \whitedot{2}{1}{0}; \whitedot{2}{1}{1}; 
    \draw[zig] (2,0,0) -- (2,0,1); \whitedot{2}{0}{0}; \whitedot{2}{0}{1};

    \foreach \z in {0,1}{
        \draw[Latex-Latex,shorten >= 5pt, shorten <= 5pt,very thin] (0,0,\z) to (2,1,\z);
        \draw[Latex-Latex,shorten >= 5pt, shorten <= 5pt,very thin] (1,1,\z) to (2,0,\z);
    }

    \node at (2.6,0.5,0) {$\xLongrightarrow{\tau_1}$};

    \draw[dashed] (3.7,1,0) -- (3.7,0,0) (3.7,1,1) -- (3.7,0,1) (4.7,1,0) -- (4.7,0,0) (4.7,1,1) -- (4.7,0,1)
        (5.7,1,0) -- (5.7,0,0) (5.7,1,1) -- (5.7,0,1) 
        (3.7,1,0) -- (4.7,1,0) -- (5.7,1,0) (3.7,1,1) -- (4.7,1,1) -- (5.7,1,1)
        (3.7,0,0) -- (4.7,0,0) -- (5.7,0,0) (3.7,0,1) -- (4.7,0,1) -- (5.7,0,1);
    \draw[double] (3.7,1,0) -- (3.7,1,1); \whitedot{3.7}{1}{0}; \blackdot{3.7}{1}{1}; 
    \draw[zig] (3.7,0,0) -- (3.7,0,1); \blackdot{3.7}{0}{0}; \blackdot{3.7}{0}{1};
    \draw[double] (4.7,1,0) -- (4.7,1,1); \blackdot{4.7}{1}{0}; \blackdot{4.7}{1}{1}; 
    \draw[zig] (4.7,0,0) -- (4.7,0,1); \whitedot{4.7}{0}{0}; \whitedot{4.7}{0}{1};
    \draw[double] (5.7,1,0) -- (5.7,1,1); \whitedot{5.7}{1}{0}; \blackdot{5.7}{1}{1}; 
    \draw[zig] (5.7,0,0) -- (5.7,0,1); \blackdot{5.7}{0}{0}; \blackdot{5.7}{0}{1};

    \foreach \y in {0,1}{
        \draw[Latex-Latex,shorten >= 5pt, shorten <= 5pt,very thin] (3.7,\y,0) to (5.7,\y,1);
    }
}}
\vspace{5pt}\\
\scalebox{0.8}{\tikz{\Large
    \node at (-1.1,0.5,0) {$\xLongrightarrow{\tau_2}$};

    \draw[dashed] (0,1,0) -- (0,0,0) (0,1,1) -- (0,0,1) (1,1,0) -- (1,0,0) (1,1,1) -- (1,0,1)
        (2,1,0) -- (2,0,0) (2,1,1) -- (2,0,1) 
        (0,1,0) -- (1,1,0) -- (2,1,0) (0,1,1) -- (1,1,1) -- (2,1,1)
        (0,0,0) -- (1,0,0) -- (2,0,0) (0,0,1) -- (1,0,1) -- (2,0,1);
    \draw[double] (0,1,0) -- (0,1,1); \blackdot{0}{1}{0}; \blackdot{0}{1}{1}; 
    \draw[zig] (0,0,0) -- (0,0,1); \blackdot{0}{0}{0}; \blackdot{0}{0}{1};
    \draw[double] (1,1,0) -- (1,1,1); \blackdot{1}{1}{0}; \blackdot{1}{1}{1}; 
    \draw[zig] (1,0,0) -- (1,0,1); \whitedot{1}{0}{0}; \whitedot{1}{0}{1};
    \draw[double] (2,1,0) -- (2,1,1); \whitedot{2}{1}{0}; \whitedot{2}{1}{1}; 
    \draw[zig] (2,0,0) -- (2,0,1); \blackdot{2}{0}{0}; \blackdot{2}{0}{1};

    \foreach \z in {0,1}{
        \draw[Latex-Latex,shorten >= 5pt, shorten <= 5pt,very thin,yshift=-0.3mm] (0,0,\z) to [bend left] (0,1,\z);
        \draw[Latex-Latex,shorten >= 5pt, shorten <= 5pt,very thin] (1,1,\z) to (2,0,\z);
    }

    \node at (2.6,0.5,0) {$\xLongrightarrow{\tau_3}$};

    \draw[dashed] (3.7,1,0) -- (3.7,0,0) (3.7,1,1) -- (3.7,0,1) (4.7,1,0) -- (4.7,0,0) (4.7,1,1) -- (4.7,0,1)
        (5.7,1,0) -- (5.7,0,0) (5.7,1,1) -- (5.7,0,1)
        (3.7,1,0) -- (4.7,1,0) -- (5.7,1,0) (3.7,1,1) -- (4.7,1,1) -- (5.7,1,1) 
        (3.7,0,0) -- (4.7,0,0) -- (5.7,0,0) (3.7,0,1) -- (4.7,0,1) -- (5.7,0,1);
    \draw[double] (3.7,1,0) -- (3.7,1,1); \whitedot{3.7}{1}{0}; \whitedot{3.7}{1}{1}; 
    \draw[zig] (3.7,0,0) -- (3.7,0,1); \whitedot{3.7}{0}{0}; \whitedot{3.7}{0}{1};
    \draw[double] (4.7,1,0) -- (4.7,1,1); \whitedot{4.7}{1}{0}; \whitedot{4.7}{1}{1}; 
    \draw[zig] (4.7,0,0) -- (4.7,0,1); \whitedot{4.7}{0}{0}; \whitedot{4.7}{0}{1};
    \draw[double] (5.7,1,0) -- (5.7,1,1); \whitedot{5.7}{1}{0}; \whitedot{5.7}{1}{1}; 
    \draw[zig] (5.7,0,0) -- (5.7,0,1); \whitedot{5.7}{0}{0}; \whitedot{5.7}{0}{1};
}}

                \end{center}
            \item $\alpha=0,\beta=1,\gamma\geq 2$ : This graph shows how to tackle $1$ B-card with $2$ C-card.
                \begin{center}                    
                    \scalebox{0.8}
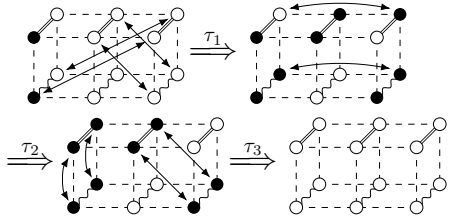
{\tikz{\Large
    \draw[dashed] (0,1,0) -- (0,0,0) (0,1,1) -- (0,0,1) (1,1,0) -- (1,0,0) (1,1,1) -- (1,0,1)
        (2,1,0) -- (2,0,0) (2,1,1) -- (2,0,1)
        (0,1,0) -- (1,1,0) -- (2,1,0) (0,1,1) -- (1,1,1) -- (2,1,1)
        (0,0,0) -- (1,0,0) -- (2,0,0) (0,0,1) -- (1,0,1) -- (2,0,1);
    \draw[double] (0,1,0) -- (0,1,1); \whitedot{0}{1}{0}; \blackdot{0}{1}{1}; 
    \draw[zig] (0,0,0) -- (0,0,1); \whitedot{0}{0}{0}; \blackdot{0}{0}{1};
    \draw[double] (1,1,0) -- (1,1,1); \whitedot{1}{1}{0}; \whitedot{1}{1}{1}; 
    \draw[zig] (1,0,0) -- (1,0,1); \whitedot{1}{0}{0}; \whitedot{1}{0}{1};
    \draw[double] (2,1,0) -- (2,1,1); \whitedot{2}{1}{0}; \whitedot{2}{1}{1}; 
    \draw[zig] (2,0,0) -- (2,0,1); \whitedot{2}{0}{0}; \whitedot{2}{0}{1};

    \foreach \z in {0,1}{
        \draw[Latex-Latex,shorten >= 5pt, shorten <= 5pt,very thin] (0,0,\z) to (2,1,\z);
        \draw[Latex-Latex,shorten >= 5pt, shorten <= 5pt,very thin] (1,1,\z) to (2,0,\z);
    }

    \node at (2.6,0.5,0) {$\xLongrightarrow{\tau_1}$};

    \draw[dashed] (3.7,1,0) -- (3.7,0,0) (3.7,1,1) -- (3.7,0,1) (4.7,1,0) -- (4.7,0,0) (4.7,1,1) -- (4.7,0,1)
        (5.7,1,0) -- (5.7,0,0) (5.7,1,1) -- (5.7,0,1) 
        (3.7,1,0) -- (4.7,1,0) -- (5.7,1,0) (3.7,1,1) -- (4.7,1,1) -- (5.7,1,1)
        (3.7,0,0) -- (4.7,0,0) -- (5.7,0,0) (3.7,0,1) -- (4.7,0,1) -- (5.7,0,1);
    \draw[double] (3.7,1,0) -- (3.7,1,1); \whitedot{3.7}{1}{0}; \blackdot{3.7}{1}{1}; 
    \draw[zig] (3.7,0,0) -- (3.7,0,1); \blackdot{3.7}{0}{0}; \blackdot{3.7}{0}{1};
    \draw[double] (4.7,1,0) -- (4.7,1,1); \blackdot{4.7}{1}{0}; \blackdot{4.7}{1}{1}; 
    \draw[zig] (4.7,0,0) -- (4.7,0,1); \whitedot{4.7}{0}{0}; \whitedot{4.7}{0}{1};
    \draw[double] (5.7,1,0) -- (5.7,1,1); \blackdot{5.7}{1}{0}; \whitedot{5.7}{1}{1}; 
    \draw[zig] (5.7,0,0) -- (5.7,0,1); \blackdot{5.7}{0}{0}; \blackdot{5.7}{0}{1};

    \foreach \y in {0,1}{
        \draw[Latex-Latex,shorten >= 5pt, shorten <= 5pt,very thin,yshift=0.7mm] (3.7,\y,0) to [bend left=10] (5.7,\y,0);
    }
}}
\vspace{5pt}\\
\scalebox{0.8}{\tikz{\Large
    \node at (-1.1,0.5,0) {$\xLongrightarrow{\tau_2}$};

    \draw[dashed] (0,1,0) -- (0,0,0) (0,1,1) -- (0,0,1) (1,1,0) -- (1,0,0) (1,1,1) -- (1,0,1)
        (2,1,0) -- (2,0,0) (2,1,1) -- (2,0,1) 
        (0,1,0) -- (1,1,0) -- (2,1,0) (0,1,1) -- (1,1,1) -- (2,1,1)
        (0,0,0) -- (1,0,0) -- (2,0,0) (0,0,1) -- (1,0,1) -- (2,0,1);
    \draw[double] (0,1,0) -- (0,1,1); \blackdot{0}{1}{0}; \blackdot{0}{1}{1}; 
    \draw[zig] (0,0,0) -- (0,0,1); \blackdot{0}{0}{0}; \blackdot{0}{0}{1};
    \draw[double] (1,1,0) -- (1,1,1); \blackdot{1}{1}{0}; \blackdot{1}{1}{1}; 
    \draw[zig] (1,0,0) -- (1,0,1); \whitedot{1}{0}{0}; \whitedot{1}{0}{1};
    \draw[double] (2,1,0) -- (2,1,1); \whitedot{2}{1}{0}; \whitedot{2}{1}{1}; 
    \draw[zig] (2,0,0) -- (2,0,1); \blackdot{2}{0}{0}; \blackdot{2}{0}{1};

    \foreach \z in {0,1}{
        \draw[Latex-Latex,shorten >= 5pt, shorten <= 5pt,very thin,yshift=-0.3mm] (0,0,\z) to [bend left] (0,1,\z);
        \draw[Latex-Latex,shorten >= 5pt, shorten <= 5pt,very thin] (1,1,\z) to (2,0,\z);
    }

    \node at (2.6,0.5,0) {$\xLongrightarrow{\tau_3}$};

    \draw[dashed] (3.7,1,0) -- (3.7,0,0) (3.7,1,1) -- (3.7,0,1) (4.7,1,0) -- (4.7,0,0) (4.7,1,1) -- (4.7,0,1)
        (5.7,1,0) -- (5.7,0,0) (5.7,1,1) -- (5.7,0,1)
        (3.7,1,0) -- (4.7,1,0) -- (5.7,1,0) (3.7,1,1) -- (4.7,1,1) -- (5.7,1,1) 
        (3.7,0,0) -- (4.7,0,0) -- (5.7,0,0) (3.7,0,1) -- (4.7,0,1) -- (5.7,0,1);
    \draw[double] (3.7,1,0) -- (3.7,1,1); \whitedot{3.7}{1}{0}; \whitedot{3.7}{1}{1}; 
    \draw[zig] (3.7,0,0) -- (3.7,0,1); \whitedot{3.7}{0}{0}; \whitedot{3.7}{0}{1};
    \draw[double] (4.7,1,0) -- (4.7,1,1); \whitedot{4.7}{1}{0}; \whitedot{4.7}{1}{1}; 
    \draw[zig] (4.7,0,0) -- (4.7,0,1); \whitedot{4.7}{0}{0}; \whitedot{4.7}{0}{1};
    \draw[double] (5.7,1,0) -- (5.7,1,1); \whitedot{5.7}{1}{0}; \whitedot{5.7}{1}{1}; 
    \draw[zig] (5.7,0,0) -- (5.7,0,1); \whitedot{5.7}{0}{0}; \whitedot{5.7}{0}{1};
}}

                \end{center}
        \end{itemize}
        
               For the  $2^{\it rd}$ case, we reduce it to the following.
        
        \begin{itemize} 
            \item $\alpha=2,\beta=1,\gamma\geq 1$: This graph shows how to tackle $2$ A-cards and $1$ B-card with $1$ C-card. 
                \begin{center}                    
                    \scalebox{0.8}{\tikz{\Large
    \draw[dashed] (0,1,0) -- (0,0,0) (0,1,1) -- (0,0,1) (1,1,0) -- (1,0,0) (1,1,1) -- (1,0,1)
        (2,1,0) -- (2,0,0) (2,1,1) -- (2,0,1) (3,1,0) -- (3,0,0) (3,1,1) -- (3,0,1) 
        (0,1,0) -- (1,1,0) -- (2,1,0) -- (3,1,0) (0,1,1) -- (1,1,1) -- (2,1,1) -- (3,1,1) 
        (0,0,0) -- (1,0,0) -- (2,0,0) -- (3,0,0) (0,0,1) -- (1,0,1) -- (2,0,1) -- (3,0,1);
    \draw[double] (0,1,0) -- (0,1,1); \whitedot{0}{1}{0}; \blackdot{0}{1}{1}; 
    \draw[zig] (0,0,0) -- (0,0,1); \blackdot{0}{0}{0}; \whitedot{0}{0}{1};
    \draw[double] (1,1,0) -- (1,1,1); \whitedot{1}{1}{0}; \blackdot{1}{1}{1}; 
    \draw[zig] (1,0,0) -- (1,0,1); \blackdot{1}{0}{0}; \whitedot{1}{0}{1};
    \draw[double] (2,1,0) -- (2,1,1); \whitedot{2}{1}{0}; \blackdot{2}{1}{1}; 
    \draw[zig] (2,0,0) -- (2,0,1); \whitedot{2}{0}{0}; \blackdot{2}{0}{1};
    \draw[double] (3,1,0) -- (3,1,1); \whitedot{3}{1}{0}; \whitedot{3}{1}{1}; 
    \draw[zig] (3,0,0) -- (3,0,1); \whitedot{3}{0}{0}; \whitedot{3}{0}{1};

    \foreach \z in {0,1}{
        \draw[Latex-,shorten >= 5pt, shorten <= 5pt,very thin,yshift=0.7mm] (0,1,\z) to [bend left=10] (3,1,\z);
        \draw[Latex-,shorten >= 5pt, shorten <= 5pt,very thin,yshift=0.6mm] (3,1,\z) to [bend left] (3,0,\z);
        \draw[Latex-,shorten >= 5pt, shorten <= 5pt,very thin,yshift=0.7mm] (3,0,\z) to [bend right=10] (2,0,\z);
        \draw[Latex-,shorten >= 5pt, shorten <= 5pt,very thin] (2,0,\z) to (0,1,\z);
    }

    \node at (3.65,0.5,0) {$\xLongrightarrow{\tau_1}$};

    \draw[dashed] (4.6,1,0) -- (4.6,0,0) (4.6,1,1) -- (4.6,0,1) (5.6,1,0) -- (5.6,0,0) (5.6,1,1) -- (5.6,0,1)
        (6.6,1,0) -- (6.6,0,0) (6.6,1,1) -- (6.6,0,1) (7.6,1,0) -- (7.6,0,0) (7.6,1,1) -- (7.6,0,1) 
        (4.6,1,0) -- (5.6,1,0) -- (6.6,1,0) -- (7.6,1,0) (4.6,1,1) -- (5.6,1,1) -- (6.6,1,1) -- (7.6,1,1) 
        (4.6,0,0) -- (5.6,0,0) -- (6.6,0,0) -- (7.6,0,0) (4.6,0,1) -- (5.6,0,1) -- (6.6,0,1) -- (7.6,0,1);
    \draw[double] (4.6,1,0) -- (4.6,1,1); \blackdot{4.6}{1}{0}; \whitedot{4.6}{1}{1}; 
    \draw[zig] (4.6,0,0) -- (4.6,0,1); \blackdot{4.6}{0}{0}; \whitedot{4.6}{0}{1};
    \draw[double] (5.6,1,0) -- (5.6,1,1); \whitedot{5.6}{1}{0}; \blackdot{5.6}{1}{1}; 
    \draw[zig] (5.6,0,0) -- (5.6,0,1); \blackdot{5.6}{0}{0}; \whitedot{5.6}{0}{1};
    \draw[double] (6.6,1,0) -- (6.6,1,1); \whitedot{6.6}{1}{0}; \blackdot{6.6}{1}{1}; 
    \draw[zig] (6.6,0,0) -- (6.6,0,1); \whitedot{6.6}{0}{0}; \whitedot{6.6}{0}{1};
    \draw[double] (7.6,1,0) -- (7.6,1,1); \whitedot{7.6}{1}{0}; \blackdot{7.6}{1}{1}; 
    \draw[zig] (7.6,0,0) -- (7.6,0,1); \blackdot{7.6}{0}{0}; \blackdot{7.6}{0}{1};

    \foreach \y in {0,1}{
        \draw[Latex-Latex,shorten >= 5pt, shorten <= 5pt,very thin,yshift=-0.7mm] (4.6,\y,1) to [bend right=10] (5.6,\y,1);
        \draw[Latex-Latex,shorten >= 5pt, shorten <= 5pt,very thin] (4.6,\y,0) to (6.6,\y,1);
        \draw[Latex-,shorten >= 5pt, shorten <= 5pt,very thin,yshift=0.7mm] (6.6,\y,0) to [bend right=10] (5.6,\y,0);
        \draw[Latex-,shorten >= 5pt, shorten <= 5pt,very thin] (5.6,\y,0) to (7.6,\y,1);
        \draw[Latex-,shorten >= 5pt, shorten <= 5pt,very thin] (7.6,\y,1) to (6.6,\y,0);
    }
}}
\vspace{5pt}\\
\hspace{-5pt}\scalebox{0.8}{\tikz{\Large
    \node at (-1,0.5,0) {$\xLongrightarrow{\tau_2}$};

    \draw[dashed] (0,1,0) -- (0,0,0) (0,1,1) -- (0,0,1) (1,1,0) -- (1,0,0) (1,1,1) -- (1,0,1)
        (2,1,0) -- (2,0,0) (2,1,1) -- (2,0,1) (3,1,0) -- (3,0,0) (3,1,1) -- (3,0,1) 
        (0,1,0) -- (1,1,0) -- (2,1,0) -- (3,1,0) (0,1,1) -- (1,1,1) -- (2,1,1) -- (3,1,1) 
        (0,0,0) -- (1,0,0) -- (2,0,0) -- (3,0,0) (0,0,1) -- (1,0,1) -- (2,0,1) -- (3,0,1);
    \draw[double] (0,1,0) -- (0,1,1); \blackdot{0}{1}{0}; \blackdot{0}{1}{1}; 
    \draw[zig] (0,0,0) -- (0,0,1); \whitedot{0}{0}{0}; \whitedot{0}{0}{1};
    \draw[double] (1,1,0) -- (1,1,1); \whitedot{1}{1}{0}; \whitedot{1}{1}{1}; 
    \draw[zig] (1,0,0) -- (1,0,1); \whitedot{1}{0}{0}; \whitedot{1}{0}{1};
    \draw[double] (2,1,0) -- (2,1,1); \blackdot{2}{1}{0}; \blackdot{2}{1}{1}; 
    \draw[zig] (2,0,0) -- (2,0,1); \blackdot{2}{0}{0}; \blackdot{2}{0}{1};
    \draw[double] (3,1,0) -- (3,1,1); \whitedot{3}{1}{0}; \whitedot{3}{1}{1}; 
    \draw[zig] (3,0,0) -- (3,0,1); \blackdot{3}{0}{0}; \blackdot{3}{0}{1};

    \foreach \z in {0,1}{
        \draw[Latex-Latex,shorten >= 5pt, shorten <= 5pt,very thin] (0,1,\z) to (2,0,\z);
        \draw[Latex-Latex,shorten >= 5pt, shorten <= 5pt,very thin] (2,1,\z) to (3,0,\z);
    }

    \node at (3.6,0.5,0) {$\xLongrightarrow{\tau_3}$};

    \draw[dashed] (4.6,1,0) -- (4.6,0,0) (4.6,1,1) -- (4.6,0,1) (5.6,1,0) -- (5.6,0,0) (5.6,1,1) -- (5.6,0,1)
        (6.6,1,0) -- (6.6,0,0) (6.6,1,1) -- (6.6,0,1) (7.6,1,0) -- (7.6,0,0) (7.6,1,1) -- (7.6,0,1) 
        (4.6,1,0) -- (5.6,1,0) -- (6.6,1,0) -- (7.6,1,0) (4.6,1,1) -- (5.6,1,1) -- (6.6,1,1) -- (7.6,1,1) 
        (4.6,0,0) -- (5.6,0,0) -- (6.6,0,0) -- (7.6,0,0) (4.6,0,1) -- (5.6,0,1) -- (6.6,0,1) -- (7.6,0,1);
    \draw[double] (4.6,1,0) -- (4.6,1,1); \whitedot{4.6}{1}{0}; \whitedot{4.6}{1}{1}; 
    \draw[zig] (4.6,0,0) -- (4.6,0,1); \whitedot{4.6}{0}{0}; \whitedot{4.6}{0}{1};
    \draw[double] (5.6,1,0) -- (5.6,1,1); \whitedot{5.6}{1}{0}; \whitedot{5.6}{1}{1}; 
    \draw[zig] (5.6,0,0) -- (5.6,0,1); \whitedot{5.6}{0}{0}; \whitedot{5.6}{0}{1};
    \draw[double] (6.6,1,0) -- (6.6,1,1); \whitedot{6.6}{1}{0}; \whitedot{6.6}{1}{1}; 
    \draw[zig] (6.6,0,0) -- (6.6,0,1); \whitedot{6.6}{0}{0}; \whitedot{6.6}{0}{1};
    \draw[double] (7.6,1,0) -- (7.6,1,1); \whitedot{7.6}{1}{0}; \whitedot{7.6}{1}{1}; 
    \draw[zig] (7.6,0,0) -- (7.6,0,1); \whitedot{7.6}{0}{0}; \whitedot{7.6}{0}{1};
}} 
                \end{center}
            \item $\alpha=0,\beta=3,\gamma\geq 1$ : This graph shows how to tackle $3$ B-cards with $1$ C-card. 
                \begin{center}
                    \scalebox{0.8}
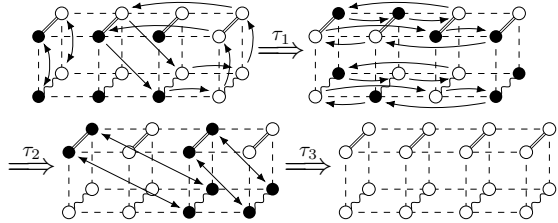
{\tikz{\Large
    \draw[dashed] (0,1,0) -- (0,0,0) (0,1,1) -- (0,0,1) (1,1,0) -- (1,0,0) (1,1,1) -- (1,0,1)
        (2,1,0) -- (2,0,0) (2,1,1) -- (2,0,1) (3,1,0) -- (3,0,0) (3,1,1) -- (3,0,1) 
        (0,1,0) -- (1,1,0) -- (2,1,0) -- (3,1,0) (0,1,1) -- (1,1,1) -- (2,1,1) -- (3,1,1) 
        (0,0,0) -- (1,0,0) -- (2,0,0) -- (3,0,0) (0,0,1) -- (1,0,1) -- (2,0,1) -- (3,0,1);
    \draw[double] (0,1,0) -- (0,1,1); \whitedot{0}{1}{0}; \blackdot{0}{1}{1}; 
    \draw[zig] (0,0,0) -- (0,0,1); \whitedot{0}{0}{0}; \blackdot{0}{0}{1};
    \draw[double] (1,1,0) -- (1,1,1); \whitedot{1}{1}{0}; \blackdot{1}{1}{1}; 
    \draw[zig] (1,0,0) -- (1,0,1); \whitedot{1}{0}{0}; \blackdot{1}{0}{1};
    \draw[double] (2,1,0) -- (2,1,1); \whitedot{2}{1}{0}; \blackdot{2}{1}{1}; 
    \draw[zig] (2,0,0) -- (2,0,1); \whitedot{2}{0}{0}; \blackdot{2}{0}{1};
    \draw[double] (3,1,0) -- (3,1,1); \whitedot{3}{1}{0}; \whitedot{3}{1}{1}; 
    \draw[zig] (3,0,0) -- (3,0,1); \whitedot{3}{0}{0}; \whitedot{3}{0}{1};

    \foreach \z in {0,1}{
        \draw[Latex-Latex,shorten >= 5pt, shorten <= 5pt,very thin,yshift=0.3mm] (0,1,\z) to [bend left] (0,0,\z);
        \draw[Latex-,shorten >= 5pt, shorten <= 5pt,very thin,yshift=0.7mm] (1,1,\z) to [bend left=10] (3,1,\z);
        \draw[Latex-,shorten >= 5pt, shorten <= 5pt,very thin,yshift=0.3mm] (3,1,\z) to [bend left] (3,0,\z);
        \draw[Latex-,shorten >= 5pt, shorten <= 5pt,very thin,yshift=0.7mm] (3,0,\z) to [bend right=10] (2,0,\z);
        \draw[Latex-,shorten >= 5pt, shorten <= 5pt,very thin] (2,0,\z) to (1,1,\z);
    }

    \node at (3.65,0.5,0) {$\xLongrightarrow{\tau_1}$};

    \draw[dashed] (4.6,1,0) -- (4.6,0,0) (4.6,1,1) -- (4.6,0,1) (5.6,1,0) -- (5.6,0,0) (5.6,1,1) -- (5.6,0,1)
        (6.6,1,0) -- (6.6,0,0) (6.6,1,1) -- (6.6,0,1) (7.6,1,0) -- (7.6,0,0) (7.6,1,1) -- (7.6,0,1) 
        (4.6,1,0) -- (5.6,1,0) -- (6.6,1,0) -- (7.6,1,0) (4.6,1,1) -- (5.6,1,1) -- (6.6,1,1) -- (7.6,1,1) 
        (4.6,0,0) -- (5.6,0,0) -- (6.6,0,0) -- (7.6,0,0) (4.6,0,1) -- (5.6,0,1) -- (6.6,0,1) -- (7.6,0,1);
    \draw[double] (4.6,1,0) -- (4.6,1,1); \blackdot{4.6}{1}{0}; \whitedot{4.6}{1}{1}; 
    \draw[zig] (4.6,0,0) -- (4.6,0,1); \blackdot{4.6}{0}{0}; \whitedot{4.6}{0}{1};
    \draw[double] (5.6,1,0) -- (5.6,1,1); \blackdot{5.6}{1}{0}; \whitedot{5.6}{1}{1}; 
    \draw[zig] (5.6,0,0) -- (5.6,0,1); \whitedot{5.6}{0}{0}; \blackdot{5.6}{0}{1};
    \draw[double] (6.6,1,0) -- (6.6,1,1); \whitedot{6.6}{1}{0}; \blackdot{6.6}{1}{1}; 
    \draw[zig] (6.6,0,0) -- (6.6,0,1); \whitedot{6.6}{0}{0}; \whitedot{6.6}{0}{1};
    \draw[double] (7.6,1,0) -- (7.6,1,1); \whitedot{7.6}{1}{0}; \blackdot{7.6}{1}{1}; 
    \draw[zig] (7.6,0,0) -- (7.6,0,1); \blackdot{7.6}{0}{0}; \blackdot{7.6}{0}{1};

    \foreach \y in {0,1}{
        \draw[Latex-,shorten >= 5pt, shorten <= 5pt,very thin,yshift=0.7mm] (4.6,\y,0) to [bend left=10] (6.6,\y,0);
        \draw[Latex-,shorten >= 5pt, shorten <= 5pt,very thin,yshift=-0.6mm] (6.6,\y,0) to [bend left=10] (5.6,\y,0);
        \draw[Latex-,shorten >= 5pt, shorten <= 5pt,very thin,yshift=-0.6mm] (5.6,\y,0) to [bend left=10] (4.6,\y,0);
        \draw[Latex-,shorten >= 5pt, shorten <= 5pt,very thin,yshift=0.7mm] (7.6,\y,1) to [bend right=10] (6.6,\y,1);
        \draw[Latex-,shorten >= 5pt, shorten <= 5pt,very thin,yshift=0.7mm] (6.6,\y,1) to [bend right=10] (4.6,\y,1);
        \draw[Latex-,shorten >= 5pt, shorten <= 5pt,very thin,yshift=-0.7mm] (4.6,\y,1) to [bend right=10] (5.6,\y,1);
        \draw[Latex-,shorten >= 5pt, shorten <= 5pt,very thin,yshift=-0.7mm] (5.6,\y,1) to [bend right=10] (7.6,\y,1);
    }
}}
\vspace{5pt}\\
\hspace{-5pt}\scalebox{0.8}{\tikz{\Large
    \node at (-1,0.5,0) {$\xLongrightarrow{\tau_2}$};

    \draw[dashed] (0,1,0) -- (0,0,0) (0,1,1) -- (0,0,1) (1,1,0) -- (1,0,0) (1,1,1) -- (1,0,1)
        (2,1,0) -- (2,0,0) (2,1,1) -- (2,0,1) (3,1,0) -- (3,0,0) (3,1,1) -- (3,0,1) 
        (0,1,0) -- (1,1,0) -- (2,1,0) -- (3,1,0) (0,1,1) -- (1,1,1) -- (2,1,1) -- (3,1,1) 
        (0,0,0) -- (1,0,0) -- (2,0,0) -- (3,0,0) (0,0,1) -- (1,0,1) -- (2,0,1) -- (3,0,1);
    \draw[double] (0,1,0) -- (0,1,1); \blackdot{0}{1}{0}; \blackdot{0}{1}{1}; 
    \draw[zig] (0,0,0) -- (0,0,1); \whitedot{0}{0}{0}; \whitedot{0}{0}{1};
    \draw[double] (1,1,0) -- (1,1,1); \whitedot{1}{1}{0}; \whitedot{1}{1}{1}; 
    \draw[zig] (1,0,0) -- (1,0,1); \whitedot{1}{0}{0}; \whitedot{1}{0}{1};
    \draw[double] (2,1,0) -- (2,1,1); \blackdot{2}{1}{0}; \blackdot{2}{1}{1}; 
    \draw[zig] (2,0,0) -- (2,0,1); \blackdot{2}{0}{0}; \blackdot{2}{0}{1};
    \draw[double] (3,1,0) -- (3,1,1); \whitedot{3}{1}{0}; \whitedot{3}{1}{1}; 
    \draw[zig] (3,0,0) -- (3,0,1); \blackdot{3}{0}{0}; \blackdot{3}{0}{1};

    \foreach \z in {0,1}{
        \draw[Latex-Latex,shorten >= 5pt, shorten <= 5pt,very thin] (0,1,\z) to (2,0,\z);
        \draw[Latex-Latex,shorten >= 5pt, shorten <= 5pt,very thin] (2,1,\z) to (3,0,\z);
    }

    \node at (3.6,0.5,0) {$\xLongrightarrow{\tau_3}$};

    \draw[dashed] (4.6,1,0) -- (4.6,0,0) (4.6,1,1) -- (4.6,0,1) (5.6,1,0) -- (5.6,0,0) (5.6,1,1) -- (5.6,0,1)
        (6.6,1,0) -- (6.6,0,0) (6.6,1,1) -- (6.6,0,1) (7.6,1,0) -- (7.6,0,0) (7.6,1,1) -- (7.6,0,1) 
        (4.6,1,0) -- (5.6,1,0) -- (6.6,1,0) -- (7.6,1,0) (4.6,1,1) -- (5.6,1,1) -- (6.6,1,1) -- (7.6,1,1) 
        (4.6,0,0) -- (5.6,0,0) -- (6.6,0,0) -- (7.6,0,0) (4.6,0,1) -- (5.6,0,1) -- (6.6,0,1) -- (7.6,0,1);
    \draw[double] (4.6,1,0) -- (4.6,1,1); \whitedot{4.6}{1}{0}; \whitedot{4.6}{1}{1}; 
    \draw[zig] (4.6,0,0) -- (4.6,0,1); \whitedot{4.6}{0}{0}; \whitedot{4.6}{0}{1};
    \draw[double] (5.6,1,0) -- (5.6,1,1); \whitedot{5.6}{1}{0}; \whitedot{5.6}{1}{1}; 
    \draw[zig] (5.6,0,0) -- (5.6,0,1); \whitedot{5.6}{0}{0}; \whitedot{5.6}{0}{1};
    \draw[double] (6.6,1,0) -- (6.6,1,1); \whitedot{6.6}{1}{0}; \whitedot{6.6}{1}{1}; 
    \draw[zig] (6.6,0,0) -- (6.6,0,1); \whitedot{6.6}{0}{0}; \whitedot{6.6}{0}{1};
    \draw[double] (7.6,1,0) -- (7.6,1,1); \whitedot{7.6}{1}{0}; \whitedot{7.6}{1}{1}; 
    \draw[zig] (7.6,0,0) -- (7.6,0,1); \whitedot{7.6}{0}{0}; \whitedot{7.6}{0}{1};
}} 
                \end{center}
            \item $\alpha=1,\beta=2,\gamma\geq 1$ : This graph shows how to tackle $1$ A-card and $2$ B-cards with $1$ C-card. 
                \begin{center}
                    \scalebox{0.8}
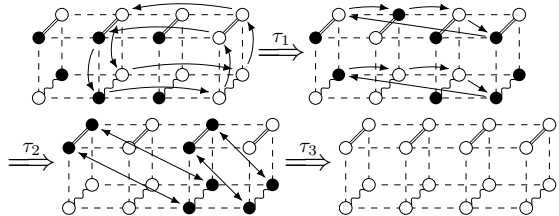
{\tikz{\Large
    \draw[dashed] (0,1,0) -- (0,0,0) (0,1,1) -- (0,0,1) (1,1,0) -- (1,0,0) (1,1,1) -- (1,0,1)
        (2,1,0) -- (2,0,0) (2,1,1) -- (2,0,1) (3,1,0) -- (3,0,0) (3,1,1) -- (3,0,1) 
        (0,1,0) -- (1,1,0) -- (2,1,0) -- (3,1,0) (0,1,1) -- (1,1,1) -- (2,1,1) -- (3,1,1) 
        (0,0,0) -- (1,0,0) -- (2,0,0) -- (3,0,0) (0,0,1) -- (1,0,1) -- (2,0,1) -- (3,0,1);
    \draw[double] (0,1,0) -- (0,1,1); \whitedot{0}{1}{0}; \blackdot{0}{1}{1}; 
    \draw[zig] (0,0,0) -- (0,0,1); \blackdot{0}{0}{0}; \whitedot{0}{0}{1};
    \draw[double] (1,1,0) -- (1,1,1); \whitedot{1}{1}{0}; \blackdot{1}{1}{1}; 
    \draw[zig] (1,0,0) -- (1,0,1); \whitedot{1}{0}{0}; \blackdot{1}{0}{1};
    \draw[double] (2,1,0) -- (2,1,1); \whitedot{2}{1}{0}; \blackdot{2}{1}{1}; 
    \draw[zig] (2,0,0) -- (2,0,1); \whitedot{2}{0}{0}; \blackdot{2}{0}{1};
    \draw[double] (3,1,0) -- (3,1,1); \whitedot{3}{1}{0}; \whitedot{3}{1}{1}; 
    \draw[zig] (3,0,0) -- (3,0,1); \whitedot{3}{0}{0}; \whitedot{3}{0}{1};

    \foreach \z in {0,1}{
        \draw[-Latex,shorten >= 5pt, shorten <= 5pt,very thin,yshift=-0.3mm] (1,1,\z) to [bend right] (1,0,\z);
        \draw[-Latex,shorten >= 5pt, shorten <= 5pt,very thin,yshift=0.7mm] (1,0,\z) to [bend left=10] (3,0,\z);
        \draw[-Latex,shorten >= 5pt, shorten <= 5pt,very thin,yshift=0.3mm] (3,0,\z) to [bend right] (3,1,\z);
        \draw[-Latex,shorten >= 5pt, shorten <= 5pt,very thin,yshift=0.7mm] (3,1,\z) to [bend right=10] (1,1,\z);
    }

    \node at (3.65,0.5,0) {$\xLongrightarrow{\tau_1}$};

    \draw[dashed] (4.6,1,0) -- (4.6,0,0) (4.6,1,1) -- (4.6,0,1) (5.6,1,0) -- (5.6,0,0) (5.6,1,1) -- (5.6,0,1)
        (6.6,1,0) -- (6.6,0,0) (6.6,1,1) -- (6.6,0,1) (7.6,1,0) -- (7.6,0,0) (7.6,1,1) -- (7.6,0,1) 
        (4.6,1,0) -- (5.6,1,0) -- (6.6,1,0) -- (7.6,1,0) (4.6,1,1) -- (5.6,1,1) -- (6.6,1,1) -- (7.6,1,1) 
        (4.6,0,0) -- (5.6,0,0) -- (6.6,0,0) -- (7.6,0,0) (4.6,0,1) -- (5.6,0,1) -- (6.6,0,1) -- (7.6,0,1);
    \draw[double] (4.6,1,0) -- (4.6,1,1); \whitedot{4.6}{1}{0}; \blackdot{4.6}{1}{1}; 
    \draw[zig] (4.6,0,0) -- (4.6,0,1); \blackdot{4.6}{0}{0}; \whitedot{4.6}{0}{1};
    \draw[double] (5.6,1,0) -- (5.6,1,1); \blackdot{5.6}{1}{0}; \whitedot{5.6}{1}{1}; 
    \draw[zig] (5.6,0,0) -- (5.6,0,1); \whitedot{5.6}{0}{0}; \whitedot{5.6}{0}{1};
    \draw[double] (6.6,1,0) -- (6.6,1,1); \whitedot{6.6}{1}{0}; \blackdot{6.6}{1}{1}; 
    \draw[zig] (6.6,0,0) -- (6.6,0,1); \whitedot{6.6}{0}{0}; \blackdot{6.6}{0}{1};
    \draw[double] (7.6,1,0) -- (7.6,1,1); \whitedot{7.6}{1}{0}; \blackdot{7.6}{1}{1}; 
    \draw[zig] (7.6,0,0) -- (7.6,0,1); \blackdot{7.6}{0}{0}; \blackdot{7.6}{0}{1};

    \foreach \y in {0,1}{
        \draw[Latex-,shorten >= 5pt, shorten <= 5pt,very thin] (4.6,\y,0) to (7.6,\y,1);
        \draw[Latex-,shorten >= 5pt, shorten <= 5pt,very thin] (7.6,\y,1) to (6.6,\y,0);
        \draw[Latex-,shorten >= 5pt, shorten <= 5pt,very thin,yshift=0.7mm] (6.6,\y,0) to [bend right=10] (5.6,\y,0);
        \draw[Latex-,shorten >= 5pt, shorten <= 5pt,very thin,yshift=0.7mm] (5.6,\y,0) to [bend right=10] (4.6,\y,0);
    }
}}
\vspace{5pt}\\
\hspace{-5pt}\scalebox{0.8}{\tikz{\Large
    \node at (-1,0.5,0) {$\xLongrightarrow{\tau_2}$};

    \draw[dashed] (0,1,0) -- (0,0,0) (0,1,1) -- (0,0,1) (1,1,0) -- (1,0,0) (1,1,1) -- (1,0,1)
        (2,1,0) -- (2,0,0) (2,1,1) -- (2,0,1) (3,1,0) -- (3,0,0) (3,1,1) -- (3,0,1) 
        (0,1,0) -- (1,1,0) -- (2,1,0) -- (3,1,0) (0,1,1) -- (1,1,1) -- (2,1,1) -- (3,1,1) 
        (0,0,0) -- (1,0,0) -- (2,0,0) -- (3,0,0) (0,0,1) -- (1,0,1) -- (2,0,1) -- (3,0,1);
    \draw[double] (0,1,0) -- (0,1,1); \blackdot{0}{1}{0}; \blackdot{0}{1}{1}; 
    \draw[zig] (0,0,0) -- (0,0,1); \whitedot{0}{0}{0}; \whitedot{0}{0}{1};
    \draw[double] (1,1,0) -- (1,1,1); \whitedot{1}{1}{0}; \whitedot{1}{1}{1}; 
    \draw[zig] (1,0,0) -- (1,0,1); \whitedot{1}{0}{0}; \whitedot{1}{0}{1};
    \draw[double] (2,1,0) -- (2,1,1); \blackdot{2}{1}{0}; \blackdot{2}{1}{1}; 
    \draw[zig] (2,0,0) -- (2,0,1); \blackdot{2}{0}{0}; \blackdot{2}{0}{1};
    \draw[double] (3,1,0) -- (3,1,1); \whitedot{3}{1}{0}; \whitedot{3}{1}{1}; 
    \draw[zig] (3,0,0) -- (3,0,1); \blackdot{3}{0}{0}; \blackdot{3}{0}{1};

    \foreach \z in {0,1}{
        \draw[Latex-Latex,shorten >= 5pt, shorten <= 5pt,very thin] (0,1,\z) to (2,0,\z);
        \draw[Latex-Latex,shorten >= 5pt, shorten <= 5pt,very thin] (2,1,\z) to (3,0,\z);
    }

    \node at (3.6,0.5,0) {$\xLongrightarrow{\tau_3}$};

    \draw[dashed] (4.6,1,0) -- (4.6,0,0) (4.6,1,1) -- (4.6,0,1) (5.6,1,0) -- (5.6,0,0) (5.6,1,1) -- (5.6,0,1)
        (6.6,1,0) -- (6.6,0,0) (6.6,1,1) -- (6.6,0,1) (7.6,1,0) -- (7.6,0,0) (7.6,1,1) -- (7.6,0,1) 
        (4.6,1,0) -- (5.6,1,0) -- (6.6,1,0) -- (7.6,1,0) (4.6,1,1) -- (5.6,1,1) -- (6.6,1,1) -- (7.6,1,1) 
        (4.6,0,0) -- (5.6,0,0) -- (6.6,0,0) -- (7.6,0,0) (4.6,0,1) -- (5.6,0,1) -- (6.6,0,1) -- (7.6,0,1);
    \draw[double] (4.6,1,0) -- (4.6,1,1); \whitedot{4.6}{1}{0}; \whitedot{4.6}{1}{1}; 
    \draw[zig] (4.6,0,0) -- (4.6,0,1); \whitedot{4.6}{0}{0}; \whitedot{4.6}{0}{1};
    \draw[double] (5.6,1,0) -- (5.6,1,1); \whitedot{5.6}{1}{0}; \whitedot{5.6}{1}{1}; 
    \draw[zig] (5.6,0,0) -- (5.6,0,1); \whitedot{5.6}{0}{0}; \whitedot{5.6}{0}{1};
    \draw[double] (6.6,1,0) -- (6.6,1,1); \whitedot{6.6}{1}{0}; \whitedot{6.6}{1}{1}; 
    \draw[zig] (6.6,0,0) -- (6.6,0,1); \whitedot{6.6}{0}{0}; \whitedot{6.6}{0}{1};
    \draw[double] (7.6,1,0) -- (7.6,1,1); \whitedot{7.6}{1}{0}; \whitedot{7.6}{1}{1}; 
    \draw[zig] (7.6,0,0) -- (7.6,0,1); \whitedot{7.6}{0}{0}; \whitedot{7.6}{0}{1};
}} 
                \end{center}
        \end{itemize}
    \end{proof}

    For the other cases, which can not be solved by \lem{goodcase}, can in turn be dealt with \lem{badcase}. 
    
    \begin{lemma}\label{lem:badcase}
        For any $r_3\in [n]\backslash\{r_1,r_2\}$, there exist $\sigma_1\in SC_{\zoton}^{(r_1)}$ and $\pi_1,\pi_2\in SC_{\zoton}^{(r_3)}$ such that $\sigma\pi_1\sigma_1\pi_2\in A_{\zoton}^{(r_1)}$ if
        \begin{enumerate}
            \item $a_3+a_4+b_3+b_4=2$ and $\min\{b_1+a_2,a_1+b_2\}=0$ hold or
            \item $a_3+a_4+b_3+b_4=0$ and $b_1+a_2$ is odd (equivalently $a_1+b_2$ is odd) hold.
        \end{enumerate}
    \end{lemma}

    Fixing $r_1$, if for some $r_2$, the corresponding canonical form falls into \lem{badcase}. Then for any $r_3\in [n]\backslash\{r_1,r_2\}$, the canonical form corresponding with $r_1,r_3$ will fall into 3-step solvable cases, that is, it can be solved by \lem{goodcase} with $r_2'=r_3$. 

    Before the proof, we show how to switch dimensions. We visualize the permutation on a black-white $4$-d cuboid as two $3$-d cuboids. When $r_1,r_2$ are fixed, pick $r_3\in[n]\backslash\{r_1,r_2\}$ and compress all the other $(n-3)$ dimensions. As before, paint $\bm x$ black if $\sigma(\bm x)_{r_1}\neq \bm x_{r_1}$ for all $\bm x \in\{0,1\}^n$. An example of $n=4,r_1=1,r_2=2,r_3=4$ is \fig{4dcube1}. The left and right $3$-d cuboids corresponding to $r_3=0$ and $r_3=1$.
    
        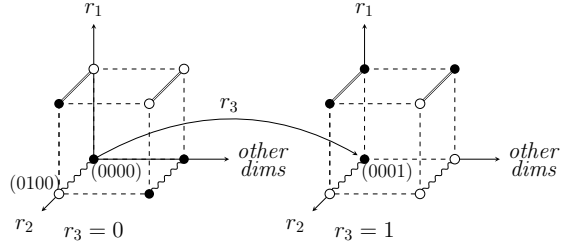
\begin{figure}[ht]
            \centering
            \scalebox{0.6}{\begin{tikzpicture}[>=stealth]
    \Large
    \draw[->] (0,0,0) -- (3,0,0) node [right] {\huge$\substack{\it other\\ \it dims}$};
    \draw[->] (0,0,0) -- (0,3,0) node [above] {$r_1$};
    \draw[->] (0,0,2) -- (0,0,3) node [below left] {$r_2$};
    
    {
    \large
    \node at (0.5,-0.3) {$(0000)$};
    \node at (6.5,-0.3) {$(0001)$};
    \node at (-1.3,-0.55) {$(0100)$};
    }
    \node at (0,-1.6) {$r_3=0$};
    \node at (6,-1.6) {$r_3=1$};
    
    \draw[->] (8,0,0) -- (9,0,0) node [right] {\huge$\substack{\it other\\ \it dims}$};
    \draw[->] (6,2,0) -- (6,3,0) node [above] {$r_1$};
    \draw[->] (6,0,2) -- (6,0,3) node [below left] {$r_2$};

    \foreach \x in {0,6} {
        \draw[dashed] (\x,2,2) -- (\x,0,2) -- (\x+2,0,2) 
                      (\x,0,0) -- (\x+2,0,0)
                      (\x+2,2,0) -- (\x,2,0) (\x,2,2) -- (\x+2,2,2);
    }
                  
    \foreach \x in {0,2,6,8} {
        \draw[double] (\x,2,0) -- (\x,2,2);
        \draw[dashed] (\x,0,2) -- (\x,2,2) (\x,0,0) -- (\x,2,0);
    }
    \foreach \x in {0,2,6,8} {
        \draw[zig] (\x,0,0) -- (\x,0,2);
    }

    \blackdot{0}{0}{0}; \blackdot{2}{0}{0}; \blackdot{6}{0}{0}; \whitedot{8}{0}{0};
    \whitedot{0}{0}{2}; \blackdot{2}{0}{2}; \whitedot{6}{0}{2}; \whitedot{8}{0}{2};
    \blackdot{0}{2}{2}; \whitedot{2}{2}{2}; \blackdot{6}{2}{2}; \whitedot{8}{2}{2};
    \whitedot{0}{2}{0}; \whitedot{2}{2}{0}; \blackdot{6}{2}{0}; \blackdot{8}{2}{0};
    \draw[->,,shorten >= 5pt] (0,0) to [bend left] node [above] {$r_3$} (6,0);
\end{tikzpicture}}
            \caption{$4$-d cuboid for $n=4,r_1=1,r_2=2,r_3=4$.}\label{fig:4dcube1}
        \end{figure}
        
    In \fig{4dcube1}, let $a_1,a_2,a_3,a_4$ be the number of \label{redef} \tikz[thick]{\draw[double] (0,1,0) -- (0,1,1); \whitedot{0}{1}{0}; \blackdot{0}{1}{1};},
\tikz[thick]{\draw[double] (0,1,0) -- (0,1,1); \blackdot{0}{1}{0}; \whitedot{0}{1}{1};},
\tikz[thick]{\draw[double] (0,1,0) -- (0,1,1); \blackdot{0}{1}{0}; \blackdot{0}{1}{1};},
\tikz[thick]{\draw[double] (0,1,0) -- (0,1,1); \whitedot{0}{1}{0}; \whitedot{0}{1}{1};}
 and $b_1,b_2,b_3,b_4$ be the number of \tikz{\draw[zig] (0,1,0) -- (0,1,1); \whitedot{0}{1}{0}; \blackdot{0}{1}{1};},
\tikz{\draw[zig] (0,1,0) -- (0,1,1); \blackdot{0}{1}{0}; \whitedot{0}{1}{1};},
\tikz{\draw[zig] (0,1,0) -- (0,1,1); \blackdot{0}{1}{0}; \blackdot{0}{1}{1};},
\tikz{\draw[zig] (0,1,0) -- (0,1,1); \whitedot{0}{1}{0}; \whitedot{0}{1}{1};}
 respectively.
    When we switching dimension $r_2$ and $r_3$, \fig{4dcube1} changes to \fig{4dcube2}.
    Similarly, in \fig{4dcube2}, denote $\hat a_1,\hat a_2,\hat a_3,\hat a_4$ to be the number of \tikz[thick]{\draw[double] (0,1,0) -- (0,1,1); \whitedot{0}{1}{0}; \blackdot{0}{1}{1};},
\tikz[thick]{\draw[double] (0,1,0) -- (0,1,1); \blackdot{0}{1}{0}; \whitedot{0}{1}{1};},
\tikz[thick]{\draw[double] (0,1,0) -- (0,1,1); \blackdot{0}{1}{0}; \blackdot{0}{1}{1};},
\tikz[thick]{\draw[double] (0,1,0) -- (0,1,1); \whitedot{0}{1}{0}; \whitedot{0}{1}{1};}
 and $\hat b_1,\hat b_2,\hat b_3,\hat b_4$ to be the number of \tikz{\draw[zig] (0,1,0) -- (0,1,1); \whitedot{0}{1}{0}; \blackdot{0}{1}{1};},
\tikz{\draw[zig] (0,1,0) -- (0,1,1); \blackdot{0}{1}{0}; \whitedot{0}{1}{1};},
\tikz{\draw[zig] (0,1,0) -- (0,1,1); \blackdot{0}{1}{0}; \blackdot{0}{1}{1};},
\tikz{\draw[zig] (0,1,0) -- (0,1,1); \whitedot{0}{1}{0}; \whitedot{0}{1}{1};}
 respectively.
    
        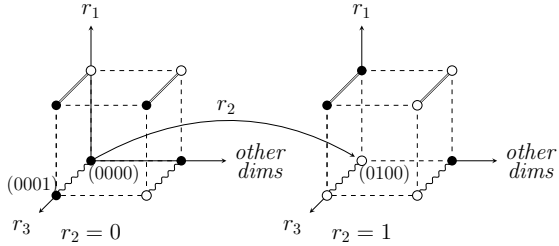
\begin{figure}[ht]
            \scalebox{0.6}{\begin{tikzpicture}[>=stealth]
    \Large
    \draw[->] (0,0,0) -- (3,0,0) node [right] {\huge$\substack{\it other\\ \it dims}$};
    \draw[->] (0,0,0) -- (0,3,0) node [above] {$r_1$};
    \draw[->] (0,0,2) -- (0,0,3) node [below left] {$r_3$};

    {
    \large
    \node at (0.5,-0.3) {$(0000)$};
    \node at (6.5,-0.3) {$(0100)$};
    \node at (-1.3,-0.55) {$(0001)$};
    }
    \node at (0,-1.6) {$r_2=0$};
    \node at (6,-1.6) {$r_2=1$};
    
    \draw[->] (8,0,0) -- (9,0,0) node [right] {\huge$\substack{\it other\\ \it dims}$};
    \draw[->] (6,2,0) -- (6,3,0) node [above] {$r_1$};
    \draw[->] (6,0,2) -- (6,0,3) node [below left] {$r_3$};
    
    \foreach \x in {0,6} {
        \draw[dashed] (\x,2,2) -- (\x,0,2) -- (\x+2,0,2) 
                      (\x,0,0) -- (\x+2,0,0)
                      (\x+2,2,0) -- (\x,2,0) (\x,2,2) -- (\x+2,2,2);
    }
                  
    \foreach \x in {0,2,6,8} {
        \draw[double] (\x,2,0) -- (\x,2,2);
        \draw[dashed] (\x,0,2) -- (\x,2,2) (\x,0,0) -- (\x,2,0);
    }
    \foreach \x in {0,2,6,8} {
        \draw[zig] (\x,0,0) -- (\x,0,2);
    }

    \blackdot{0}{0}{0}; \blackdot{2}{0}{0}; \whitedot{6}{0}{0}; \blackdot{8}{0}{0};
    \blackdot{0}{0}{2}; \whitedot{2}{0}{2}; \whitedot{6}{0}{2}; \whitedot{8}{0}{2};
    \blackdot{0}{2}{2}; \blackdot{2}{2}{2}; \blackdot{6}{2}{2}; \whitedot{8}{2}{2};
    \whitedot{0}{2}{0}; \whitedot{2}{2}{0}; \blackdot{6}{2}{0}; \whitedot{8}{2}{0};
    \draw[->,,shorten >= 5pt] (0,0) to [bend left] node [above] {$r_2$} (6,0);
\end{tikzpicture}}
            \caption{Switching from $r_2$ to $r_3$}\label{fig:4dcube2}
        \end{figure}
    
    \begin{proof}[Proof of \lem{badcase}]
    For the $1^{\it st}$ case in \lem{badcase},
        w.l.o.g, assume $b_1+a_2=0$. And we have the following 4 cases.
        \begin{itemize}
            \item $a_3+a_4=2$ : Thus all $\bm x\in\{0,1\}^n,\bm x_{r_1}=\bm x_{r_2}=0$ 
                are black; and all $\bm x\in\{0,1\}^n,\bm x_{r_1}=0,\bm x_{r_2}=1$ are white. 
                Therefore, $\hat b_3=\hat b_4=2^{n-3}$, which is 3-step solvable in the $1^{\it st}$ case of \lem{goodcase}.
            \item $b_3+b_4=2$ : Similar with case $a_3+a_4=2$.
            \item $b_3=1$ : Thus all $\bm x\in\{0,1\}^n,\bm x_{r_1}=\bm x_{r_2}=0$ are black;
                and all $\bm x\in\{0,1\}^n,\bm x_{r_1}=0,\bm x_{r_2}=1$ are white except one. 
                Therefore, $\hat b_3=2^{n-3},\hat b_4=2^{n-3}-1$, which is 3-step solvable in
                the $1^{\it st}$ case of \lem{goodcase}.
            \item $b_4=1$ : Similar with case $b_3=1$.
        \end{itemize}

        For the $2^{\it nd}$ case in \lem{badcase}, 
        since $a_3+a_4+b_3+b_4=0$, then for any $\bm x\in\{0,1\}^n$ the color of $\bm x$ is different 
        from the color of $\bm x^{\oplus r_2}$. Define
        \begin{align*}
            u_b=&\left|\left\{\text{black }\bm x\in\{0,1\}^n\middle|\bm x_{r_1}=1,\bm x_{r_2}=\bm x_{r_3}=0\right\}\right|\\
            u_w=&\left|\left\{\text{white }\bm x\in\{0,1\}^n\middle|\bm x_{r_1}=1,\bm x_{r_2}=\bm x_{r_3}=0\right\}\right|\\
            l_b=&\left|\left\{\text{black }\bm x\in\{0,1\}^n\middle|\bm x_{r_1}=0,\bm x_{r_2}=\bm x_{r_3}=0\right\}\right|\\
            l_w=&\left|\left\{\text{white }\bm x\in\{0,1\}^n\middle|\bm x_{r_1}=0,\bm x_{r_2}=\bm x_{r_3}=0\right\}\right|
        \end{align*}
        and
        \begin{align*}
            u_b'=&\left|\left\{\text{black }\bm x\in\{0,1\}^n\middle|\bm x_{r_2}=0,\bm x_{r_1}=\bm x_{r_3}=1\right\}\right|\\
            u_w'=&\left|\left\{\text{white }\bm x\in\{0,1\}^n\middle|\bm x_{r_2}=0,\bm x_{r_1}=\bm x_{r_3}=1\right\}\right|\\
            l_b'=&\left|\left\{\text{black }\bm x\in\{0,1\}^n\middle|\bm x_{r_1}=\bm x_{r_2}=0,\bm x_{r_3}=1\right\}\right|\\
            l_w'=&\left|\left\{\text{white }\bm x\in\{0,1\}^n\middle|\bm x_{r_1}=\bm x_{r_2}=0,\bm x_{r_3}=1\right\}\right|.
        \end{align*}
        By assumption, $a_1+b_2=u_w+u'_w+l_b+l'_b$, $b_1+a_2=u_b+u'_b+l_w+l'_w$. 
        And $u_w+u_b=u'_w+u'_b=l_w+l_b=l'_w+l'_b=2^{n-3}$. 
        Thus 
        \begin{align*}
        u_b+u'_b+l_b+l'_b=&(u_b+u'_b+l_w+l'_w)+(l_w+l_b)\\
        &+(l'_w+l'_b)-2(l_w+l'_w)
        \end{align*}
        is odd.
        On the other hand, 
        $\left|\{\bm x\ |\ \bm x_{r_2}=0,\bm x_{r_3}=0\}\right|=\left|\{\bm x\ |\ \bm x_{r_2}=0,\bm x_{r_3}=1\}\right|=2^{n-2}$ is even. Therefore there exists $\bm x\in\{0,1\}^n,\bm x_{r_2}=0$ such that the color of $\bm x$ is the same with the color of $\bm x^{\oplus r_3}$. Thus, $\hat a_3+\hat a_4+\hat b_3+\hat b_4>0$.
        \begin{itemize}
            \item $\hat a_3+\hat a_4+\hat b_3+\hat b_4>2$ : It is 3-step solvable in the $1^{\it st}$ case of \lem{goodcase}.
           
            \item $\hat a_3+\hat a_4+\hat b_3+\hat b_4=2$ : Thus there exists $\bm x\in\{0,1\}^n,\bm x_{r_1}=0$,
                such that $\bm x$ is white; then $\bm x^{\oplus r_3}$ and $\bm x^{\oplus r_2}$ are all black; 
                and $\bm x^{\oplus r_2,r_3}$ is white. 
                Thus when $r_2$ is swapped 
                with $r_3$, $\bm x$ with $\bm x^{\oplus r_3}$ and $\bm x^{\oplus r_2}$ with $\bm x^{\oplus r_2,r_3}$
                form \tikz[thick]{\draw[zig] (0,1,0) -- (0,1,1); \whitedot{0}{1}{0}; \blackdot{0}{1}{1};} and
                    \tikz[thick]{\draw[zig] (0,1,0) -- (0,1,1); \blackdot{0}{1}{0}; \whitedot{0}{1}{1};}. Therefore
                    $\hat b_1,\hat b_2>0$, which is 3-step solvable in the $2^{\it nd}$ case of \lem{goodcase}.
        \end{itemize}
    \end{proof}
    
    For completeness, in \lem{cannot}, we show that cases in \lem{badcase} can not be solved in the order $r_2,r_1,r_2$. 
    The proof is deferred into the appendix.

    \begin{lemma}\label{lem:cannot}
        For any $\sigma_1\in SC_{\zoton}^{(r_1)},\pi_1,\pi_2\in SC_{\zoton}^{(r_2)}$, $\sigma\pi_1\sigma_1\pi_2\notin A_{\zoton}^{(r_1)}$ if
        \begin{enumerate}
            \item $a_3+a_4+b_3+b_4=2$ and $\min\{b_1+a_2,a_1+b_2\}=0$ hold or 
            \item $a_3+a_4+b_3+b_4=0$ and $b_1+a_2$ is odd (equivalently $a_1+b_2$ is odd) hold.
        \end{enumerate}
    \end{lemma}

    \lem{new1tight} shows that 3 steps is tight for transforming arbitrary permutation into a CRBF. The proof is put into the appendix.

    \begin{lemma}\label{lem:new1tight}
        For all even number $n\geq 4$, there exists $\sigma\in A_{\zoton}$ such that $\sigma\tau\pi\notin S_{\zoton}^{(r_3)}$ for any $r_1,r_2,r_3\in[n], \tau\in SC_{\zoton}^{(r_1)},\pi\in SC_{\zoton}^{(r_2)}$.
    \end{lemma}

    \section{Transforming CRBF to identity}\label{sec:new2}
    
    In this section, we transform an even CRBF to $\id$ through $5$ CCRBFs, where the first block can be merged with the last block of \prop{new1}.

    Recall that given $\sigma\in S^{(1)}_{\zoto{n}}$, there exist $f,g\in S_{\zoto{n-1}}$ such that for all $\bm y\in \zoto{n-1}$, 
    $\sigma(0\bm y)=0f(\bm y), \sigma(1\bm y)=1g(\bm y)$. We represent $\sigma$ by $2^n\times 2^n$ matrix and $f,g$ by $2^{n-1}\times 2^{n-1}$ matrix.
    For example, if $\tau=(00,01)(10,11)\in SC_{\zoto{2}}^{(1)}$, the basis is $00,01,10,11$, then 
    \begin{align*}
        \tau=\begin{bmatrix} 0&1 & 0&0\\1&0&0&0\\0&0& 0&1\\0&0&1&0\end{bmatrix}.
    \end{align*}
    where  $f_{\tau},g_{\tau}$ are
    \begin{align*}
        f_{\tau}=g_{\tau}=\begin{bmatrix} 0&1\\1&0\end{bmatrix}=(0,1)\in S_{\zoto{1}}.
    \end{align*}

    The proof in this section is based on the following two observations.   
    The first observation is that, for any $h\in S_{\zoto{n-1}}$,
    \begin{align*}
        &\sigma\!=\!\begin{bmatrix} f & 0\\0 & g\end{bmatrix}\!
            =\!\begin{bmatrix} fh^{-1} & 0\\0&fh^{-1}\end{bmatrix}\!
        \begin{bmatrix}\id & 0\\0 & hf^{-1}gh^{-1}\end{bmatrix}\!\begin{bmatrix}h & 0\\0 & h\end{bmatrix}\!.
    \end{align*}

    The second observation is that, for $q\in S_{\zoto{n-2}}$, the following $\pi\in S_{\zoto{n}}^{(1)}$ is actually in $SC_{\zoto{n}}^{(2)}$ 
    \begin{align*}
        \pi=\begin{bmatrix} \id&0 & 0&0\\0&\id&0&0\\0&0& q&0\\0&0&0&q\end{bmatrix}.
    \end{align*}
     

    Notice that $hf^{-1}gh^{-1}$ shares same cycle pattern with $f^{-1}g$.
    If we aim to prove $\sigma$ can be decomposed to identity in $4$ steps, it suffices to show there exist $\sigma_1\in SC_{\zoto{n-1}}^{(j)},\sigma_2\in SC_{\zoto{n-1}}^{(k)}$ such that $\sigma_1\sigma_2$ has same cycle pattern with $f^{-1}g\in S_{\zoto{n-1}}$. 
    
    However, \lem{35free} indicates $\sigma_1\sigma_2$ can not formulate a single 3/5-cycle. 
    In contrast, we show that $\sigma_1\sigma_2$ can indeed achieve any cycle pattern free of 3/5-cycle by \prop{new2sum}. 
    To reduce 3/5-cycles, we develop a cycle elimination algorithm as \lem{new2rearrange}, which can be absorbed into the last block of \prop{new1}. 
    
    \begin{lemma}\label{lem:new2rearrange}
        For $n\geq 5$, $r_1\in[n]$ and $\sigma\in A_{\zoton}$, there exists $\pi\in SC_{\zoton}^{(r_1)}$ such that $\sigma\pi$ is free of 3/5-cycles.
    \end{lemma} 

    \begin{proof}
        This $\pi$ is constructed in several rounds. In round-$i$, $\pi_i\in SC_{\zoton}^{(r_1)}$ is performed. 
        Let $S_{i,c}$ be the set of $c$-cycles in $\sigma_{i-1}$ ($\sigma_0=\sigma$ and $\sigma_t=\sigma\pi_1\pi_2\cdots\pi_t$).
        
        Denote $\zeta_i=|S_{i,1}|+|S_{i,2}|+|S_{i,3}|+|S_{i,4}|+|S_{i,5}|$. If $S_{i-1,3}\cup S_{i-1,5}\neq\emptyset$, pick an arbitrary cycle $\mathscr C_1$ from it. Since $\mathscr C_1$ is 
        an odd cycle, there exists $\bm u\in\mathscr C_1$ such that $\bm v:=\bm u^{\oplus r_1}\notin\mathscr C_1$.
        Let $\mathscr C_2$ be the cycle where $\bm v$ belongs. Define
        $$
             T=\mathscr C_1\cup\big\{\bm w\in\mathscr C_2\mid\mdist^{\sigma_{i-1}}(\bm v,\bm w)\leq 5\big\}.
        $$
        Note that $|T|\leq5+11$. Since $n\geq5$ and $2^{n-1}>|T|-1$, there must exist $t\notin T$ such that $\bm u_{r_1}=\bm t_{r_1}$ and $\bm s:=\bm t^{\oplus r_1}\notin T$.
        Then, let $\pi_i=(\bm u,\bm t)(\bm v,\bm s)\in SC_{\zoton}^{(r_1)}$. We will prove $\zeta_{\sigma_i}<\zeta_{\sigma_{i-1}}$, by checking the following cases.
        \begin{itemize}
            \item $\bm t,\bm s\notin\mathscr C_2$: 
                Swapping $\bm u,\bm t$ merges $\mathscr C_1$ with another cycle
                And similarly when swapping $\bm v,\bm s$.
            \item $\bm t\notin\mathscr C_2,\bm s\in\mathscr C_2$: 
                Swapping $\bm u,\bm t$ merges $\mathscr C_1$ with another cycle.
                Then swapping $\bm v,\bm s$ splits new $\mathscr C_2$ into two cycles; and the length of neither is smaller than $6$, which will not increase the number of short cycles.
            \item $\bm t\in\mathscr C_2,\bm s\notin\mathscr C_2$: 
                Swapping $\bm u,\bm t$ merges $\mathscr C_1$ with $\mathscr C_2$.
                Then swapping $\bm v,\bm s$ merges new $\mathscr C_2$ with another cycle.
            \item $\bm t,\bm s\in\mathscr C_2$:
                Swapping $\bm u,\bm t$ merges $\mathscr C_1$ with $\mathscr C_2$.
                Then swapping $\bm v,\bm s$ splits new $\mathscr C_2$ into two cycles; 
                and the length of neither is smaller than $6$, which will not increase the number of short cycles.
        \end{itemize}
        Repeat until $S_{i,3}\cup S_{i,5}=\emptyset$. Suppose this process has $k$ rounds, then the desired permutation $\pi$ is $\pi_1\pi_2\cdots\pi_k$.
    \end{proof}


Given $r_1,r_2\in[n]$, for any $\bm x\in\zoton$, define $\bm x_{\text{out}}$ as the binary string of $\bm x$ throwing away the $r_1$- and $r_2$-th bit;
then for any $S\subseteq\zoton$ and $a,b\in\{0,1\}$, define 
$$
    S_{ab}=\{\bm x_{\text{out}}\mid
    \bm x\in S,\bm x_{r_1}=a,\bm x_{r_2}=b\}.
$$
Now we present two algorithms (\textsc{RPack} and \textsc{TPack}) to generate desired cycle patterns. 
\textsc{RPack} in \alg{rpack} performs two inplace concurrent permutations to obtain $a,b$-cycle. 
For example, Let $r_1=1,r_2=2$ and $a=4,b=6$, 
\begin{align*}
S=&\{0000,0001,0010,0100, 0101,0110,\\
&    1000,1001,1010,1100,1101,1110\}.
\end{align*}
As in \fig{rpack}, \textsc{RPack}$(r_1,r_2,a,b,S)$ returns
\begin{align*}
    \tau=&(1100,0100)(1000,0000),\\
    \pi=&(1100,1101,1110,1010)(1000,1001)\\
    &(0100,0101,0110,0010)(0000,0001).
\end{align*}

\begin{figure}[ht]
    \centering
    \scalebox{0.7}{\begin{tikzpicture}[
    scale=0.95,
    dot/.style={circle,draw=black,inner sep=0pt,minimum size=6pt}
    ]
    \Large
    \foreach \p in {0,1}
        \foreach \t/\x in {00/0,01/1,10/2}{
            \foreach \z in {0}
                \node[dot, label={[below,xshift=2mm,yshift=-1mm]\scriptsize \p\z\t}] (\p\x\z) at (\x,1.2*\z+\p*3) {};
            \foreach \z in {1}
                \node[dot, label={[above,xshift=2mm,yshift=-1mm]\scriptsize \p\z\t}] (\p\x\z) at (\x,1.2*\z+\p*3) {};
        }

    \foreach \p in {0,1}{
        \draw[-Latex,shorten >= 1pt] (\p01) -- (\p11);
        \draw[-Latex,shorten >= 1pt] (\p11) -- (\p21);
        \draw[-Latex,shorten >= 1pt] (\p21) -- (\p20);
        \draw[-Latex,shorten >= 1pt] (\p20) -- (\p01);
        \draw[Latex-Latex,shorten >= 1pt] (\p00) -- (\p10);

    }
    \draw[Latex-Latex,shorten >= 3pt,shorten <= 3pt,dashed] (000) to [bend left] (100);
    \draw[Latex-Latex,shorten >= 3pt,shorten <= 3pt,dashed] (001) to [bend left] (101);
\end{tikzpicture}}
    \caption{An example of \alg{rpack}}\label{fig:rpack}
\end{figure}

\begin{center}
    \begin{algorithm}[ht]
    \footnotesize
    \caption{\small$a,b$-cycle in rectangles (\textsc{RPack})}\label{alg:rpack}
    \DontPrintSemicolon
    \KwIn{$r_1,r_2,a,b,S$ ($0<a\leq b$)}
    \KwOut{$\pi\in SC_{\zoton}^{(r_1)},\tau\in SC_{\zoton}^{(r_2)}$}
    \tcc{$\pi\tau$ is $a,b$-cycle, $\Supp(\pi),\Supp(\tau)\subseteq S$}
    \If{$(|S|\not\equiv0\mod 4)~{\tt or}~(|S|\neq a+b)$}{
        \Return Error
        \tcc*{Invalid pattern}
    }
    \If{${\tt not}~(S_{00}=S_{01}=S_{10}=S_{11})$}{
        \Return Error
        \tcc*{Invalid support}
    }
    $k\gets\lfloor a/2\rfloor,l\gets\lfloor b/2\rfloor$\;
    \Switch{$a,b$}{
        \tcc{Fall into the first satisfied}
        \lCase{$a=b$}{Top left case}
        \lCase{$a$ is even}{Top right case}
        \lCase{$a=1,b\geq7$}{Bottom left case}
        \lCase{$a$ is odd, $a,b\geq5$}{Bottom right case}
        \lOther{\Return Error}
    }
    $\pi\gets$ solid arrows, $\tau\gets$ dashed arrows\;
    \Return $\pi,\tau$\;
        \tcc{For the meaning of following figures, see \fig{sigmapattern} and Example x}
    \resizebox{2.5cm}{2.5cm}
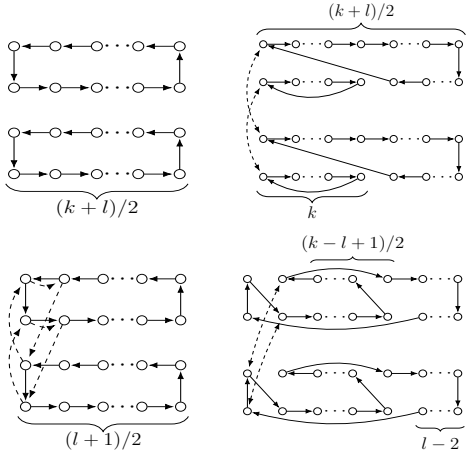
{\begin{tikzpicture}[
    scale=0.8,
    dot/.style={circle,draw=black,inner sep=0pt,minimum size=6pt}
    ]
    \large
    \foreach \p in {0,1}{
        \foreach \x in {0,1,2,3,4}
            \foreach \z in {0,1}
                \node[dot] (\p\x\z) at (\x,1.2*\z+\p*2.5) {};

        \draw[-Latex,shorten >= 1pt] (\p00) -- (\p10);
        \draw[-Latex,shorten >= 1pt] (\p10) -- (\p20);
        \draw[Latex-,shorten <= 1pt] (\p40) -- (\p30);
        \draw[-Latex,shorten >= 1pt] (\p40) -- (\p41);
        \draw[-Latex,shorten >= 1pt] (\p41) -- (\p31);
        \draw[Latex-,shorten <= 1pt] (\p11) -- (\p21);
        \draw[-Latex,shorten >= 1pt] (\p11) -- (\p01);
        \draw[-Latex,shorten >= 1pt] (\p01) -- (\p00); 

        \foreach \x in {2.25,2.5,2.75}{
            \node at (\x,0+\p*2.5) {\Large$\cdot$};
            \node at (\x,1.2+\p*2.5) {\Large$\cdot$};
        }
    }
    \draw [decorate,decoration={brace,amplitude=10pt,mirror,raise=12pt},yshift=7pt]
        (-0.2,0) -- (4+0.2,0) node [black,midway,yshift=-28pt] {$(k+l)/2$};
\end{tikzpicture}}
    \qquad
    \resizebox{3cm}{3cm}{\begin{tikzpicture}[
    scale=0.95,
    dot/.style={circle,draw=black,inner sep=0pt,minimum size=6pt}
    ]
    \Large
    \foreach \p in {0,1}{
        \foreach \x in {0,1,2,3,4,5,6}
            \foreach \z in {0,1}
                \node[dot] (\p\x\z) at (\x,1.2*\z+\p*3) {};

        \draw[-Latex,shorten >= 1pt] (\p01) -- (\p11);
        \draw[Latex-,shorten <= 1pt] (\p31) -- (\p21);
        \draw[-Latex,shorten >= 1pt] (\p31) -- (\p41);
        \draw[Latex-,shorten <= 1pt] (\p61) -- (\p51);
        \draw[-Latex,shorten >= 1pt] (\p61) -- (\p60);
        \draw[Latex-,shorten <= 1pt] (\p40) -- (\p50);
        \draw[-Latex,shorten >= 1pt] (\p00) -- (\p10);
        \draw[Latex-,shorten <= 1pt] (\p30) -- (\p20);
        \draw[-Latex,shorten >= 1pt] (\p30) to [bend left] (\p00);
        \draw[-Latex,shorten >= 1pt] (\p40) -- (\p01);

        \foreach \x in {1.25,1.5,1.75}{
            \node at (\x,0+\p*3) {\Large$\cdot$};
            \node at (\x,1.2+\p*3) {\Large$\cdot$};
        }
        \foreach \x in {4.25,4.5,4.75}{
            \node at (\x,1.2+\p*3) {\Large$\cdot$};
        }
        \foreach \x in {5.25,5.5,5.75}{
            \node at (\x,0+\p*3) {\Large$\cdot$};
        }
    }
    \draw[Latex-Latex,shorten >= 3pt,shorten <= 3pt,dashed] (000) to [bend left] (100);
    \draw[Latex-Latex,shorten >= 3pt,shorten <= 3pt,dashed] (001) to [bend left] (101);
    \draw [decorate,decoration={brace,amplitude=10pt,mirror,raise=8pt},yshift=0pt]
        (6+0.2,1.2*1+1*3) -- (0-0.2,1.2*1+1*3) node [black,midway,yshift=28pt] {$(k+l)/2$};
    \draw [decorate,decoration={brace,amplitude=10pt,mirror,raise=12pt},yshift=0pt]
        (-0.2,0) -- (3+0.2,0) node [black,midway,yshift=-30pt] {$k$};
\end{tikzpicture}}\;

    \resizebox{2.5cm}{2.5cm}{\begin{tikzpicture}[
    scale=0.8,
    dot/.style={circle,draw=black,inner sep=0pt,minimum size=6pt}
    ]
    \large
    \foreach \p in {0,1}{
        \foreach \x in {0,1,2,3,4}
            \foreach \z in {0,1}
                \node[dot] (\p\x\z) at (\x,1.2*\z+\p*2.5) {};

        \draw[-Latex,shorten >= 1pt] (\p00) -- (\p10);
        \draw[-Latex,shorten >= 1pt] (\p10) -- (\p20);
        \draw[Latex-,shorten <= 1pt] (\p40) -- (\p30);
        \draw[-Latex,shorten >= 1pt] (\p40) -- (\p41);
        \draw[-Latex,shorten >= 1pt] (\p41) -- (\p31);
        \draw[Latex-,shorten <= 1pt] (\p11) -- (\p21);
        \draw[-Latex,shorten >= 1pt] (\p11) -- (\p01);
        \draw[-Latex,shorten >= 1pt] (\p01) -- (\p00); 

        \foreach \x in {2.25,2.5,2.75}{
            \node at (\x,0+\p*2.5) {\Large$\cdot$};
            \node at (\x,1.2+\p*2.5) {\Large$\cdot$};
        }
    }
    \draw[-Latex,shorten >= 3pt,dashed] (000) to [bend left] (100);
    \draw[-Latex,shorten >= 3pt,dashed] (100) to [bend right] (110);
    \draw[-Latex,shorten >= 3pt,dashed] (110) to (000);
    \draw[-Latex,shorten >= 3pt,dashed] (001) to [bend left] (101);
    \draw[-Latex,shorten >= 3pt,dashed] (101) to [bend right] (111);
    \draw[-Latex,shorten >= 3pt,dashed] (111) to (001);
    \draw [decorate,decoration={brace,amplitude=10pt,mirror,raise=12pt},yshift=7pt]
        (-0.2,0) -- (4+0.2,0) node [black,midway,yshift=-28pt] {$(l+1)/2$};
\end{tikzpicture}}
    \qquad
    \resizebox{3cm}{3cm}{\begin{tikzpicture}[
    -Latex,shorten >= 1pt,scale=0.95,
    dot/.style={circle,draw=black,inner sep=0pt,minimum size=6pt}
    ]
    \Large
    \foreach \p in {0,1}{
        \foreach \x in {0,1,2,3,4,5,6}
            \foreach \z in {0,1}
                \node[dot] (\p\x\z) at (\x,1.2*\z+\p*3) {};

        \draw (\p00) -- (\p01);
        \draw (\p01) -- (\p10);
        \draw (\p10) -- (\p20);
        \foreach \x in {2.25,2.5,2.75,5.25,5.5,5.75}{
            \node at (\x,\p*3) {\Large$\cdot$};
            \node at (\x,1.2+\p*3) {\Large$\cdot$};
        }
        \draw (\p30) -- (\p40);
        \draw (\p40) -- (\p31);
        \draw (\p21) -- (\p11);
        \draw (\p11) to [bend left=20] (\p41);
        \draw (\p41) -- (\p51);
        \draw (\p61) -- (\p60);
        \draw (\p50) to [bend left=15] (\p00);
    }
    \draw[Latex-Latex,shorten <= 3pt,shorten >= 3pt,dashed] (000) to [bend left=5] (110);
    \draw[Latex-Latex,shorten <= 3pt,shorten >= 3pt,dashed] (001) to [bend left=5] (111);
    \draw [-,decorate,decoration={brace,amplitude=5pt,raise=12pt},yshift=2pt]
        (2-0.2,1.2*1+1*3) -- (4+0.2,1.2*1+1*3) node [black,midway,yshift=28pt] {$(k-l+1)/2$};
    \draw [-,decorate,decoration={brace,amplitude=5pt,mirror,raise=12pt},yshift=0pt]
        (5-0.2,0) -- (6+0.2,0) node [black,midway,yshift=-28pt] {$l-2$};

\end{tikzpicture}}\;
    \end{algorithm}
\end{center}

The aim of \textsc{TPack} in \alg{tpack} is to obtain $a,b,c,d$-cycle. It first divides the general rectangle shaped $S$ into two trapezoid shaped $X_0,X_1$, then performs two inplace concurrent permutations on $X_0,X_1$ to obtain $a,b$-cycle and $c,d$-cycle respectively. 
Since $a,b$-cycle and $c,d$-cycle are generated separately on $X_0,X_1$, these two parts can be performed simultaneously, thus can be combined together.

\begin{center}
    \begin{algorithm}[ht]
    \footnotesize
    \caption{\small$a,b,c,d$-cycle in trapezoids (\textsc{TPack})}\label{alg:tpack}
    \DontPrintSemicolon
    \KwIn{$r_1,r_2,a,b,c,d,S$ ($0<a\leq b,0<c\leq d$)}
    \KwOut{$\pi\in SC_{\zoton}^{(r_1)},\tau\in SC_{\zoton}^{(r_2)}$}
    \tcc{$\pi\tau$ is $a,b,c,d$-cycle, $\Supp(\pi),\Supp(\tau)\subseteq S\!\!\!$}
    \If{$(|S|\not\equiv0\mod 4)~{\tt or}~(|S|\neq a+b+c+d)$}{
        \Return Error
        \tcc*{Invalid pattern}
    }
    \If{${\tt not}~(S_{00}=S_{01}=S_{10}=S_{11})$}{
        \Return Error
        \tcc*{Invalid support}
    }
    \If{$a+b\not\equiv 2\mod 4$}{
        \Return Error
        \tcc*{Invalid pattern}
    }
    Pick $T\subseteq S_{00},|T|=\lfloor(a+b)/4\rfloor$ and $\bm t\in S_{00}\backslash T$\;
    $X_0\gets\{\bm x\in S\mid (\bm x_{\text{out}}\in T_0)\lor(\bm x_{\text{out}}=\bm t\land \bm x_{r_2}=1)\}$\;
    $X_1\gets S\backslash X_0$\;
    $\pi\gets\id,\tau\gets\id$\;
    \ForEach{$(u,v,i)\in\{(a,b,0),(c,d,1)\}$}{
        \tcc{$\Supp(\pi_i),\Supp(\tau_i)\subseteq X_i$}
        \lIf{$u=v=1$}{Skip the following}
        $k\gets\lfloor u/2\rfloor,l\gets\lfloor v/2\rfloor$\;
        \Switch{$u,v$}{
            \tcc{Fall into the first satisfied}
            \lCase{$u=v$}{Top left case}
            \lCase{$u$ is even}{Top right case}
            \lCase{$u=1,v\geq7$}{Bottom left case}
            \lCase{$u$ is odd, $u,v\geq5$}{Bottom right case}
            \lOther{\Return Error}
        }
        $\pi_i\gets$ solid arrows, $\tau_i\gets$ dashed arrows\;
        $\pi\gets\pi\pi_i,\tau\gets\tau\tau_i$\;
    }
    \Return $\pi,\tau$\;
    \tcc{For the meaning of following figures, see \fig{sigmapattern} and Example x}
    \resizebox{2.5cm}{2.5cm}{\begin{tikzpicture}[
    scale=0.8,
    dot/.style={circle,draw=black,inner sep=0pt,minimum size=6pt}
    ]
    \large
    \foreach \p in {0,1}{
        \foreach \x in {0,1,2,3,4}
            \foreach \z in {0,1}
                \node[dot] (\p\x\z) at (\x+6,1.2*\z+\p*2.5) {};
        \node[dot] (\p50) at (5+6,\p*2.5) {};

        \draw[-Latex,shorten >= 1pt] (\p00) -- (\p10);
        \draw[-Latex,shorten >= 1pt] (\p10) -- (\p20);
        \draw[Latex-,shorten <= 1pt] (\p40) -- (\p30);
        \draw[-Latex,shorten >= 1pt] (\p40) -- (\p50);
        \draw[-Latex,shorten >= 1pt] (\p50) -- (\p41);
        \draw[-Latex,shorten >= 1pt] (\p41) -- (\p31);
        \draw[Latex-,shorten <= 1pt] (\p11) -- (\p21);
        \draw[-Latex,shorten >= 1pt] (\p11) -- (\p01);
        \draw[-Latex,shorten >= 1pt] (\p01) -- (\p00); 

        \foreach \x in {2.25,2.5,2.75}{
            \node at (\x+6,0+\p*2.5) {\Large$\cdot$};
            \node at (\x+6,1.2+\p*2.5) {\Large$\cdot$};
        }
    }
    \draw [decorate,decoration={brace,amplitude=10pt,mirror,raise=12pt},yshift=7pt]
        (6-0.2,0) -- (5+6+0.2,0) node [black,midway,yshift=-28pt] {$(k+l)/2$};
\end{tikzpicture}}
    \qquad
    \resizebox{3cm}{3cm}{\begin{tikzpicture}[
    scale=0.95,
    dot/.style={circle,draw=black,inner sep=0pt,minimum size=6pt}
    ]
    \Large
    \foreach \p in {0,1}{
        \foreach \x in {0,1,2,3,4,5,6}
            \foreach \z in {0,1}
                \node[dot] (\p\x\z) at (\x+7,1.2*\z+\p*3) {};
        \node[dot] (\p80) at (16-2,\p*3) {};

        \draw[-Latex,shorten >= 1pt] (\p01) -- (\p11);
        \draw[Latex-,shorten <= 1pt] (\p31) -- (\p21);
        \draw[-Latex,shorten >= 1pt] (\p31) -- (\p41);
        \draw[Latex-,shorten <= 1pt] (\p61) -- (\p51);
        \draw[-Latex,shorten >= 1pt] (\p61) -- (\p80);
        \draw[-Latex,shorten >= 1pt] (\p80) -- (\p60);
        \draw[Latex-,shorten <= 1pt] (\p40) -- (\p50);
        \draw[-Latex,shorten >= 1pt] (\p00) -- (\p10);
        \draw[Latex-,shorten <= 1pt] (\p30) -- (\p20);
        \draw[-Latex,shorten >= 1pt] (\p30) to [bend left] (\p00);
        \draw[-Latex,shorten >= 1pt] (\p40) -- (\p01);

        \foreach \x in {1.25,1.5,1.75}{
            \node at (\x+7,0+\p*3) {\Large$\cdot$};
            \node at (\x+7,1.2+\p*3) {\Large$\cdot$};
        }
        \foreach \x in {4.25,4.5,4.75}{
            \node at (\x+7,1.2+\p*3) {\Large$\cdot$};
        }
        \foreach \x in {5.25,5.5,5.75}{
            \node at (\x+7,0+\p*3) {\Large$\cdot$};
        }
    }
    \draw[Latex-Latex,shorten >= 3pt,shorten <= 3pt,dashed] (000) to [bend left] (100);
    \draw[Latex-Latex,shorten >= 3pt,shorten <= 3pt,dashed] (001) to [bend left] (101);
    \draw [decorate,decoration={brace,amplitude=10pt,mirror,raise=8pt},yshift=0pt]
        (6+7+0.2,1.2*1+1*3) -- (0+7-0.2,1.2*1+1*3) node [black,midway,yshift=28pt] {$(k+l-1)/2$};
    \draw [decorate,decoration={brace,amplitude=10pt,mirror,raise=12pt},yshift=0pt]
        (7-0.2,0) -- (3+7+0.2,0) node [black,midway,yshift=-30pt] {$k$};
\end{tikzpicture}}\;

    \resizebox{2.5cm}{2.5cm}{\begin{tikzpicture}[
    scale=0.8,
    dot/.style={circle,draw=black,inner sep=0pt,minimum size=6pt}
    ]
    \large
    \foreach \p in {0,1}{
        \foreach \x in {0,1,2,3,4}
            \foreach \z in {0,1}
                \node[dot] (\p\x\z) at (\x+6,1.2*\z+\p*2.5) {};
        \node[dot] (\p50) at (5+6,\p*2.5) {};

        \draw[-Latex,shorten >= 1pt] (\p00) -- (\p10);
        \draw[-Latex,shorten >= 1pt] (\p10) -- (\p20);
        \draw[Latex-,shorten <= 1pt] (\p40) -- (\p30);
        \draw[-Latex,shorten >= 1pt] (\p40) -- (\p50);
        \draw[-Latex,shorten >= 1pt] (\p50) -- (\p41);
        \draw[-Latex,shorten >= 1pt] (\p41) -- (\p31);
        \draw[Latex-,shorten <= 1pt] (\p11) -- (\p21);
        \draw[-Latex,shorten >= 1pt] (\p11) -- (\p01);
        \draw[-Latex,shorten >= 1pt] (\p01) -- (\p00); 

        \foreach \x in {2.25,2.5,2.75}{
            \node at (\x+6,0+\p*2.5) {\Large$\cdot$};
            \node at (\x+6,1.2+\p*2.5) {\Large$\cdot$};
        }
    }
    \draw[-Latex,shorten >= 3pt,dashed] (000) to [bend left] (100);
    \draw[-Latex,shorten >= 3pt,dashed] (100) to [bend right] (110);
    \draw[-Latex,shorten >= 3pt,dashed] (110) to (000);
    \draw[-Latex,shorten >= 3pt,dashed] (001) to [bend left] (101);
    \draw[-Latex,shorten >= 3pt,dashed] (101) to [bend right] (111);
    \draw[-Latex,shorten >= 3pt,dashed] (111) to (001);
    \draw [decorate,decoration={brace,amplitude=10pt,mirror,raise=12pt},yshift=7pt]
        (6-0.2,0) -- (5+6+0.2,0) node [black,midway,yshift=-28pt] {$(l+2)/2$};
\end{tikzpicture}}
    \qquad
    \resizebox{3cm}{3cm}{\begin{tikzpicture}[
    -Latex,shorten >= 1pt,scale=0.95,
    dot/.style={circle,draw=black,inner sep=0pt,minimum size=6pt}
    ]
    \Large
    \foreach \p in {0,1}{
        \foreach \x in {0,1,2,3,4,5,6}
            \foreach \z in {0,1}
                \node[dot] (\p\x\z) at (\x+7,1.2*\z+\p*3) {};
        \node[dot] (\p70) at (14,\p*3) {}; 

        \draw (\p00) -- (\p01);
        \draw (\p01) -- (\p10);
        \draw (\p10) -- (\p20);
        \foreach \x in {2.25,2.5,2.75,5.25,5.5,5.75}{
            \node at (\x+7,\p*3) {\Large$\cdot$};
            \node at (\x+7,1.2+\p*3) {\Large$\cdot$};
        }

        \draw (\p30) -- (\p40);
        \draw (\p40) -- (\p41);
        \draw (\p21) -- (\p11);
        \draw (\p11) to [bend left=20] (\p51);
        \draw (\p41) -- (\p31);
        \draw (\p61) -- (\p70);
        \draw (\p70) -- (\p60);
        \draw (\p50) to [bend left=15] (\p00);
    }
    \draw[Latex-Latex,shorten <= 3pt,shorten >= 3pt,dashed] (000) to [bend left=5] (110);
    \draw[Latex-Latex,shorten <= 3pt,shorten >= 3pt,dashed] (001) to [bend left=5] (111);
    \draw [-,decorate,decoration={brace,amplitude=5pt,raise=12pt},yshift=2pt]
        (2+7-0.2,1.2*1+1*3) -- (4+7+0.2,1.2*1+1*3) node [black,midway,yshift=28pt] {$(k-l)/2$};
    \draw [-,decorate,decoration={brace,amplitude=5pt,mirror,raise=12pt},yshift=0pt]
        (5+7-0.2,0) -- (7+7+0.2,0) node [black,midway,yshift=-28pt] {$l-1$};
\end{tikzpicture}}\;
    \end{algorithm}
\end{center}

    Now we give the proof of \prop{new2sum}, which states two CCRBFs can compose most of the patterns.
    
    \begin{proof}[Proof of \prop{new2sum}]\label{ProofP3}
    W.l.o.g, assume $r_1=1,r_2=2$.
    Let $c_k$ be the number of $k$-cycles in $\sigma$ and $c_1$ is the number of fix-points.
    
    Now, we initialize $\pi=\tau=\id,T=\{0,1\}^{n-2}$ and construct them in two stages.

    \noindent\textbf{Stage I (Pairing).}\quad
    Initialize the set of pairs as $P=\emptyset$.
    \begin{itemize}
        \item Pick $i$ with $c_i>0$ and update $c_i\gets c_i-1$.
        \item Pick $j$ with $c_j>0,i+j\equiv0\mod2$ and update $c_j\gets c_j-1$.
        \item Swap $i,j$ if $i>j$. Then add $(i,j)$ to $P$.
    \end{itemize}
    Repeat the procedure until $c_i=0$ for any $i$.
    
    Since $\sigma$ is even, we have $\sum_ic_{2i}\equiv0\mod2$.
    Meanwhile, $\sum_ic_{2i-1}\equiv\sum_kkc_k\equiv2^n\equiv0\mod 2$.
    Thus as long as the first step succeeds, the second step will not fail.
    
    \noindent\textbf{Stage II (Construct).}\quad
    Now we construct $\pi,\tau$.
    \begin{itemize}
        \item Pick $(a,b)\in P$ and remove it from $P$.
        \item If $a+b\equiv0\mod4$, select $S\subseteq T,|S|=(a+b)/4$. Let
        $$
        \pi',\tau'\gets\textsc{RPack}\left(r_1,r_2,a,b,\{0,1\}^2\times T\right).
        $$
        \item If $a+b\equiv2\mod4$, pick $(c,d)\in P,c+d\equiv2\mod4$ and remove it from $P$.
        Select $S\subseteq T,|S|=(a+b+c+d)/4$. Let
        $$
        \pi',\tau'\gets\textsc{TPack}\left(r_1,r_2,a,b,c,d,\{0,1\}^2\times T\right).
        $$
        \item Update $T\gets T\backslash S,\pi\gets\pi\pi',\tau\gets\tau\tau'$.
    \end{itemize}
    Repeat the procedure until $P=\emptyset$.
    
    Since $\sum_{(a,b)\in P}a+b=2^n$ and $n\geq4$, if there is $a+b\equiv2\mod4$ then there must be another pair $c+d\equiv2\mod4$. Also, $\sigma$ is free of 3/5-cycle, thus \textsc{RPack} and \textsc{TPack} will not err.
    
    Since $\pi',\tau'$'s are inplace and separate, $\pi,\tau$ is the desired permutation.
    \end{proof}

    Combining these result, finally we are able to prove \prop{new2}. 

    \begin{proof}[Proof of \prop{new2}]
        W.l.o.g, we assume $r_1=1, r_2=2$.
        Since $\sigma\in A_{\zoton}^{(r_1)}$, there exist $f,g\in S_{\zoto{n-1}}$ such that 
        \begin{align*}
            &\sigma=\begin{bmatrix} f & 0\\0 & g\end{bmatrix}.
        \end{align*}
        Let $\pi_1=\begin{bmatrix}\id&0\\0&g'\end{bmatrix}$, we have
        \begin{align*}
            \sigma\pi_1
                &=\begin{bmatrix}f & 0\\0 & g\end{bmatrix}\begin{bmatrix} \id & 0\\0 & g'\end{bmatrix}\\
                    &=\begin{bmatrix} fh^{-1} & 0\\0&fh^{-1}\end{bmatrix}
            \begin{bmatrix}\id & 0\\0 & hf^{-1}gg'h^{-1}\end{bmatrix}\begin{bmatrix}h & 0\\0 & h\end{bmatrix},
        \end{align*}
        where $f,g',g,h\in S_{\zoto{n-1}}$ and $g',h$ shall be determined later.
        
        Since $f^{-1}g$ is even, by \lem{new2rearrange}, there exists $g'\in SC_{\zoto{n-1}}^{(r_2)}$ such that
        $f^{-1}gg'$ is free of 3/5-cycle.
        Then by \prop{new2sum}, there exist $\rho_1\in SC_{\zoto{n-1}}^{(r_4)},\rho_2\in SC_{\zoto{n-1}}^{(r_3)}$ such that
        $\rho_1\rho_2$ has the same cycle pattern as $f^{-1}gg'$. This condition is equal to that there exists $h\in S_{\zoto{n-1}}$ such that  $hf^{-1}gg'h^{-1}=\rho_1\rho_2$. Therefore
        $$
            \sigma\pi_1=\begin{bmatrix}fh^{-1}&0\\0&fh^{-1}\end{bmatrix}
            \begin{bmatrix}\id&0\\0&\rho_1\end{bmatrix}
            \begin{bmatrix}\id&0\\0&\rho_2\end{bmatrix}\begin{bmatrix}h&0\\0&h\end{bmatrix}.
        $$
        Then setting 
        \begin{gather*}
            \pi_1=\begin{bmatrix}\id&0\\0&g'\end{bmatrix},
            \sigma_1=\begin{bmatrix}h^{-1}&0\\0&h^{-1}\end{bmatrix},
            \tau_1= \begin{bmatrix}\id&0\\0&\rho_2^{-1}\end{bmatrix},\\
            \tau_2=\begin{bmatrix}\id&0\\0&\rho_1^{-1}\end{bmatrix},
            \sigma_2=\begin{bmatrix}hf^{-1}&0\\0&hf^{-1}\end{bmatrix}
        \end{gather*} 
        will do. 
    \end{proof}
    
   For completeness, we show in \lem{35free} that the restriction that the cycle pattern contains no 3/5-cycle is inevitable. The proof is put into the appendix.
    \begin{lemma}\label{lem:35free}
        For any $\sigma_1\in SC_{\zoton}^{(r_1)},\sigma_2\in SC_{\zoton}^{(r_2)}$, $\sigma_1\sigma_2$ can not be a permutation that is merely a $3$-cycle or a $5$-cycle.
    \end{lemma}

\section{An explicit example of our algorithm}\label{sec:example}

In this section, we decompose a specified $\sigma\in A_{\{0,1\}^4}$ to 7 blocks of $3$-bit RBFs by our algorithm. Here
\begin{align*}
\sigma=&(1001,1100,0101)(1110,0110,0111,1111)\\ &(1010,0010,0011,1011).
\end{align*}
\subsection[Transform sigma to CRBF ]{Transform $\sigma$ to CRBF}

\noindent\textbf{Step 1.} Choose $r_1=1,r_2=2$. Using method in \sec{visualizing}, we construct colored cube for $\sigma$ as \fig{sigmacolor1}.
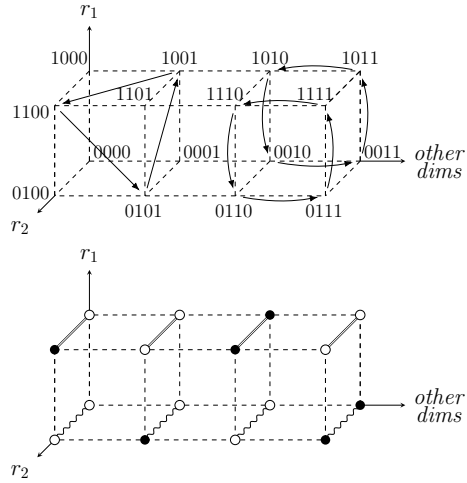
\begin{figure}[ht]
    \centering
    \scalebox{0.6}{\begin{tikzpicture}[>=stealth]
    \Large
    \draw[dashed] (0,0,0) -- (6,0,0);
    \draw[->] (6,0,0) -- (7,0,0) node [right] {\huge$\substack{\it other\\ \it dims}$};
    \draw[dashed] (0,0,0) -- (0,2,0);
    \draw[->] (0,2,0) -- (0,3,0) node [above] {$r_1$};
    \draw[->] (0,0,2) -- (0,0,3) node [below left] {$r_2$};

    \draw[dashed] (0,2,2) -- (0,0,2) -- (6,0,2) (6,0,0) -- (6,2,0) -- (0,2,0)
               (0,2,2) -- (6,2,2) -- (6,2,0) (6,2,2) -- (6,0,2);
    \foreach \x in {0,2,4,6} {
        \draw[dashed] (\x,2,0) -- (\x,2,2);
        \draw[dashed] (\x,0,2) -- (\x,2,2) (\x,0,0) -- (\x,2,0);
    }
    \foreach \x in {0,2,4,6} {
        \draw[dashed] (\x,0,0) -- (\x,0,2);
    }

    \large
    \node at (0.5,0.2) {$0000$}; \node at (2.5,0.2) {$0001$};
    \node at (4.5,0.2) {$0010$}; \node at (6.5,0.2) {$0011$};
    \node at (-0.45,2.3) {$1000$}; \node at (2,2.3) {$1001$};
    \node at (4,2.3) {$1010$}; \node at (6,2.3) {$1011$};
    \node at (-1.3,1.1) {$1100$}; \node at (-1+2,1.5) {$1101$};
    \node at (-1+4,1.5) {$1110$}; \node at (-1+6,1.5) {$1111$};
    \node at (-1.3,-0.7) {$0100$}; \node at (1.2,-1.1) {$0101$};
    \node at (1.2+2,-1.1) {$0110$}; \node at (1.2+4,-1.1) {$0111$};

    \draw[-Latex,shorten >= 5pt, shorten <= 5pt] (2,0,2) to (2,2,0);
    \draw[-Latex,shorten >= 5pt, shorten <= 5pt] (2,2,0) to (0,2,2);
    \draw[-Latex,shorten >= 5pt, shorten <= 5pt] (0,2,2) to (2,0,2);
    \foreach \z in {0,2}{
        \draw[-Latex,shorten >= 5pt, shorten <= 5pt] (4,0,\z) to [bend right=10] (6,0,\z);
        \draw[-Latex,shorten >= 5pt, shorten <= 5pt] (6,0,\z) to [bend right=15] (6,2,\z);
        \draw[-Latex,shorten >= 5pt, shorten <= 5pt] (6,2,\z) to [bend right=10] (4,2,\z);
        \draw[-Latex,shorten >= 5pt, shorten <= 5pt] (4,2,\z) to [bend right=15] (4,0,\z);
    }

\end{tikzpicture}}
    
    \scalebox{0.6}{\begin{tikzpicture}[>=stealth]
    \Large
    \draw[dashed] (0,0,0) -- (6,0,0);
    \draw[->] (6,0,0) -- (7,0,0) node [right] {\huge$\substack{\it other\\ \it dims}$};
    \draw[dashed] (0,0,0) -- (0,2,0);
    \draw[->] (0,2,0) -- (0,3,0) node [above] {$r_1$};
    \draw[->] (0,0,2) -- (0,0,3) node [below left] {$r_2$};

    \draw[dashed] (0,2,2) -- (0,0,2) -- (6,0,2) (6,0,0) -- (6,2,0) -- (0,2,0)
               (0,2,2) -- (6,2,2) -- (6,2,0) (6,2,2) -- (6,0,2);
    \foreach \x in {0,2,4,6} {
        \draw[double] (\x,2,0) -- (\x,2,2);
        \draw[dashed] (\x,0,2) -- (\x,2,2) (\x,0,0) -- (\x,2,0);
    }
    \foreach \x in {0,2,4,6} {
        \draw[zig] (\x,0,0) -- (\x,0,2);
    }
    \foreach \x/\y/\z in {0/0/0,2/0/0,4/0/0,0/2/0,2/2/0,6/2/0,0/0/2,4/0/2,2/2/2,6/2/2}{
        \whitedot{\x}{\y}{\z};
    }
    \foreach \x/\y/\z in {6/0/0,4/2/0,0/2/2,2/0/2,6/0/2,4/2/2}{
        \blackdot{\x}{\y}{\z};
    }
\end{tikzpicture}}
    \caption{Visualize $\sigma$ on a colored cube} \label{fig:sigmacolor1}
\end{figure}
Read the colored cube, we get $a_1=1,a_2=0,a_3=1,a_4=2;b_1=1,b_2=0,b_3=1,b_4=2$.
\medskip

\noindent\textbf{Step 2.} Check \lem{goodcase} and \lem{badcase}, we find this case falls into \lem{goodcase}. we can transform $\sigma$ to $S^{(1)}_{\{0,1\}^4}$ by $SC^{(2)}_{\{0,1\}^4}$,$SC^{(1)}_{\{0,1\}^4}$,$SC^{(2)}_{\{0,1\}^4}$ by \lem{goodcase}. Specific construction are as follows.
\smallskip

\begin{casepar}{Step 2.1}
    Using \lem{canonical}, we transform $\sigma$ to canonical form by $\pi=\pi_1\pi_2$. Let 
    \[
        \pi_1=(1110,0111)(1010,0011),
    \]
    which transforms the colored cube to a cube with $a_3=b_3=0,a_2=0$. Setting $\pi_2=(0100,0101)(0000,0001)$, it rearrange the cube to canonical form. The process is pictured as \fig{sigmatau}, \fig{sigmatautau}.
    \begin{figure}[ht]
        \centering
        \scalebox{0.6}{\begin{tikzpicture}[>=stealth]
    \Large
    \draw[dashed] (0,0,0) -- (6,0,0);
    \draw[->] (6,0,0) -- (7,0,0) node [right] {\huge$\substack{\it other\\ \it dims}$};
    \draw[dashed] (0,0,0) -- (0,2,0);
    \draw[->] (0,2,0) -- (0,3,0) node [above] {$r_1$};
    \draw[->] (0,0,2) -- (0,0,3) node [below left] {$r_2$};

    \draw[dashed] (0,2,2) -- (0,0,2) -- (6,0,2) (6,0,0) -- (6,2,0) -- (0,2,0)
               (0,2,2) -- (6,2,2) -- (6,2,0) (6,2,2) -- (6,0,2);
    \foreach \x in {0,2,4,6} {
        \draw[double] (\x,2,0) -- (\x,2,2);
        \draw[dashed] (\x,0,2) -- (\x,2,2) (\x,0,0) -- (\x,2,0);
    }
    \foreach \x in {0,2,4,6} {
        \draw[zig] (\x,0,0) -- (\x,0,2);
    }
    \foreach \x/\y/\z in {4/0/0,0/2/0,2/2/0,6/2/0,0/0/2,4/0/2,2/2/2,6/2/2,0/2/2,2/0/2,6/0/2,4/2/2,6/0/0,4/2/0,0/0/0,2/0/0}{
        \whitedot{\x}{\y}{\z};
    }
    \foreach \x/\y/\z in {0/2/2,2/0/2}{
        \blackdot{\x}{\y}{\z};
    }
\end{tikzpicture}}
        \caption{Colored cube for $\sigma\pi_1$}        \label{fig:sigmatau}
        \scalebox{0.6}{\begin{tikzpicture}[>=stealth]
    \Large
    \draw[dashed] (0,0,0) -- (6,0,0);
    \draw[->] (6,0,0) -- (7,0,0) node [right] {\huge$\substack{\it other\\ \it dims}$};
    \draw[dashed] (0,0,0) -- (0,2,0);
    \draw[->] (0,2,0) -- (0,3,0) node [above] {$r_1$};
    \draw[->] (0,0,2) -- (0,0,3) node [below left] {$r_2$};

    \draw[dashed] (0,2,2) -- (0,0,2) -- (6,0,2) (6,0,0) -- (6,2,0) -- (0,2,0)
               (0,2,2) -- (6,2,2) -- (6,2,0) (6,2,2) -- (6,0,2);
    \foreach \x in {0,2,4,6} {
        \draw[double] (\x,2,0) -- (\x,2,2);
        \draw[dashed] (\x,0,2) -- (\x,2,2) (\x,0,0) -- (\x,2,0);
    }
    \foreach \x in {0,2,4,6} {
        \draw[zig] (\x,0,0) -- (\x,0,2);
    }
    \foreach \x/\y/\z in {4/0/0,0/2/0,2/2/0,6/2/0,0/0/2,4/0/2,2/2/2,6/2/2,0/2/2,2/0/2,6/0/2,4/2/2,6/0/0,4/2/0,0/0/0,2/0/0}{
        \whitedot{\x}{\y}{\z};
    }
    \foreach \x/\y/\z in {0/2/2,0/0/2}{
        \blackdot{\x}{\y}{\z};
    }
\end{tikzpicture}}
        \caption{Colored cube for $\sigma\pi_1\pi_2$} \label{fig:sigmatautau}
    \end{figure}
\end{casepar}

\begin{casepar}{Step 2.2}
    Using \lem{goodcase}, we construct the following CCRBFs
    \begin{align*}
        \pi_3=  &(0100,1110)(0000,1010)(1101,\\
                &0110)(1001,0010)\\
        \pi_4=  &(1000,1010)(0000,0010)\\
        \pi_5=  &(1100,0100)(1000,0000)(1101,\\
                &0110)(1001,0010).
    \end{align*}
    
    It's easy to verify $\pi_1,\pi_2,\pi_3,\pi_5 \in SC^{(2)}_{\{0,1\}^4}$,$\pi_4\in SC^{(1)}_{\{0,1\}^4}$ and 
    \begin{align*}
        \pi_1\pi_2\pi_3=&(0000,0011,1010,0001)(0100,0111,\\
                        &1110,0101)(0010,1001)(0110,1101).
    \end{align*}

    And finally we transform the colored cube for $\sigma$ to a white cube by verifying
    \begin{align*}
        \sigma^{(1)}=&\sigma(\pi_1\pi_2\pi_3)\pi_4\pi_5\\
        =&(0000,0001)(0010,0011)(0100,0101)\\
        &(0110,0111)(1000,1100,1111,1110,\\
        &1001,1011,1010).
    \end{align*}
\end{casepar}

\subsection[Transform sigma1 to identity]{Transform $\sigma^{(1)}$ to identity}
We can use two $3$-bit RBFs to represent $\sigma^{(1)}$. That is
\begin{align*}
    f=&(000,001)(010,011)(100,101)(110,111)\\
    g=&(000,100,111,110,001,011,010)
\end{align*}
such that 
$$
    \sigma^{(1)}: (0,\bm x)\rightarrow (0,f(\bm x)), (1,\bm x)\rightarrow(1,g(\bm x)).
$$

\begin{casepar}{Step 1}
    Determine whether $f^{-1}g$ has 3/5-cycle. By directly calculating $f^{-1}g$ we know the answer is no. So we can jump the process for eliminating 3/5-cycles.
    $$
        f^{-1}g=(000,101,100,110)(001,010).
    $$
\end{casepar}
\begin{casepar}{Step 2}
    First, we construct a $\sigma_1\in SC^{(1)}_{\{0,1\}^3},$ a $\sigma_2\in SC^{(2)}_{\{0,1\}^3}$ to generate a 2,4-cycle pattern like $f^{-1}g$. Based on Algorithm \textsc{TPack} 
    \begin{align*}
        \sigma_1=&(000,011)(100,111)\\
        \sigma_2=&(010,110)(000,100).
    \end{align*}
\end{casepar}
\begin{casepar}{Step 3}
    Find $h\in S_{\{0,1\}^3}$ such that $h(f^{-1}g)h^{-1}=\sigma_1\sigma_2$. By group theory we know that if $\tau=(i_1,i_2,...,i_k)$, then $h\tau h^{-1}=(h(i_1),h(i_2),...,h(i_k))$. So we can construct 
    $$
    h=(101,111)(001,010,110,011).
    $$
\end{casepar}
\begin{casepar}{Step 4}
    Now we verify $h(f^{-1}g)h^{-1}=\sigma_1\sigma_2$. Thus
    \begin{align*}
        \sigma^{(1)}\!=\!&\begin{bmatrix}f& \\&g\end{bmatrix}\\
        \!=\!&\begin{bmatrix}fh^{-1}&\\&fh^{-1}\end{bmatrix}\!
        \begin{bmatrix}\id&\\&\sigma_1\end{bmatrix}\!\begin{bmatrix}\id&\\&\sigma_2\end{bmatrix}\!\begin{bmatrix}h&\\&h\end{bmatrix}\\
        \!\triangleq&\pi_6\pi_7\pi_8\pi_9.
    \end{align*}
    Written in the form of permutation cycle pattern, 
    \begin{align*}
        \pi_6=  &(0000,0001,0010)(0011,0111,0100,0101,0110)\\
                &(1000,1001,1010)(1011,1111,1100,1101,1110)\\
        \pi_7=  &(1000,1011)(1100,1111)\\
        \pi_8=  &(1010,1110)(1000,1100)\\
        \pi_9=  &(0101,0111)(0001,0010,0110,0011)\\
                &(1101,1111)(1001,1010,1110,1011).
    \end{align*} 
\end{casepar}

\subsection{Summary}
In a word, $\sigma=\pi_6\pi_7\pi_8\pi_9\pi_5^{-1}\pi_4^{-1}(\pi_1\pi_2\pi_3)^{-1}$,
where 
\begin{align*}
&\pi_6,\pi_9,\pi_4^{-1}\in SC^{(1)}_{\{0,1\}^4},\\
&\pi_7,\pi_5^{-1},(\pi_1\pi_2\pi_3)^{-1}\in SC^{(2)}_{\{0,1\}^4},\\
&\pi_8\in SC^{(3)}_{\{0,1\}^4}.
\end{align*}

\section{Even block depth}\label{sec:EvenBlock}

In previous sections, we prove for any $\sigma\in A_{\zoton},n\geq6$, $\sigma$ has block depth $7$. However, the block itself may be an odd permutation which resists further decomposition. In this section, we address this concern and show that any $\sigma\in A_{\zoton}$, with $n\geq10$, has even block depth $10$, which is stated as \thm{10steps}. This is proven by some modification of the framework in previous sections. The idea is similar, but the analysis is much more complicated. Here we only sketch the proof and leave the detail in the appendix. 

We prove \thm{10steps} by the modified versions of \prop{new1} and \prop{new2}. Specifically, we prove that arbitrary even $n$-bit permutation can be transformed to even CRBF by $3$ even blocks; arbitrary even CRBF can be transformed to identity by $8$ even blocks. Choosing carefully, we can merge some of them and finally decompose even $n$-bit permutation to identity using $10$ even blocks. The results are summarized as the following two propositions.

\begin{proposition}\label{prop:new1even}
    For $n\geq 4,\sigma\in A_{\zoton}$ and $r_1\in[n]$, there exist at least $(n-2)$ different $r_2\in[n]\backslash\{r_1\}$ such that there exist $\sigma_1\in AC_{\zoton}^{(r_1)},\pi_1,\pi_2\in AC_{\zoton}^{(r_2)}$ satisfying $\sigma\pi_1\sigma_1\pi_2\in A_{\zoton}^{(r_1)}.$
\end{proposition}

Here we only give the intuition. The key observation in the proof of \lem{goodcase} is that we can always swap some nodes without changing color in cuboid. For example, if we swap two nodes who has the the same color and lie in the same face, then the corresponding colored cuboid will not change. This observation can be used to modify the  permutation to be concurrently even. 

For example, we can transform two B-cards to white cube by the following two methods. 
\begin{figure}[ht]
    \centering
    \input{./pics/part1split00x2.tex}
    \caption{Transform two B-cards to identity where $\tau_1$ is concurrently even, $\tau_2,\tau_3$ are  concurrently odd. }
    \vspace{1em}
    \scalebox{0.8}{\begin{tikzpicture}
    \Large
    \draw[dashed] (0,2,0) -- (0,1,0) (0,2,1) -- (0,1,1) (1,2,0) -- (1,1,0) (1,2,1) -- (1,1,1)
        (0,2,0) -- (1,2,0) (0,2,1) -- (1,2,1) (0,1,0) -- (1,1,0) (0,1,1) -- (1,1,1);
    \draw[dashed] (0,2,0) -- (0,2,1); \whitedot{0}{2}{0}; \blackdot{0}{2}{1}; 
    \draw[dashed] (0,1,0) -- (0,1,1); \blackdot{0}{1}{1}; \whitedot{0}{1}{0};
    \draw[dashed] (1,2,0) -- (1,2,1); \whitedot{1}{2}{0}; \blackdot{1}{2}{1}; 
    \draw[dashed] (1,1,0) -- (1,1,1); \blackdot{1}{1}{1}; \whitedot{1}{1}{0};
    \draw[Latex-Latex,shorten >= 5pt, shorten <= 5pt,very thin] (0,2,0) to [bend left=15] (1,2,0);
    \draw[Latex-Latex,shorten >= 5pt, shorten <= 5pt,very thin] (0,2,1) to [bend left=15] (1,2,1);
    \node at (1.6,1.5,0) {$\xLongrightarrow{\phantom{\tau_1}}$};

    \pgfmathsetmacro{\x}{2.7}
    \draw[dashed] (0+\x,2,0) -- (0+\x,1,0) (0+\x,2,1) -- (0+\x,1,1) (1+\x,2,0) -- (1+\x,1,0) (1+\x,2,1) -- (1+\x,1,1)
        (0+\x,2,0) -- (1+\x,2,0) (0+\x,2,1) -- (1+\x,2,1) (0+\x,1,0) -- (1+\x,1,0) (0+\x,1,1) -- (1+\x,1,1);
    \draw[dashed] (0+\x,2,0) -- (0+\x,2,1); \whitedot{0+\x}{2}{0}; \blackdot{0+\x}{2}{1}; 
    \draw[dashed] (0+\x,1,0) -- (0+\x,1,1); \blackdot{0+\x}{1}{1}; \whitedot{0+\x}{1}{0};
    \draw[dashed] (1+\x,2,0) -- (1+\x,2,1); \whitedot{1+\x}{2}{0}; \blackdot{1+\x}{2}{1}; 
    \draw[dashed] (1+\x,1,0) -- (1+\x,1,1); \blackdot{1+\x}{1}{1}; \whitedot{1+\x}{1}{0};
    \node at (1.6+\x,1.5,0) {$\xLongrightarrow{\tau_1'}$};
    
    \pgfmathsetmacro{\x}{5.4}
    \draw[dashed] (0+\x,2,0) -- (0+\x,1,0) (0+\x,2,1) -- (0+\x,1,1) (1+\x,2,0) -- (1+\x,1,0) (1+\x,2,1) -- (1+\x,1,1)
        (0+\x,2,0) -- (1+\x,2,0) (0+\x,2,1) -- (1+\x,2,1) (0+\x,1,0) -- (1+\x,1,0) (0+\x,1,1) -- (1+\x,1,1);
    \draw[dashed] (0+\x,2,0) -- (0+\x,2,1); \whitedot{0+\x}{2}{0}; \blackdot{0+\x}{2}{1}; 
    \draw[dashed] (0+\x,1,0) -- (0+\x,1,1); \blackdot{0+\x}{1}{1}; \whitedot{0+\x}{1}{0};
    \draw[dashed] (1+\x,2,0) -- (1+\x,2,1); \whitedot{1+\x}{2}{0}; \blackdot{1+\x}{2}{1}; 
    \draw[dashed] (1+\x,1,0) -- (1+\x,1,1); \blackdot{1+\x}{1}{1}; \whitedot{1+\x}{1}{0};
    \draw[Latex-Latex,shorten >= 5pt, shorten <= 5pt,very thin] (0+\x,2,1) to [bend left=15] (1+\x,2,1);
    \draw[Latex-Latex,shorten >= 5pt, shorten <= 5pt,very thin] (0+\x,1,1) to [bend left=15] (1+\x,1,1);
    \node at (1.6+\x,1.5,0) {$\xLongrightarrow{\phantom{\tau_1}}$};

    \pgfmathsetmacro{\x}{8.1}
    \draw[dashed] (0+\x,2,0) -- (0+\x,1,0) (0+\x,2,1) -- (0+\x,1,1) (1+\x,2,0) -- (1+\x,1,0) (1+\x,2,1) -- (1+\x,1,1)
        (0+\x,2,0) -- (1+\x,2,0) (0+\x,2,1) -- (1+\x,2,1) (0+\x,1,0) -- (1+\x,1,0) (0+\x,1,1) -- (1+\x,1,1);
    \draw[dashed] (0+\x,2,0) -- (0+\x,2,1); \whitedot{0+\x}{2}{0}; \blackdot{0+\x}{2}{1}; 
    \draw[dashed] (0+\x,1,0) -- (0+\x,1,1); \blackdot{0+\x}{1}{1}; \whitedot{0+\x}{1}{0};
    \draw[dashed] (1+\x,2,0) -- (1+\x,2,1); \whitedot{1+\x}{2}{0}; \blackdot{1+\x}{2}{1}; 
    \draw[dashed] (1+\x,1,0) -- (1+\x,1,1); \blackdot{1+\x}{1}{1}; \whitedot{1+\x}{1}{0};
    \draw[Latex-Latex,shorten >= 5pt, shorten <= 5pt,very thin] (0+\x,2,1) to (1+\x,2,0);
    \draw[Latex-Latex,shorten >= 5pt, shorten <= 5pt,very thin] (0+\x,1,1) to (1+\x,1,0);

    \pgfmathsetmacro{\y}{-2}
    \pgfmathsetmacro{\x}{0}
    \node at (1.6+\x,1.5+\y,0) {$\xLongrightarrow{\tau_2'}$};

    \pgfmathsetmacro{\x}{2.7}
    \draw[dashed] (0+\x,2+\y,0) -- (0+\x,1+\y,0) (0+\x,2+\y,1) -- (0+\x,1+\y,1) (1+\x,2+\y,0) -- (1+\x,1+\y,0) (1+\x,2+\y,1) -- (1+\x,1+\y,1)
        (0+\x,2+\y,0) -- (1+\x,2+\y,0) (0+\x,2+\y,1) -- (1+\x,2+\y,1) (0+\x,1+\y,0) -- (1+\x,1+\y,0) (0+\x,1+\y,1) -- (1+\x,1+\y,1);
    \draw[dashed] (0+\x,2+\y,0) -- (0+\x,2+\y,1); \whitedot{0+\x}{2+\y}{0}; \whitedot{0+\x}{2+\y}{1}; 
    \draw[dashed] (0+\x,1+\y,0) -- (0+\x,1+\y,1); \whitedot{0+\x}{1+\y}{1}; \whitedot{0+\x}{1+\y}{0};
    \draw[dashed] (1+\x,2+\y,0) -- (1+\x,2+\y,1); \blackdot{1+\x}{2+\y}{0}; \blackdot{1+\x}{2+\y}{1}; 
    \draw[dashed] (1+\x,1+\y,0) -- (1+\x,1+\y,1); \blackdot{1+\x}{1+\y}{1}; \blackdot{1+\x}{1+\y}{0};
    \draw[Latex-Latex,shorten >= 5pt, shorten <= 5pt,very thin] (0+\x,2+\y,0) to [bend left=15] (1+\x,2+\y,0);
    \draw[Latex-Latex,shorten >= 5pt, shorten <= 5pt,very thin] (0+\x,2+\y,1) to [bend left=15] (1+\x,2+\y,1);
    \node at (1.6+\x,1.5+\y,0) {$\xLongrightarrow{\phantom{\tau_1}}$};

    \pgfmathsetmacro{\x}{5.4}
    \draw[dashed] (0+\x,2+\y,0) -- (0+\x,1+\y,0) (0+\x,2+\y,1) -- (0+\x,1+\y,1) (1+\x,2+\y,0) -- (1+\x,1+\y,0) (1+\x,2+\y,1) -- (1+\x,1+\y,1)
        (0+\x,2+\y,0) -- (1+\x,2+\y,0) (0+\x,2+\y,1) -- (1+\x,2+\y,1) (0+\x,1+\y,0) -- (1+\x,1+\y,0) (0+\x,1+\y,1) -- (1+\x,1+\y,1);
    \draw[dashed] (0+\x,2+\y,0) -- (0+\x,2+\y,1); \blackdot{0+\x}{2+\y}{0}; \blackdot{0+\x}{2+\y}{1}; 
    \draw[dashed] (0+\x,1+\y,0) -- (0+\x,1+\y,1); \whitedot{0+\x}{1+\y}{1}; \whitedot{0+\x}{1+\y}{0};
    \draw[dashed] (1+\x,2+\y,0) -- (1+\x,2+\y,1); \whitedot{1+\x}{2+\y}{0}; \whitedot{1+\x}{2+\y}{1}; 
    \draw[dashed] (1+\x,1+\y,0) -- (1+\x,1+\y,1); \blackdot{1+\x}{1+\y}{1}; \blackdot{1+\x}{1+\y}{0};
    \draw[Latex-Latex,shorten >= 5pt, shorten <= 5pt,very thin] (1+\x,1+\y,0) to (0+\x,2+\y,0);
    \draw[Latex-Latex,shorten >= 5pt, shorten <= 5pt,very thin] (1+\x,1+\y,1) to (0+\x,2+\y,1);
    \node at (1.6+\x,1.5+\y,0) {$\xLongrightarrow{\tau_3'}$};

    \pgfmathsetmacro{\x}{8.1}
    \draw[dashed] (0+\x,2+\y,0) -- (0+\x,1+\y,0) (0+\x,2+\y,1) -- (0+\x,1+\y,1) (1+\x,2+\y,0) -- (1+\x,1+\y,0) (1+\x,2+\y,1) -- (1+\x,1+\y,1)
        (0+\x,2+\y,0) -- (1+\x,2+\y,0) (0+\x,2+\y,1) -- (1+\x,2+\y,1) (0+\x,1+\y,0) -- (1+\x,1+\y,0) (0+\x,1+\y,1) -- (1+\x,1+\y,1);
    \draw[dashed] (0+\x,2+\y,0) -- (0+\x,2+\y,1); \whitedot{0+\x}{2+\y}{0}; \whitedot{0+\x}{2+\y}{1}; 
    \draw[dashed] (0+\x,1+\y,0) -- (0+\x,1+\y,1); \whitedot{0+\x}{1+\y}{1}; \whitedot{0+\x}{1+\y}{0};
    \draw[dashed] (1+\x,2+\y,0) -- (1+\x,2+\y,1); \whitedot{1+\x}{2+\y}{0}; \whitedot{1+\x}{2+\y}{1}; 
    \draw[dashed] (1+\x,1+\y,0) -- (1+\x,1+\y,1); \whitedot{1+\x}{1+\y}{1}; \whitedot{1+\x}{1+\y}{0};
\end{tikzpicture}}
    \caption{Transform two B-cards to identity where $\tau_1'$ is concurrently odd, $\tau_2',\tau_3'$ are  concurrently even.}
    \label{new1even:exam}
\end{figure}
     
    \prop{new2even} states that we can recover any even $n$-bit CRBF by $8$ concurrently even CCRBFs. 
    \begin{proposition}\label{prop:new2even}
        For $n\geq 10$, $r_1\in[n]$, $\sigma\in A_{\zoton}^{(r_1)}$ and distinct $r_2,r_3,r_4\in [n]/\{r_1\}$. There exist $\sigma_1,\sigma_4,\sigma_7\in AC^{(r_1)}_{\{0,1\}^n}$,  $,\sigma_6,\sigma_8\in AC^{(r_2)}_{\{0,1\}^n}$,  $\sigma_2,\sigma_5\in AC^{(r_3)}_{\{0,1\}^n}$,  $\sigma_3\in AC^{(r_4)}_{\{0,1\}^n}$ such that $\sigma=\sigma_1\circ\cdots\circ\sigma_8$.
    \end{proposition}
    
    Similar to the proof of \prop{new2}, here we first construct a concurrently even CCRBF $\pi$ such that $\sigma\pi$ is free of 3/5-cycle and $\sigma\pi$ has an even cycle. Then we use concurrently even CCRBFs to formulate cycles. Besides, we need to solve some special cases. Those proofs are similar to the corresponding ones and are put into the appendix. 
    Here is the new lemma for eliminating cycles.
    
    \begin{lemma}\label{e35a2}
       For $n\geq 8$, $r_1\in[n]$ and $\sigma\in A_{\zoton}$, there exists $\pi\in AC_{\zoton}^{(r_1)}$ such that $\sigma\pi$ is free of 3/5-cycle, and $\sigma\pi$ has at least an even cycle.
    \end{lemma}
    
    The additional demand for an even cycle comes from the following lemma.
    
    \begin{lemma}\label{evenh}
        For $\sigma,\pi\in S_{\{0,1\}^n}$. $\sigma,\pi$ have the same cycle pattern and $\sigma$ has an even cycle. Then there exists $h\in A_{\zoton}$ such that $h\sigma h^{-1}=\pi$.
    \end{lemma}
    
    These following 2 lemmas ensure that cycle pattern can be constructed by 2 concurrently even CCRBFs on different dimensions \textbf{under some restrictions}.
    
    \begin{lemma}\label{lem:GcyclesEven1}
        For $\sigma \in A_{\zoton}$ which is free of 3/5-cycle and contains at least 12 cycles with the length of at least $2$, there exist $\pi \in AC^{(r_1)}_{\zoton}$ and $\tau \in AC^{(r_2)}_{\zoton}$ such that $\pi\tau$ has the same cycle pattern with $\sigma$.
    \end{lemma}
    
    \begin{lemma}\label{lem:GcyclesEven2}
        For $\sigma \in A_{\zoton}$ which is free of 3/5-cycle and contains a cycle with the length of at least $12$, there exist $\pi \in AC^{(r_1)}_{\zoton}$ and $\tau \in AC^{(r_2)}_{\zoton}$ such that $\pi\tau$ has the same cycle pattern with $\sigma$.
    \end{lemma}
    
    The last preparation is to construct a concurrently odd CCRBF by $4$ concurrently even CCRBFs.
     \begin{lemma}\label{Odd4}
     For $n\geq 3$, distinct $r_1,r_2,r_3\in[n]$, there exists concurrently odd $\pi\in SC^{(r_1)}_{\{0,1\}^n}$, such that $\pi=\tau_1\tau_2\tau_3\tau_4$, where $\tau_1\in AC^{(r_3)}_{\{0,1\}^n},\tau_2,\tau_4\in AC^{(r_2)}_{\{0,1\}^n},\tau_3\in AC^{(r_1)}_{\{0,1\}^n}$.
     \end{lemma}
     Finally we give proof of \prop{new2even}.
    \begin{proof} [Proof of \prop{new2even}]
        W.l.o.g, assume $r_1 = 1, r_2 = 2$. Similar to the proof of \prop{new2}, since $\sigma\in A_{\zoton}^{(r_1)}$, there exist $f,g\in S_{\zoto{n-1}}$ such that $\sigma=\begin{bmatrix}f&\\&g\end{bmatrix}$. Observe that for any $g',s, h\in S_{\{0,1\}^{n-1}}$, let $\pi_9=\begin{bmatrix}\id&\\&g'\end{bmatrix}$, we have
        \begin{align*}
            \sigma\pi_9=&\begin{bmatrix}
                fsh & \\
                    & fsh
            \end{bmatrix}
            \begin{bmatrix}
                \id & \\
                    & h^{-1}(fs)^{-1}(gg's)h
                \end{bmatrix}\\
            &\begin{bmatrix}
                h^{-1}  & \\
                        & h^{-1}
            \end{bmatrix}
            \begin{bmatrix}
                s^{-1}  & \\
                        & s^{-1}
            \end{bmatrix}.
        \end{align*}
        
        We first use Lemma \ref{e35a2} to choose $g'\in AC^{(r_2)}_{\{0,1\}^{n-1}}$, such that $f^{-1}gg'$ is free of 3/5-cycle and has an even cycle. For convenience, we  perform another pre-processing. Technically, if $f^{-1}gg'$ has a cycle of length $\geq 12$ or has at least $12$ cycles, we do nothing. Otherwise, there are at least $13$ fix-point pairs $(\bm x_1, \bm y_1), \ldots, (\bm x_{13}, \bm y_{13})$ in $f^{-1}gg'$ satisfying $(\bm x_i)_{r_1} = (\bm y_i)_{r_1} = 1$ and $\bm x_i = \bm y_i^{\oplus r_2}$ for all $i\in [13]$ since $n\geq 10$. 
        Thus, we can perform $g''\in AC^{(r_2)}_{\{0,1\}^{n-1}}$ to add two $13$-cycles without affecting other cycle in $f^{-1}gg'$. For simplicity, we update $g'$ as $g'g''$.
        
        Since $\sigma, \pi_9$ are even, $f, gg'$ are either both even or both odd. If $f,gg'$ are both even, we choose $s=\id$. If otherwise, using Lemma \ref{Odd4}, we choose concurrently odd $\begin{bmatrix}s^{-1}&\\&s^{-1}\end{bmatrix}\in SC^{(r_1)}_{\zoton}$ where  $s^{-1}$ is odd,
        and construct it with $4$ even blocks in order $r_3,r_2,r_1,r_2$ (i.e., $\pi_5,\pi_6,\pi_7,\pi_8$). Then $fs,gg's$ will be both even. 
     
        Next we synthesize $\begin{bmatrix}\id&\\&h^{-1}(fs)^{-1}(gg's)h\end{bmatrix}$. Note that $f^{-1}gg'$ either contains at least 12 cycles, or contains a long cycle of length at least $12$. According to \lem{GcyclesEven1} and \lem{GcyclesEven2}, there exist $\tau_1 \in AC_{\zoto{n-1}}^{(r_3)}$ and $\tau_2 \in AC_{\zoto{n-1}}^{(r_4)}$ such that $\tau_1\tau_2$ has the same cycle pattern with $f^{-1}gg'$ and $(fs)^{-1}gg's$. Furthermore, since $f^{-1}gg'$ has an even cycle, by Lemma \ref{evenh}, there exists $h\in A_{\{0,1\}^{n-1}}$ such that $\tau_1\tau_2=h^{-1}(fs)^{-1}gg'sh$.
        
        To sum up, let $\pi_1=\begin{bmatrix}(fs)h&\\&(fs)h\end{bmatrix}$, $\pi_2=\begin{bmatrix}id &\\&\tau_1\end{bmatrix}$, $\pi_3=\begin{bmatrix}id &\\&\tau_2\end{bmatrix}$, $\pi_3=\begin{bmatrix}h^{-1} &\\&h^{-1}\end{bmatrix}$. Then $\pi_1,\pi_4,\pi_7\in AC^{(r_1)}_{\{0,1\}^n}$,  $\pi_6,\pi_8,\pi_9\in AC^{(r_2)}_{\{0,1\}^n}$,  $\pi_2,\pi_5\in AC^{(r_3)}_{\{0,1\}^n}$,  $\pi_3\in AC^{(r_4)}_{\{0,1\}^n}$, and
        $$
            \sigma=\pi_1\pi_2\pi_3\pi_4\pi_5\pi_6\pi_7(\pi_8\pi_9^{-1}).
        $$
    \end{proof}

\section{Conclusion and open questions}\label{sec:con}
    In our work, we offer a method to decompose arbitrary even $n$-bit reversible Boolean function (RBF) into $7$ blocks of $(n-1)$-bit RBFs for $n\geq 6$, or into $10$ blocks of even $(n-1)$-bit RBFs for $n\geq 10$,
    where the blocks have certain freedom to choose. Technically, we transform even RBF to an even controlled reversible Boolean function (CRBF) by $3$ blocks. Then we transform the even CRBF to identity by $5$ blocks.
    In addition, the last block of the first step can be merged with the first block of the second step, thus providing a $7$-depth decomposition. The road map of even block depth is similar but much more complicated.

    One direct open question is whether the constant $7$ (and $10$) can be further improved and what is the optimal constant.
    Besides, one may try to relax the conditions that $n\geq 6$ and $n\geq 10$. 
    Another interesting question is, given an even $n$-bit RBF, if we are allowed to use general unitary blocks to synthesize it, can we use strictly fewer blocks than only using RBF blocks?

    \bibliographystyle{ieeetr}
    \bibliography{reference}

\begin{thebibliography}{10}

\bibitem{bennett1988notes}
C.~H. Bennett, ``Notes on the history of reversible computation,'' {\em ibm
  Journal of Research and Development}, vol.~32, no.~1, pp.~16--23, 1988.

\bibitem{saeedi2013synthesis}
M.~Saeedi and I.~L. Markov, ``Synthesis and optimization of reversible
  circuits-a survey,'' {\em ACM Computing Surveys (CSUR)}, vol.~45, no.~2,
  p.~21, 2013.

\bibitem{arabzadeh2010rule}
M.~Arabzadeh, M.~Saeedi, and M.~S. Zamani, ``Rule-based optimization of
  reversible circuits,'' in {\em Proceedings of the 2010 Asia and South Pacific
  Design Automation Conference}, pp.~849--854, IEEE Press, 2010.

\bibitem{Landauer1961Irreversibility}
R.~Landauer, ``Irreversibility and heat generation in the computing process,''
  {\em Ibm J.res.dev}, vol.~5, no.~1.2, pp.~261--269, 1961.

\bibitem{book}
M.~A. Nielsen and I.~L. Chuang, {\em Quantum Computation and Quantum
  Information}.
\newblock Cambridge University Press, 2010.

\bibitem{shor1999polynomial}
P.~W. Shor, ``Polynomial-time algorithms for prime factorization and discrete
  logarithms on a quantum computer,'' {\em SIAM review}, vol.~41, no.~2,
  pp.~303--332, 1999.

\bibitem{Grover1996A}
L.~K. Grover, ``A fast quantum mechanical algorithm for database search,'' in
  {\em Twenty-Eighth ACM Symposium on Theory of Computing}, pp.~212--219, 1996.

\bibitem{shende2003synthesis}
V.~V. Shende, A.~K. Prasad, I.~L. Markov, and J.~P. Hayes, ``Synthesis of
  reversible logic circuits,'' {\em IEEE Transactions on Computer-Aided Design
  of Integrated Circuits and Systems}, vol.~22, no.~6, pp.~710--722, 2003.

\bibitem{maslov2008quantum}
D.~Maslov, G.~W. Dueck, D.~M. Miller, and C.~Negrevergne, ``Quantum circuit
  simplification and level compaction,'' {\em IEEE Transactions on
  Computer-Aided Design of Integrated Circuits and Systems}, vol.~27, no.~3,
  pp.~436--444, 2008.

\bibitem{takahashi2009quantum}
Y.~Takahashi, S.~Tani, and N.~Kunihiro, ``Quantum addition circuits and
  unbounded fan-out,'' {\em arXiv preprint arXiv:0910.2530}, 2009.

\bibitem{selinger2018finite}
P.~Selinger, ``A finite alternation result for reversible boolean circuits,''
  {\em Science of Computer Programming}, vol.~151, pp.~2--17, 2018.

\bibitem{saeedi2010reversible}
M.~Saeedi, M.~S. Zamani, M.~Sedighi, and Z.~Sasanian, ``Reversible circuit
  synthesis using a cycle-based approach,'' {\em ACM Journal on Emerging
  Technologies in Computing Systems (JETC)}, vol.~6, no.~4, p.~13, 2010.

\bibitem{saeedi2010block}
M.~Saeedi, M.~Arabzadeh, M.~S. Zamani, and M.~Sedighi, ``Block-based
  quantum-logic synthesis,'' {\em arXiv preprint arXiv:1011.2159}, 2010.

\bibitem{shende2006synthesis}
V.~V. Shende, S.~S. Bullock, and I.~L. Markov, ``Synthesis of quantum-logic
  circuits,'' {\em IEEE Transactions on Computer-Aided Design of Integrated
  Circuits and Systems}, vol.~25, no.~6, pp.~1000--1010, 2006.

\bibitem{devitt2016performing}
S.~J. Devitt, ``Performing quantum computing experiments in the cloud,'' {\em
  Physical Review A}, vol.~94, no.~3, p.~032329, 2016.

\bibitem{divincenzo2000physical}
D.~P. DiVincenzo, ``The physical implementation of quantum computation,'' {\em
  Fortschritte der Physik: Progress of Physics}, vol.~48, no.~9-11,
  pp.~771--783, 2000.

\bibitem{zulehner2018efficient}
A.~Zulehner, A.~Paler, and R.~Wille, ``An efficient methodology for mapping
  quantum circuits to the {IBM} {QX} architectures,'' {\em IEEE Transactions on
  Computer-Aided Design of Integrated Circuits and Systems}, 2018.

\bibitem{almudever2017engineering}
C.~G. Almudever, L.~Lao, X.~Fu, N.~Khammassi, I.~Ashraf, D.~Iorga,
  S.~Varsamopoulos, C.~Eichler, A.~Wallraff, L.~Geck, {\em et~al.}, ``The
  engineering challenges in quantum computing,'' in {\em 2017 Design,
  Automation \& Test in Europe Conference \& Exhibition (DATE)}, pp.~836--845,
  IEEE, 2017.

\bibitem{veldhorst2017silicon}
M.~Veldhorst, H.~Eenink, C.~Yang, and A.~Dzurak, ``Silicon {CMOS} architecture
  for a spin-based quantum computer,'' {\em Nature communications}, vol.~8,
  no.~1, p.~1766, 2017.

\end{thebibliography}

   \clearpage
    
    \section{Appendix}
    
    \begin{proof}[Proof of \lem{cannot}]
        Assume for contradiction there exist $\sigma_1\in SC_{\zoton}^{(r_1)},\pi_1,\pi_2\in SC_{\zoton}^{(r_2)},$ such that $\sigma\pi_1\sigma_1\pi_2\in S_{\zoton}^{(r_1)}$. Construct the black-white cuboid for $\sigma$.
        
        For the $1^{\it st}$ case, define $\eta$ as the number of \tikz[thick]{\draw[dashed] (0,0,0) -- (0,0.7,0); \blackdot{0}{0}{0}; \whitedot{0}{0.7}{0};} and \tikz[thick]{\draw[dashed] (0,0,0) -- (0,0.7,0); \whitedot{0}{0}{0}; \blackdot{0}{0.7}{0};}. It is easy to check $\eta\equiv 2\mod 4$ at the beginning and any permutation $\pi_1 \in SC_{\zoton}^{(r_2)}$ does not changes the value of $\eta\mod 4$. Note that any permutation $\sigma_1\in SC_{\zoton}^{(r_1)}$ does not change $\eta$. Thus, $\pi_2\in SC_{\zoton}^{(r_2)}$ can not transform all node to white, since it requires $\eta\equiv 0\mod 4$, which is a contradiction.
        
        For the $2^{\it nd}$ case, define $\xi$ as the number of \tikz[thick]{\draw[dashed] (0,0,0) -- (0,0.7,0); \blackdot{0}{0}{0}; \blackdot{0}{0.7}{0};}. It is easy to check $\xi$ is odd at the beginning and any permutation $\pi_1 \in SC_{\zoton}^{(r_2)}$ does not change its parity. Note that any permutation $\sigma_1\in SC_{\zoton}^{(r_1)}$ does not change $\xi$. Thus, $\pi_2\in SC_{\zoton}^{(r_2)}$ can not transform all node to white, since it requires $\xi$ is even, which is a contradiction.
    \end{proof}
    
    \begin{proof}[Proof of \lem{new1tight}]
        Let 
        $$
        \sigma_3=(000,001) (101,111) (010,110)\in S_{\{0,1\}^3},
        $$
        then define $\sigma_{k+1}$ recursively based on $\sigma_k$ and let $\sigma=\sigma_n$.
        Assume $\bm u\in\{0,1\}^k$ is a fix-point under $\tau_k$, then
        $$
        \sigma_{k+1}(x)=\begin{cases}
            0\sigma_k(\bm v), &  x=0\bm v,\bm v\neq\sigma_k(v)\\
            1\bm u, & \bm x=0\bm u\\
            0\bm u, & \bm x=1\bm u\\
            \bm x, & \it otherwise.
        \end{cases}
        $$
        Thus $\sigma_k\in S_{\{0,1\}^k}$ is the composition of $k$ disjoint swaps.
        
        We paint $\bm x\in\{0,1\}^n$ black if 
        $\sigma(\bm x)_{r_3}\neq\bm x_{r_3}$. Therefore, only two $\bm x$'s will be black
        and their coordinates are distinct only in $r_3$-th.
        Thus, w.l.o.g, assume $r_1,r_2,r_3$ are distinct. 
        Following the same notation $a_i$'s, $b_i$'s in \sec{visualizing}, we have
        $a_2=b_2=1,a_1=a_3=a_4=b_1=b_3=b_4=0$. Thus after $\tau\in SC_{\zoton}^{(r_1)}$, $a_1+b_2=b_1+a_2=1$.
        Since $\pi\in SC_{\zoton}^{(r_2)}$ will have to eliminate all black nodes, the pattern in the $r_2=0$ part should be
        the same with the $r_2=1$ part. Thus a contradiction.
    \end{proof}

    \begin{proof}[Proof of \lem{35free}]
        W.l.o.g, assume $r_1=1, r_2=2$. 
        Suppose $\sigma=\sigma_1\sigma_2$ is a $3$-cycle. 
            \begin{itemize}
                \item If $\sigma\in S_{\zoto{n}}^{(1)}$, then $\sigma_2=\sigma_1^{-1}\sigma$ must belong to $S_{\zoto{n}}^{(1)}\cap SC_{\zoto{n}}^{(2)}$, thus there exist $\tau_0,\tau_1\in S_{\zoto{n-2}}$ that for any $\bm x\in\{0,1\}^{n-2}$, $\sigma_2(0a\bm x)=0a\tau_0(\bm x)$,$\sigma_2(1a\bm x)=1a\tau_1(\bm x)$, for $a=0,1$. 

                For $\sigma_1\in SC_{\zoto{n}}^{(1)}$, there exists $g\in S_{\zoto{n-1}}$ such that for any $\bm y\in\{0,1\}^{n-1}$, $\sigma_1(a\bm y)=ag(\bm y)$. Then $\sigma(ab\bm x)=ag(b\tau_a(\bm x))$. Thus, if $\sigma$ is 3-cycle, then w.l.o.g, we can assume $\sigma(0b\bm x)=0b\bm x$, then $g(b\tau_0(\bm x))=b\bm x$ and $\sigma(1b\bm x)=1g(b\tau_1(\bm x))=1b\tau_0^{-1}\tau_1(\bm x)$. Patterns in $\{10\}\times\{0,1\}^{n-2}$ should be the same with patterns in $\{11\}\times\{0,1\}^{n-2}$. Thus patterns in the whole space can not be only a cycle., which means $\sigma$ can not be a $3$-cycle.

                \item If $\sigma\in S_{\zoton}^{(2)}$, the analysis is similar as $\sigma^{-1}=\sigma_2^{-1}\sigma_1^{-1}$.
                
                \item If $\sigma\not\in S_{\zoton}^{(1)}\cup S_{\zoton}^{(2)}$.  We prove $\sigma\sigma_2^{-1}\sigma_1^{-1}$ does not belong to $S_{\zoton}^{(1)}$ thus it can not be $\id$. Towards this, we construct a colored cuboid described in \sec{colored graph}. Then the cuboid will have $2$ black nodes. 
                
                Notice that $\sigma_1^{-1}$ does not change the number of black nodes. Thus the colored cuboid for $\sigma\sigma_2^{-1}$ is white. If we use $\eta$ to denote the number of black nodes. Then $\eta$ in the colored cuboid for $\sigma$ must satisfy $\eta\equiv 0\mod 4$, thus a contradiction.
            \end{itemize} 
            
        On the other hand, suppose $\sigma=\sigma_1\sigma_2$ is a $5$-cycle.
            \begin{itemize}
                \item If $\sigma\in S_{\zoton}^{(1)}\cup S_{\zoton}^{(2)}$, the analysis is similar.
                
                \item If $\sigma\not\in S_{\zoton}^{(1)}\cup S_{\zoton}^{(2)}$. Construct a colored cuboid and use $\eta$ to denote the number of black nodes in the cuboid. According to the definition, $\eta$ must be even. 
                If $\eta=2$, the analysis is similar. 
                
                Now we assume $\eta=4$. Since $\sigma\sigma_2^{-1}\sigma_1^{-1}=\id$ and $\sigma_1^{-1}$ does not change number of black nodes, we conclude colored cuboid for $\sigma\sigma_2^{-1}$ is white and the 4 black nodes for $\sigma$ must be $\bm x,\bm x^{\oplus 2},\tilde{\bm x},\tilde{\bm x}^{\oplus2}$ for some $\bm x,\tilde{\bm x},\bm x_1\neq \tilde{\bm x}_1, \bm x_2=\tilde{\bm x_2}$. W.l.o.g, we assume the fifth element in the $5$-cycle to be $\bm z$ where $\bm z_1=\bm x_1$. 
    
                Also, we can assume the relative position of the black nodes in the cycle is $\bm x,\tilde{\bm x},\bm x^{\oplus 2},\tilde{\bm x}^{\oplus 2}$ or $\bm x,\tilde{\bm x}^{\oplus 2},\bm x^{\oplus 2},\tilde{\bm x}$.  Let $\pi=(\bm x,\tilde{\bm x})(\bm x^{\oplus 2},\bm \tilde{\bm x}^{\oplus 2})\in SC_{\zoto{n}}^{(2)}$. By checking all possible arrangement of $\bm z$, we have the following cases:
                \begin{itemize}
                    \item $\sigma=(\bm x,\tilde{\bm x},\bm x^{\oplus 2},\tilde{\bm x}^{\oplus 2},\bm z)$ .  Then $\sigma\pi=(\bm x,\bm x^{\oplus 2},\bm z)=\sigma_1(\sigma_2\pi)$, which is impossible.
            
                    \item $\sigma=(\bm x,\tilde{\bm x}^{\oplus 2},\bm x^{\oplus 2},\tilde{\bm x},\bm z)$. Then $\sigma\pi=(\bm x,\bm z)(\tilde{\bm x},\tilde{ \bm x}^{\oplus 2})=\sigma_1(\sigma_2\pi)$. Construct the colored cuboid for $\sigma\pi$ with $r_1,r_2$ swapped and let $\xi$ be the number of 
                    \tikz[thick]{\draw[dashed] (0,0,0) -- (0,0.7,0); \blackdot{0}{0}{0}; \blackdot{0}{0.7}{0};}. Then $\xi=1$. Since $(\sigma_2\pi)^{-1}$ does not change $\xi$, $(\sigma\pi)(\sigma_2\pi)^{-1}\sigma_1^{-1}=\id$ requires $\xi\equiv 0\mod 2$, thus a contradiction.
                \end{itemize}
            \end{itemize}
    \end{proof}

\begin{proof}[Proof of \prop{new1even} ]
    In the following, we transformed paired cards to identity by  CCRBFs where $\tau_1$ has the different concurrently parity of the original construction in \lem{goodcase}. $\tau_2,\tau_3$ are concurrently even. Whether we use concurrently odd or even $\tau_1$ depends on the concurrently parity of $\pi'$, which is constructed for modifying cycles in Lemma \ref{e35a2}.
 
    Notice that,  whether we use the even or concurrently odd construction of $\tau_1$ does not change the  resulted cuboid, thus does not influence the following modification of $\tau_2,\tau_3$.
    
    The constructions are as below. First, we give the new construction which transforms  two A-cards or two B-cards into white cube.
        
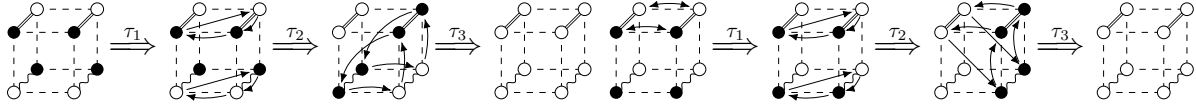
\begin{figure}[ht]                  
            \centering
            \scalebox{0.8}{\tikz{
    \Large
    \draw[dashed] (0,2,0) -- (0,1,0) (0,2,1) -- (0,1,1) (1,2,0) -- (1,1,0) (1,2,1) -- (1,1,1)
        (0,2,0) -- (1,2,0) (0,2,1) -- (1,2,1) (0,1,0) -- (1,1,0) (0,1,1) -- (1,1,1);
    \draw[double] (0,2,0) -- (0,2,1); \whitedot{0}{2}{0}; \blackdot{0}{2}{1}; 
    \draw[zig] (0,1,0) -- (0,1,1); \blackdot{0}{1}{0}; \whitedot{0}{1}{1};
    \draw[double] (1,2,0) -- (1,2,1); \whitedot{1}{2}{0}; \blackdot{1}{2}{1}; 
    \draw[zig] (1,1,0) -- (1,1,1); \blackdot{1}{1}{0}; \whitedot{1}{1}{1};

    \node at (1.6,1.5,0) {$\xLongrightarrow{\tau_1}$};
    \draw[dashed] (2.7,2,0) -- (2.7,1,0) (2.7,2,1) -- (2.7,1,1) (3.7,2,0) -- (3.7,1,0) (3.7,2,1) -- (3.7,1,1)
        (2.7,2,0) -- (3.7,2,0) (2.7,2,1) -- (3.7,2,1) (2.7,1,0) -- (3.7,1,0) (2.7,1,1) -- (3.7,1,1);
    \draw[double] (2.7,2,0) -- (2.7,2,1); \whitedot{2.7}{2}{0}; \blackdot{2.7}{2}{1}; 
    \draw[zig] (2.7,1,0) -- (2.7,1,1); \blackdot{2.7}{1}{0}; \whitedot{2.7}{1}{1};
    \draw[double] (3.7,2,0) -- (3.7,2,1); \whitedot{3.7}{2}{0}; \blackdot{3.7}{2}{1}; 
    \draw[zig] (3.7,1,0) -- (3.7,1,1); \blackdot{3.7}{1}{0}; \whitedot{3.7}{1}{1};
    \draw[-Latex,shorten >= 5pt, shorten <= 5pt,very thin] (2.7,2,1) -- (3.7,2,0);
    \draw[-Latex,shorten >= 3.5pt, shorten <= 3.5pt,very thin] (3.7,2,0) to [bend left] (3.7,2,1);
    \draw[-Latex,shorten >= 5pt, shorten <= 5pt,very thin] (3.7,2,1) to [bend left=15] (2.7,2,1);
    \draw[-Latex,shorten >= 5pt, shorten <= 5pt,very thin] (2.7,1,1) -- (3.7,1,0);
    \draw[-Latex,shorten >= 3.5pt, shorten <= 3.5pt,very thin] (3.7,1,0) to [bend left] (3.7,1,1);
    \draw[-Latex,shorten >= 5pt, shorten <= 5pt,very thin] (3.7,1,1) to [bend left=15] (2.7,1,1);
    
    \node at (4.3,1.5,0) {$\xLongrightarrow{\tau_2}$};
    \draw[dashed] (5.4,2,0) -- (5.4,1,0) (5.4,2,1) -- (5.4,1,1) (6.4,2,0) -- (6.4,1,0) (6.4,2,1) -- (6.4,1,1)
        (5.4,2,0) -- (6.4,2,0) (5.4,2,1) -- (6.4,2,1) (5.4,1,0) -- (6.4,1,0) (5.4,1,1) -- (6.4,1,1);
    \draw[double] (5.4,2,0) -- (5.4,2,1); \whitedot{5.4}{2}{0}; \whitedot{5.4}{2}{1}; 
    \draw[zig] (5.4,1,0) -- (5.4,1,1); \blackdot{5.4}{1}{0}; \blackdot{5.4}{1}{1};
    \draw[double] (6.4,2,0) -- (6.4,2,1); \blackdot{6.4}{2}{0}; \blackdot{6.4}{2}{1}; 
    \draw[zig] (6.4,1,0) -- (6.4,1,1); \whitedot{6.4}{1}{0}; \whitedot{6.4}{1}{1};
    \draw[Latex-,shorten >= 5pt, shorten <= 5pt,very thin] (5.4,1,0) to [bend left] (6.4,2,0);
    \draw[Latex-,shorten >= 5pt, shorten <= 5pt,very thin] (6.4,2,0) to [bend left=15] (6.4,1,0);
    \draw[Latex-,shorten >= 5pt, shorten <= 4pt,very thin] (6.4,1,0) to [bend right=15] (5.4,1,0);
    \draw[Latex-,shorten >= 5pt, shorten <= 5pt,very thin] (5.4,1,1) to [bend left] (6.4,2,1);
    \draw[Latex-,shorten >= 5pt, shorten <= 5pt,very thin] (6.4,2,1) to [bend left=15] (6.4,1,1);
    \draw[Latex-,shorten >= 5pt, shorten <= 4pt,very thin] (6.4,1,1) to [bend right=15] (5.4,1,1);
    
    \node at (7,1.5,0) {$\xLongrightarrow{\tau_3}$};
    \draw[dashed] (8.1,2,0) -- (8.1,1,0) (8.1,2,1) -- (8.1,1,1) (9.1,2,0) -- (9.1,1,0) (9.1,2,1) -- (9.1,1,1)
        (8.1,2,0) -- (9.1,2,0) (8.1,2,1) -- (9.1,2,1) (8.1,1,0) -- (9.1,1,0) (8.1,1,1) -- (9.1,1,1);
    \draw[double] (8.1,2,0) -- (8.1,2,1); \whitedot{8.1}{2}{0}; \whitedot{8.1}{2}{1}; 
    \draw[zig] (8.1,1,0) -- (8.1,1,1); \whitedot{8.1}{1}{0}; \whitedot{8.1}{1}{1};
    \draw[double] (9.1,2,0) -- (9.1,2,1); \whitedot{9.1}{2}{0}; \whitedot{9.1}{2}{1}; 
    \draw[zig] (9.1,1,0) -- (9.1,1,1); \whitedot{9.1}{1}{0}; \whitedot{9.1}{1}{1};
}}  
            \scalebox{0.8}{\tikz{
    \Large
    \draw[dashed] (0,2,0) -- (0,1,0) (0,2,1) -- (0,1,1) (1,2,0) -- (1,1,0) (1,2,1) -- (1,1,1)
        (0,2,0) -- (1,2,0) (0,2,1) -- (1,2,1) (0,1,0) -- (1,1,0) (0,1,1) -- (1,1,1);
    \draw[double] (0,2,0) -- (0,2,1); \whitedot{0}{2}{0}; \blackdot{0}{2}{1}; 
    \draw[zig] (0,1,0) -- (0,1,1); \blackdot{0}{1}{1}; \whitedot{0}{1}{0};
    \draw[double] (1,2,0) -- (1,2,1); \whitedot{1}{2}{0}; \blackdot{1}{2}{1}; 
    \draw[zig] (1,1,0) -- (1,1,1); \blackdot{1}{1}{1}; \whitedot{1}{1}{0};
    \draw[Latex-Latex,shorten >= 5pt, shorten <= 5pt,very thin] (0,2,0) to [bend left=15] (1,2,0);
    \draw[Latex-Latex,shorten >= 5pt, shorten <= 5pt,very thin] (0,2,1) to [bend left=15] (1,2,1);

    \node at (1.6,1.5,0) {$\xLongrightarrow{\tau_1}$};
    \draw[dashed] (2.7,2,0) -- (2.7,1,0) (2.7,2,1) -- (2.7,1,1) (3.7,2,0) -- (3.7,1,0) (3.7,2,1) -- (3.7,1,1)
        (2.7,2,0) -- (3.7,2,0) (2.7,2,1) -- (3.7,2,1) (2.7,1,0) -- (3.7,1,0) (2.7,1,1) -- (3.7,1,1);
    \draw[double] (2.7,2,0) -- (2.7,2,1); \whitedot{2.7}{2}{0}; \blackdot{2.7}{2}{1}; 
    \draw[zig] (2.7,1,0) -- (2.7,1,1); \blackdot{2.7}{1}{1}; \whitedot{2.7}{1}{0};
    \draw[double] (3.7,2,0) -- (3.7,2,1); \whitedot{3.7}{2}{0}; \blackdot{3.7}{2}{1}; 
    \draw[zig] (3.7,1,0) -- (3.7,1,1); \blackdot{3.7}{1}{1}; \whitedot{3.7}{1}{0};
    \draw[-Latex,shorten >= 5pt, shorten <= 5pt,very thin] (2.7,2,1) -- (3.7,2,0);
    \draw[-Latex,shorten >= 3.5pt, shorten <= 3.5pt,very thin] (3.7,2,0) to [bend left] (3.7,2,1);
    \draw[-Latex,shorten >= 5pt, shorten <= 5pt,very thin] (3.7,2,1) to [bend left=15] (2.7,2,1);
    \draw[-Latex,shorten >= 5pt, shorten <= 5pt,very thin] (2.7,1,1) -- (3.7,1,0);
    \draw[-Latex,shorten >= 3.5pt, shorten <= 3.5pt,very thin] (3.7,1,0) to [bend left] (3.7,1,1);
    \draw[-Latex,shorten >= 5pt, shorten <= 5pt,very thin] (3.7,1,1) to [bend left=15] (2.7,1,1);
    
    \node at (4.3,1.5,0) {$\xLongrightarrow{\tau_2}$};
    \draw[dashed] (5.4,2,0) -- (5.4,1,0) (5.4,2,1) -- (5.4,1,1) (6.4,2,0) -- (6.4,1,0) (6.4,2,1) -- (6.4,1,1)
        (5.4,2,0) -- (6.4,2,0) (5.4,2,1) -- (6.4,2,1) (5.4,1,0) -- (6.4,1,0) (5.4,1,1) -- (6.4,1,1);
    \draw[double] (5.4,2,0) -- (5.4,2,1); \whitedot{5.4}{2}{0}; \whitedot{5.4}{2}{1}; 
    \draw[zig] (5.4,1,0) -- (5.4,1,1); \whitedot{5.4}{1}{0}; \whitedot{5.4}{1}{1};
    \draw[double] (6.4,2,0) -- (6.4,2,1); \blackdot{6.4}{2}{0}; \blackdot{6.4}{2}{1}; 
    \draw[zig] (6.4,1,0) -- (6.4,1,1); \blackdot{6.4}{1}{0}; \blackdot{6.4}{1}{1};
    \draw[-Latex,shorten >= 5pt, shorten <= 5pt,very thin,yshift=-0.3mm] (6.4,1,0) to [bend left] (6.4,2,0);
    \draw[-Latex,shorten >= 5pt, shorten <= 5pt,very thin] (6.4,2,0) to [bend right=15] (5.4,2,0);
    \draw[-Latex,shorten >= 3.5pt, shorten <= 5pt,very thin] (5.4,2,0) to (6.4,1,0);
    \draw[-Latex,shorten >= 5pt, shorten <= 5pt,very thin,yshift=-0.3mm] (6.4,1,1) to [bend left] (6.4,2,1);
    \draw[-Latex,shorten >= 5pt, shorten <= 5pt,very thin] (6.4,2,1) to [bend right=15] (5.4,2,1);
    \draw[-Latex,shorten >= 3.5pt, shorten <= 5pt,very thin] (5.4,2,1) to (6.4,1,1);
    
    \node at (7,1.5,0) {$\xLongrightarrow{\tau_3}$};
    \draw[dashed] (8.1,2,0) -- (8.1,1,0) (8.1,2,1) -- (8.1,1,1) (9.1,2,0) -- (9.1,1,0) (9.1,2,1) -- (9.1,1,1)
        (8.1,2,0) -- (9.1,2,0) (8.1,2,1) -- (9.1,2,1) (8.1,1,0) -- (9.1,1,0) (8.1,1,1) -- (9.1,1,1);
    \draw[double] (8.1,2,0) -- (8.1,2,1); \whitedot{8.1}{2}{0}; \whitedot{8.1}{2}{1}; 
    \draw[zig] (8.1,1,0) -- (8.1,1,1); \whitedot{8.1}{1}{0}; \whitedot{8.1}{1}{1};
    \draw[double] (9.1,2,0) -- (9.1,2,1); \whitedot{9.1}{2}{0}; \whitedot{9.1}{2}{1}; 
    \draw[zig] (9.1,1,0) -- (9.1,1,1); \whitedot{9.1}{1}{0}; \whitedot{9.1}{1}{1};
}}
        \end{figure}
        
       The left cases can be modified to

        \begin{itemize}
            \item $\alpha=1,\beta=1,\gamma\geq 2$ : This graph shows how to tackle $1$ A-card and $1$ B-card with $2$ C-cards.
                \begin{center}
                    \scalebox{0.8}
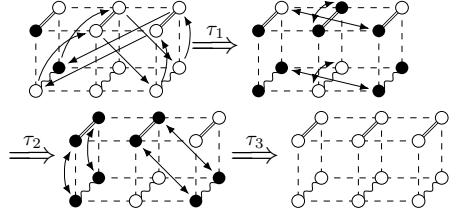
{\tikz{\Large
    \draw[dashed] (0,1,0) -- (0,0,0) (0,1,1) -- (0,0,1) (1,1,0) -- (1,0,0) (1,1,1) -- (1,0,1)
        (2,1,0) -- (2,0,0) (2,1,1) -- (2,0,1)
        (0,1,0) -- (1,1,0) -- (2,1,0) (0,1,1) -- (1,1,1) -- (2,1,1)
        (0,0,0) -- (1,0,0) -- (2,0,0) (0,0,1) -- (1,0,1) -- (2,0,1);
    \draw[double] (0,1,0) -- (0,1,1); \whitedot{0}{1}{0}; \blackdot{0}{1}{1}; 
    \draw[zig] (0,0,0) -- (0,0,1); \blackdot{0}{0}{0}; \whitedot{0}{0}{1};
    \draw[double] (1,1,0) -- (1,1,1); \whitedot{1}{1}{0}; \whitedot{1}{1}{1}; 
    \draw[zig] (1,0,0) -- (1,0,1); \whitedot{1}{0}{0}; \whitedot{1}{0}{1};
    \draw[double] (2,1,0) -- (2,1,1); \whitedot{2}{1}{0}; \whitedot{2}{1}{1}; 
    \draw[zig] (2,0,0) -- (2,0,1); \whitedot{2}{0}{0}; \whitedot{2}{0}{1};

    \foreach \z in {0,1}{
        \draw[-Latex, shorten >= 5pt, shorten <= 5pt, very thin] (0,0,\z) to [bend left] (1,1,\z);
        \draw[-Latex, shorten >= 5pt, shorten <= 5pt, very thin] (1,1,\z) to (2,0,\z);
        \draw[-Latex, shorten >= 5pt, shorten <= 5pt, very thin] (2,0,\z) to [bend right] (2,1,\z);
        \draw[-Latex, shorten >= 5pt, shorten <= 5pt, very thin] (2,1,\z) to (0,0,\z);
    }

    \node at (2.6,0.5,0) {$\xLongrightarrow{\tau_1}$};

    \draw[dashed] (3.7,1,0) -- (3.7,0,0) (3.7,1,1) -- (3.7,0,1) (4.7,1,0) -- (4.7,0,0) (4.7,1,1) -- (4.7,0,1)
        (5.7,1,0) -- (5.7,0,0) (5.7,1,1) -- (5.7,0,1) 
        (3.7,1,0) -- (4.7,1,0) -- (5.7,1,0) (3.7,1,1) -- (4.7,1,1) -- (5.7,1,1)
        (3.7,0,0) -- (4.7,0,0) -- (5.7,0,0) (3.7,0,1) -- (4.7,0,1) -- (5.7,0,1);
    \draw[double] (3.7,1,0) -- (3.7,1,1); \whitedot{3.7}{1}{0}; \blackdot{3.7}{1}{1}; 
    \draw[zig] (3.7,0,0) -- (3.7,0,1); \blackdot{3.7}{0}{0}; \blackdot{3.7}{0}{1};
    \draw[double] (4.7,1,0) -- (4.7,1,1); \blackdot{4.7}{1}{0}; \blackdot{4.7}{1}{1}; 
    \draw[zig] (4.7,0,0) -- (4.7,0,1); \whitedot{4.7}{0}{0}; \whitedot{4.7}{0}{1};
    \draw[double] (5.7,1,0) -- (5.7,1,1); \whitedot{5.7}{1}{0}; \blackdot{5.7}{1}{1}; 
    \draw[zig] (5.7,0,0) -- (5.7,0,1); \blackdot{5.7}{0}{0}; \blackdot{5.7}{0}{1};

    \foreach \y in {0,1}{
        \draw[Latex-Latex,shorten >= 5pt, shorten <= 5pt,very thin] (3.7,\y,0) to (5.7,\y,1);
        \draw[Latex-Latex,shorten >= 1pt, shorten <= 1pt,very thin, xshift=-1mm, yshift=1mm] (4.7,\y,0) to [bend right] (4.7,\y,1);
    }
}}
\vspace{5pt}\\
\scalebox{0.8}{\tikz{\Large
    \node at (-1.1,0.5,0) {$\xLongrightarrow{\tau_2}$};

    \draw[dashed] (0,1,0) -- (0,0,0) (0,1,1) -- (0,0,1) (1,1,0) -- (1,0,0) (1,1,1) -- (1,0,1)
        (2,1,0) -- (2,0,0) (2,1,1) -- (2,0,1) 
        (0,1,0) -- (1,1,0) -- (2,1,0) (0,1,1) -- (1,1,1) -- (2,1,1)
        (0,0,0) -- (1,0,0) -- (2,0,0) (0,0,1) -- (1,0,1) -- (2,0,1);
    \draw[double] (0,1,0) -- (0,1,1); \blackdot{0}{1}{0}; \blackdot{0}{1}{1}; 
    \draw[zig] (0,0,0) -- (0,0,1); \blackdot{0}{0}{0}; \blackdot{0}{0}{1};
    \draw[double] (1,1,0) -- (1,1,1); \blackdot{1}{1}{0}; \blackdot{1}{1}{1}; 
    \draw[zig] (1,0,0) -- (1,0,1); \whitedot{1}{0}{0}; \whitedot{1}{0}{1};
    \draw[double] (2,1,0) -- (2,1,1); \whitedot{2}{1}{0}; \whitedot{2}{1}{1}; 
    \draw[zig] (2,0,0) -- (2,0,1); \blackdot{2}{0}{0}; \blackdot{2}{0}{1};

    \foreach \z in {0,1}{
        \draw[Latex-Latex,shorten >= 5pt, shorten <= 5pt,very thin,yshift=-0.3mm] (0,0,\z) to [bend left] (0,1,\z);
        \draw[Latex-Latex,shorten >= 5pt, shorten <= 5pt,very thin] (1,1,\z) to (2,0,\z);
    }

    \node at (2.6,0.5,0) {$\xLongrightarrow{\tau_3}$};

    \draw[dashed] (3.7,1,0) -- (3.7,0,0) (3.7,1,1) -- (3.7,0,1) (4.7,1,0) -- (4.7,0,0) (4.7,1,1) -- (4.7,0,1)
        (5.7,1,0) -- (5.7,0,0) (5.7,1,1) -- (5.7,0,1)
        (3.7,1,0) -- (4.7,1,0) -- (5.7,1,0) (3.7,1,1) -- (4.7,1,1) -- (5.7,1,1) 
        (3.7,0,0) -- (4.7,0,0) -- (5.7,0,0) (3.7,0,1) -- (4.7,0,1) -- (5.7,0,1);
    \draw[double] (3.7,1,0) -- (3.7,1,1); \whitedot{3.7}{1}{0}; \whitedot{3.7}{1}{1}; 
    \draw[zig] (3.7,0,0) -- (3.7,0,1); \whitedot{3.7}{0}{0}; \whitedot{3.7}{0}{1};
    \draw[double] (4.7,1,0) -- (4.7,1,1); \whitedot{4.7}{1}{0}; \whitedot{4.7}{1}{1}; 
    \draw[zig] (4.7,0,0) -- (4.7,0,1); \whitedot{4.7}{0}{0}; \whitedot{4.7}{0}{1};
    \draw[double] (5.7,1,0) -- (5.7,1,1); \whitedot{5.7}{1}{0}; \whitedot{5.7}{1}{1}; 
    \draw[zig] (5.7,0,0) -- (5.7,0,1); \whitedot{5.7}{0}{0}; \whitedot{5.7}{0}{1};
}}
                \end{center}
            \item $\alpha=1,\beta=0,\gamma\geq 2$ : This graph shows how to tackle $1$ A-card with $2$ C-cards.
                \begin{center}
                    \scalebox{0.8}{\tikz{\Large
    \draw[dashed] (0,1,0) -- (0,0,0) (0,1,1) -- (0,0,1) (1,1,0) -- (1,0,0) (1,1,1) -- (1,0,1)
        (2,1,0) -- (2,0,0) (2,1,1) -- (2,0,1)
        (0,1,0) -- (1,1,0) -- (2,1,0) (0,1,1) -- (1,1,1) -- (2,1,1)
        (0,0,0) -- (1,0,0) -- (2,0,0) (0,0,1) -- (1,0,1) -- (2,0,1);
    \draw[double] (0,1,0) -- (0,1,1); \whitedot{0}{1}{0}; \blackdot{0}{1}{1}; 
    \draw[zig] (0,0,0) -- (0,0,1); \blackdot{0}{0}{0}; \whitedot{0}{0}{1};
    \draw[double] (1,1,0) -- (1,1,1); \whitedot{1}{1}{0}; \whitedot{1}{1}{1}; 
    \draw[zig] (1,0,0) -- (1,0,1); \whitedot{1}{0}{0}; \whitedot{1}{0}{1};
    \draw[double] (2,1,0) -- (2,1,1); \whitedot{2}{1}{0}; \whitedot{2}{1}{1}; 
    \draw[zig] (2,0,0) -- (2,0,1); \whitedot{2}{0}{0}; \whitedot{2}{0}{1};

    \foreach \z in {0,1}{
        \draw[-Latex, shorten >= 5pt, shorten <= 5pt, very thin] (0,0,\z) to [bend left] (1,1,\z);
        \draw[-Latex, shorten >= 5pt, shorten <= 5pt, very thin] (1,1,\z) to (2,0,\z);
        \draw[-Latex, shorten >= 5pt, shorten <= 5pt, very thin] (2,0,\z) to [bend right] (2,1,\z);
        \draw[-Latex, shorten >= 5pt, shorten <= 5pt, very thin] (2,1,\z) to (0,0,\z);
    }

    \node at (2.6,0.5,0) {$\xLongrightarrow{\tau_1}$};

    \draw[dashed] (3.7,1,0) -- (3.7,0,0) (3.7,1,1) -- (3.7,0,1) (4.7,1,0) -- (4.7,0,0) (4.7,1,1) -- (4.7,0,1)
        (5.7,1,0) -- (5.7,0,0) (5.7,1,1) -- (5.7,0,1) 
        (3.7,1,0) -- (4.7,1,0) -- (5.7,1,0) (3.7,1,1) -- (4.7,1,1) -- (5.7,1,1)
        (3.7,0,0) -- (4.7,0,0) -- (5.7,0,0) (3.7,0,1) -- (4.7,0,1) -- (5.7,0,1);
    \draw[double] (3.7,1,0) -- (3.7,1,1); \whitedot{3.7}{1}{0}; \blackdot{3.7}{1}{1}; 
    \draw[zig] (3.7,0,0) -- (3.7,0,1); \blackdot{3.7}{0}{0}; \blackdot{3.7}{0}{1};
    \draw[double] (4.7,1,0) -- (4.7,1,1); \blackdot{4.7}{1}{0}; \blackdot{4.7}{1}{1}; 
    \draw[zig] (4.7,0,0) -- (4.7,0,1); \whitedot{4.7}{0}{0}; \whitedot{4.7}{0}{1};
    \draw[double] (5.7,1,0) -- (5.7,1,1); \whitedot{5.7}{1}{0}; \blackdot{5.7}{1}{1}; 
    \draw[zig] (5.7,0,0) -- (5.7,0,1); \blackdot{5.7}{0}{0}; \blackdot{5.7}{0}{1};

    \foreach \y in {0,1}{
        \draw[Latex-Latex,shorten >= 5pt, shorten <= 5pt,very thin] (3.7,\y,0) to (5.7,\y,1);
        \draw[Latex-Latex,shorten >= 3pt, shorten <= 3pt,very thin] (4.7,\y,0) to [bend right] (4.7,\y,1);
    }
}}
\vspace{5pt}\\
\scalebox{0.8}{\tikz{\Large
    \node at (-1.1,0.5,0) {$\xLongrightarrow{\tau_2}$};

    \draw[dashed] (0,1,0) -- (0,0,0) (0,1,1) -- (0,0,1) (1,1,0) -- (1,0,0) (1,1,1) -- (1,0,1)
        (2,1,0) -- (2,0,0) (2,1,1) -- (2,0,1) 
        (0,1,0) -- (1,1,0) -- (2,1,0) (0,1,1) -- (1,1,1) -- (2,1,1)
        (0,0,0) -- (1,0,0) -- (2,0,0) (0,0,1) -- (1,0,1) -- (2,0,1);
    \draw[double] (0,1,0) -- (0,1,1); \blackdot{0}{1}{0}; \blackdot{0}{1}{1}; 
    \draw[zig] (0,0,0) -- (0,0,1); \blackdot{0}{0}{0}; \blackdot{0}{0}{1};
    \draw[double] (1,1,0) -- (1,1,1); \blackdot{1}{1}{0}; \blackdot{1}{1}{1}; 
    \draw[zig] (1,0,0) -- (1,0,1); \whitedot{1}{0}{0}; \whitedot{1}{0}{1};
    \draw[double] (2,1,0) -- (2,1,1); \whitedot{2}{1}{0}; \whitedot{2}{1}{1}; 
    \draw[zig] (2,0,0) -- (2,0,1); \blackdot{2}{0}{0}; \blackdot{2}{0}{1};

    \foreach \z in {0,1}{
        \draw[Latex-Latex,shorten >= 5pt, shorten <= 5pt,very thin,yshift=-0.3mm] (0,0,\z) to [bend left] (0,1,\z);
        \draw[Latex-Latex,shorten >= 5pt, shorten <= 5pt,very thin] (1,1,\z) to (2,0,\z);
    }

    \node at (2.6,0.5,0) {$\xLongrightarrow{\tau_3}$};

    \draw[dashed] (3.7,1,0) -- (3.7,0,0) (3.7,1,1) -- (3.7,0,1) (4.7,1,0) -- (4.7,0,0) (4.7,1,1) -- (4.7,0,1)
        (5.7,1,0) -- (5.7,0,0) (5.7,1,1) -- (5.7,0,1)
        (3.7,1,0) -- (4.7,1,0) -- (5.7,1,0) (3.7,1,1) -- (4.7,1,1) -- (5.7,1,1) 
        (3.7,0,0) -- (4.7,0,0) -- (5.7,0,0) (3.7,0,1) -- (4.7,0,1) -- (5.7,0,1);
    \draw[double] (3.7,1,0) -- (3.7,1,1); \whitedot{3.7}{1}{0}; \whitedot{3.7}{1}{1}; 
    \draw[zig] (3.7,0,0) -- (3.7,0,1); \whitedot{3.7}{0}{0}; \whitedot{3.7}{0}{1};
    \draw[double] (4.7,1,0) -- (4.7,1,1); \whitedot{4.7}{1}{0}; \whitedot{4.7}{1}{1}; 
    \draw[zig] (4.7,0,0) -- (4.7,0,1); \whitedot{4.7}{0}{0}; \whitedot{4.7}{0}{1};
    \draw[double] (5.7,1,0) -- (5.7,1,1); \whitedot{5.7}{1}{0}; \whitedot{5.7}{1}{1}; 
    \draw[zig] (5.7,0,0) -- (5.7,0,1); \whitedot{5.7}{0}{0}; \whitedot{5.7}{0}{1};
}}
                \end{center}
            \item $\alpha=0,\beta=1,\gamma\geq 2$ : This graph shows how to tackle $1$ B-card with $2$ C-card.
                \begin{center}                    
                    \scalebox{0.8}
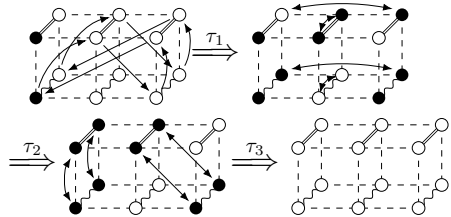
{\tikz{\Large
    \draw[dashed] (0,1,0) -- (0,0,0) (0,1,1) -- (0,0,1) (1,1,0) -- (1,0,0) (1,1,1) -- (1,0,1)
        (2,1,0) -- (2,0,0) (2,1,1) -- (2,0,1)
        (0,1,0) -- (1,1,0) -- (2,1,0) (0,1,1) -- (1,1,1) -- (2,1,1)
        (0,0,0) -- (1,0,0) -- (2,0,0) (0,0,1) -- (1,0,1) -- (2,0,1);
    \draw[double] (0,1,0) -- (0,1,1); \whitedot{0}{1}{0}; \blackdot{0}{1}{1}; 
    \draw[zig] (0,0,0) -- (0,0,1); \whitedot{0}{0}{0}; \blackdot{0}{0}{1};
    \draw[double] (1,1,0) -- (1,1,1); \whitedot{1}{1}{0}; \whitedot{1}{1}{1}; 
    \draw[zig] (1,0,0) -- (1,0,1); \whitedot{1}{0}{0}; \whitedot{1}{0}{1};
    \draw[double] (2,1,0) -- (2,1,1); \whitedot{2}{1}{0}; \whitedot{2}{1}{1}; 
    \draw[zig] (2,0,0) -- (2,0,1); \whitedot{2}{0}{0}; \whitedot{2}{0}{1};

    \foreach \z in {0,1}{
        \draw[-Latex, shorten >= 5pt, shorten <= 5pt, very thin] (0,0,\z) to [bend left] (1,1,\z);
        \draw[-Latex, shorten >= 5pt, shorten <= 5pt, very thin] (1,1,\z) to (2,0,\z);
        \draw[-Latex, shorten >= 5pt, shorten <= 5pt, very thin] (2,0,\z) to [bend right] (2,1,\z);
        \draw[-Latex, shorten >= 5pt, shorten <= 5pt, very thin] (2,1,\z) to (0,0,\z);
    }

    \node at (2.6,0.5,0) {$\xLongrightarrow{\tau_1}$};

    \draw[dashed] (3.7,1,0) -- (3.7,0,0) (3.7,1,1) -- (3.7,0,1) (4.7,1,0) -- (4.7,0,0) (4.7,1,1) -- (4.7,0,1)
        (5.7,1,0) -- (5.7,0,0) (5.7,1,1) -- (5.7,0,1) 
        (3.7,1,0) -- (4.7,1,0) -- (5.7,1,0) (3.7,1,1) -- (4.7,1,1) -- (5.7,1,1)
        (3.7,0,0) -- (4.7,0,0) -- (5.7,0,0) (3.7,0,1) -- (4.7,0,1) -- (5.7,0,1);
    \draw[double] (3.7,1,0) -- (3.7,1,1); \whitedot{3.7}{1}{0}; \blackdot{3.7}{1}{1}; 
    \draw[zig] (3.7,0,0) -- (3.7,0,1); \blackdot{3.7}{0}{0}; \blackdot{3.7}{0}{1};
    \draw[double] (4.7,1,0) -- (4.7,1,1); \blackdot{4.7}{1}{0}; \blackdot{4.7}{1}{1}; 
    \draw[zig] (4.7,0,0) -- (4.7,0,1); \whitedot{4.7}{0}{0}; \whitedot{4.7}{0}{1};
    \draw[double] (5.7,1,0) -- (5.7,1,1); \blackdot{5.7}{1}{0}; \whitedot{5.7}{1}{1}; 
    \draw[zig] (5.7,0,0) -- (5.7,0,1); \blackdot{5.7}{0}{0}; \blackdot{5.7}{0}{1};

    \foreach \y in {0,1}{
        \draw[Latex-Latex,shorten >= 5pt, shorten <= 5pt,very thin,yshift=0.7mm] (3.7,\y,0) to [bend left=10] (5.7,\y,0);
        \draw[Latex-Latex,shorten >= 3pt, shorten <= 3pt,very thin] (4.7,\y,0) to [bend right] (4.7,\y,1);
    }
}}
\vspace{5pt}\\
\scalebox{0.8}{\tikz{\Large
    \node at (-1.1,0.5,0) {$\xLongrightarrow{\tau_2}$};

    \draw[dashed] (0,1,0) -- (0,0,0) (0,1,1) -- (0,0,1) (1,1,0) -- (1,0,0) (1,1,1) -- (1,0,1)
        (2,1,0) -- (2,0,0) (2,1,1) -- (2,0,1) 
        (0,1,0) -- (1,1,0) -- (2,1,0) (0,1,1) -- (1,1,1) -- (2,1,1)
        (0,0,0) -- (1,0,0) -- (2,0,0) (0,0,1) -- (1,0,1) -- (2,0,1);
    \draw[double] (0,1,0) -- (0,1,1); \blackdot{0}{1}{0}; \blackdot{0}{1}{1}; 
    \draw[zig] (0,0,0) -- (0,0,1); \blackdot{0}{0}{0}; \blackdot{0}{0}{1};
    \draw[double] (1,1,0) -- (1,1,1); \blackdot{1}{1}{0}; \blackdot{1}{1}{1}; 
    \draw[zig] (1,0,0) -- (1,0,1); \whitedot{1}{0}{0}; \whitedot{1}{0}{1};
    \draw[double] (2,1,0) -- (2,1,1); \whitedot{2}{1}{0}; \whitedot{2}{1}{1}; 
    \draw[zig] (2,0,0) -- (2,0,1); \blackdot{2}{0}{0}; \blackdot{2}{0}{1};

    \foreach \z in {0,1}{
        \draw[Latex-Latex,shorten >= 5pt, shorten <= 5pt,very thin,yshift=-0.3mm] (0,0,\z) to [bend left] (0,1,\z);
        \draw[Latex-Latex,shorten >= 5pt, shorten <= 5pt,very thin] (1,1,\z) to (2,0,\z);
    }

    \node at (2.6,0.5,0) {$\xLongrightarrow{\tau_3}$};

    \draw[dashed] (3.7,1,0) -- (3.7,0,0) (3.7,1,1) -- (3.7,0,1) (4.7,1,0) -- (4.7,0,0) (4.7,1,1) -- (4.7,0,1)
        (5.7,1,0) -- (5.7,0,0) (5.7,1,1) -- (5.7,0,1)
        (3.7,1,0) -- (4.7,1,0) -- (5.7,1,0) (3.7,1,1) -- (4.7,1,1) -- (5.7,1,1) 
        (3.7,0,0) -- (4.7,0,0) -- (5.7,0,0) (3.7,0,1) -- (4.7,0,1) -- (5.7,0,1);
    \draw[double] (3.7,1,0) -- (3.7,1,1); \whitedot{3.7}{1}{0}; \whitedot{3.7}{1}{1}; 
    \draw[zig] (3.7,0,0) -- (3.7,0,1); \whitedot{3.7}{0}{0}; \whitedot{3.7}{0}{1};
    \draw[double] (4.7,1,0) -- (4.7,1,1); \whitedot{4.7}{1}{0}; \whitedot{4.7}{1}{1}; 
    \draw[zig] (4.7,0,0) -- (4.7,0,1); \whitedot{4.7}{0}{0}; \whitedot{4.7}{0}{1};
    \draw[double] (5.7,1,0) -- (5.7,1,1); \whitedot{5.7}{1}{0}; \whitedot{5.7}{1}{1}; 
    \draw[zig] (5.7,0,0) -- (5.7,0,1); \whitedot{5.7}{0}{0}; \whitedot{5.7}{0}{1};
}}
                \end{center}
        \end{itemize}
        
        \begin{itemize} 
            \item $\alpha=2,\beta=1,\gamma\geq 1$: This graph shows how to tackle $2$ A-cards and $1$ B-card with $1$ C-card. 
                \begin{center}                    
                    \scalebox{0.8}{\tikz{\Large
    \draw[dashed] (0,1,0) -- (0,0,0) (0,1,1) -- (0,0,1) (1,1,0) -- (1,0,0) (1,1,1) -- (1,0,1)
        (2,1,0) -- (2,0,0) (2,1,1) -- (2,0,1) (3,1,0) -- (3,0,0) (3,1,1) -- (3,0,1) 
        (0,1,0) -- (1,1,0) -- (2,1,0) -- (3,1,0) (0,1,1) -- (1,1,1) -- (2,1,1) -- (3,1,1) 
        (0,0,0) -- (1,0,0) -- (2,0,0) -- (3,0,0) (0,0,1) -- (1,0,1) -- (2,0,1) -- (3,0,1);
    \draw[double] (0,1,0) -- (0,1,1); \whitedot{0}{1}{0}; \blackdot{0}{1}{1}; 
    \draw[zig] (0,0,0) -- (0,0,1); \blackdot{0}{0}{0}; \whitedot{0}{0}{1};
    \draw[double] (1,1,0) -- (1,1,1); \whitedot{1}{1}{0}; \blackdot{1}{1}{1}; 
    \draw[zig] (1,0,0) -- (1,0,1); \blackdot{1}{0}{0}; \whitedot{1}{0}{1};
    \draw[double] (2,1,0) -- (2,1,1); \whitedot{2}{1}{0}; \blackdot{2}{1}{1}; 
    \draw[zig] (2,0,0) -- (2,0,1); \whitedot{2}{0}{0}; \blackdot{2}{0}{1};
    \draw[double] (3,1,0) -- (3,1,1); \whitedot{3}{1}{0}; \whitedot{3}{1}{1}; 
    \draw[zig] (3,0,0) -- (3,0,1); \whitedot{3}{0}{0}; \whitedot{3}{0}{1};

    \foreach \z in {0,1}{
        \draw[Latex-,shorten >= 5pt, shorten <= 5pt,very thin,yshift=0.7mm] (0,1,\z) to [bend left=10] (3,1,\z);
        \draw[Latex-,shorten >= 5pt, shorten <= 5pt,very thin,yshift=0.6mm] (3,1,\z) to [bend left] (3,0,\z);
        \draw[Latex-,shorten >= 5pt, shorten <= 5pt,very thin,yshift=0.7mm] (3,0,\z) to [bend right=10] (2,0,\z);
        \draw[Latex-,shorten >= 5pt, shorten <= 5pt,very thin] (2,0,\z) to (0,1,\z);
        \draw[Latex-Latex,shorten >= 5pt, shorten <= 5pt,very thin,yshift=0.7mm] (1,1,\z) to [bend left=10] (2,1,\z);
    }

    \node at (3.65,0.5,0) {$\xLongrightarrow{\tau_1}$};

    \draw[dashed] (4.6,1,0) -- (4.6,0,0) (4.6,1,1) -- (4.6,0,1) (5.6,1,0) -- (5.6,0,0) (5.6,1,1) -- (5.6,0,1)
        (6.6,1,0) -- (6.6,0,0) (6.6,1,1) -- (6.6,0,1) (7.6,1,0) -- (7.6,0,0) (7.6,1,1) -- (7.6,0,1) 
        (4.6,1,0) -- (5.6,1,0) -- (6.6,1,0) -- (7.6,1,0) (4.6,1,1) -- (5.6,1,1) -- (6.6,1,1) -- (7.6,1,1) 
        (4.6,0,0) -- (5.6,0,0) -- (6.6,0,0) -- (7.6,0,0) (4.6,0,1) -- (5.6,0,1) -- (6.6,0,1) -- (7.6,0,1);
    \draw[double] (4.6,1,0) -- (4.6,1,1); \blackdot{4.6}{1}{0}; \whitedot{4.6}{1}{1}; 
    \draw[zig] (4.6,0,0) -- (4.6,0,1); \blackdot{4.6}{0}{0}; \whitedot{4.6}{0}{1};
    \draw[double] (5.6,1,0) -- (5.6,1,1); \whitedot{5.6}{1}{0}; \blackdot{5.6}{1}{1}; 
    \draw[zig] (5.6,0,0) -- (5.6,0,1); \blackdot{5.6}{0}{0}; \whitedot{5.6}{0}{1};
    \draw[double] (6.6,1,0) -- (6.6,1,1); \whitedot{6.6}{1}{0}; \blackdot{6.6}{1}{1}; 
    \draw[zig] (6.6,0,0) -- (6.6,0,1); \whitedot{6.6}{0}{0}; \whitedot{6.6}{0}{1};
    \draw[double] (7.6,1,0) -- (7.6,1,1); \whitedot{7.6}{1}{0}; \blackdot{7.6}{1}{1}; 
    \draw[zig] (7.6,0,0) -- (7.6,0,1); \blackdot{7.6}{0}{0}; \blackdot{7.6}{0}{1};

    \foreach \y in {0,1}{
        \draw[Latex-Latex,shorten >= 5pt, shorten <= 5pt,very thin,yshift=-0.7mm] (4.6,\y,1) to [bend right=10] (5.6,\y,1);
        \draw[Latex-Latex,shorten >= 5pt, shorten <= 5pt,very thin] (4.6,\y,0) to (6.6,\y,1);
        \draw[Latex-,shorten >= 5pt, shorten <= 5pt,very thin,yshift=0.7mm] (6.6,\y,0) to [bend right=10] (5.6,\y,0);
        \draw[Latex-,shorten >= 5pt, shorten <= 5pt,very thin] (5.6,\y,0) to (7.6,\y,1);
        \draw[Latex-,shorten >= 5pt, shorten <= 5pt,very thin] (7.6,\y,1) to (6.6,\y,0);
    }
}}
\vspace{5pt}\\
\hspace{-5pt}\scalebox{0.8}{\tikz{\Large
    \node at (-1,0.5,0) {$\xLongrightarrow{\tau_2}$};

    \draw[dashed] (0,1,0) -- (0,0,0) (0,1,1) -- (0,0,1) (1,1,0) -- (1,0,0) (1,1,1) -- (1,0,1)
        (2,1,0) -- (2,0,0) (2,1,1) -- (2,0,1) (3,1,0) -- (3,0,0) (3,1,1) -- (3,0,1) 
        (0,1,0) -- (1,1,0) -- (2,1,0) -- (3,1,0) (0,1,1) -- (1,1,1) -- (2,1,1) -- (3,1,1) 
        (0,0,0) -- (1,0,0) -- (2,0,0) -- (3,0,0) (0,0,1) -- (1,0,1) -- (2,0,1) -- (3,0,1);
    \draw[double] (0,1,0) -- (0,1,1); \blackdot{0}{1}{0}; \blackdot{0}{1}{1}; 
    \draw[zig] (0,0,0) -- (0,0,1); \whitedot{0}{0}{0}; \whitedot{0}{0}{1};
    \draw[double] (1,1,0) -- (1,1,1); \whitedot{1}{1}{0}; \whitedot{1}{1}{1}; 
    \draw[zig] (1,0,0) -- (1,0,1); \whitedot{1}{0}{0}; \whitedot{1}{0}{1};
    \draw[double] (2,1,0) -- (2,1,1); \blackdot{2}{1}{0}; \blackdot{2}{1}{1}; 
    \draw[zig] (2,0,0) -- (2,0,1); \blackdot{2}{0}{0}; \blackdot{2}{0}{1};
    \draw[double] (3,1,0) -- (3,1,1); \whitedot{3}{1}{0}; \whitedot{3}{1}{1}; 
    \draw[zig] (3,0,0) -- (3,0,1); \blackdot{3}{0}{0}; \blackdot{3}{0}{1};

    \foreach \z in {0,1}{
        \draw[Latex-Latex,shorten >= 5pt, shorten <= 5pt,very thin] (0,1,\z) to (2,0,\z);
        \draw[Latex-Latex,shorten >= 5pt, shorten <= 5pt,very thin] (2,1,\z) to (3,0,\z);
    }

    \node at (3.6,0.5,0) {$\xLongrightarrow{\tau_3}$};

    \draw[dashed] (4.6,1,0) -- (4.6,0,0) (4.6,1,1) -- (4.6,0,1) (5.6,1,0) -- (5.6,0,0) (5.6,1,1) -- (5.6,0,1)
        (6.6,1,0) -- (6.6,0,0) (6.6,1,1) -- (6.6,0,1) (7.6,1,0) -- (7.6,0,0) (7.6,1,1) -- (7.6,0,1) 
        (4.6,1,0) -- (5.6,1,0) -- (6.6,1,0) -- (7.6,1,0) (4.6,1,1) -- (5.6,1,1) -- (6.6,1,1) -- (7.6,1,1) 
        (4.6,0,0) -- (5.6,0,0) -- (6.6,0,0) -- (7.6,0,0) (4.6,0,1) -- (5.6,0,1) -- (6.6,0,1) -- (7.6,0,1);
    \draw[double] (4.6,1,0) -- (4.6,1,1); \whitedot{4.6}{1}{0}; \whitedot{4.6}{1}{1}; 
    \draw[zig] (4.6,0,0) -- (4.6,0,1); \whitedot{4.6}{0}{0}; \whitedot{4.6}{0}{1};
    \draw[double] (5.6,1,0) -- (5.6,1,1); \whitedot{5.6}{1}{0}; \whitedot{5.6}{1}{1}; 
    \draw[zig] (5.6,0,0) -- (5.6,0,1); \whitedot{5.6}{0}{0}; \whitedot{5.6}{0}{1};
    \draw[double] (6.6,1,0) -- (6.6,1,1); \whitedot{6.6}{1}{0}; \whitedot{6.6}{1}{1}; 
    \draw[zig] (6.6,0,0) -- (6.6,0,1); \whitedot{6.6}{0}{0}; \whitedot{6.6}{0}{1};
    \draw[double] (7.6,1,0) -- (7.6,1,1); \whitedot{7.6}{1}{0}; \whitedot{7.6}{1}{1}; 
    \draw[zig] (7.6,0,0) -- (7.6,0,1); \whitedot{7.6}{0}{0}; \whitedot{7.6}{0}{1};
}} 
                \end{center}
            \item $\alpha=0,\beta=3,\gamma\geq 1$ : This graph shows how to tackle $3$ B-cards with $1$ C-card. 
                \begin{center}
                    \scalebox{0.8}
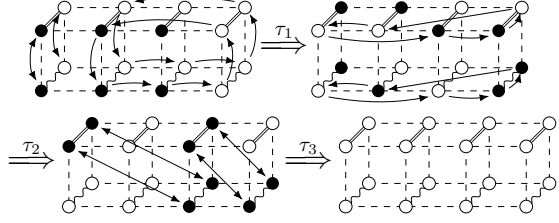
{\tikz{\Large
    \draw[dashed] (0,1,0) -- (0,0,0) (0,1,1) -- (0,0,1) (1,1,0) -- (1,0,0) (1,1,1) -- (1,0,1)
        (2,1,0) -- (2,0,0) (2,1,1) -- (2,0,1) (3,1,0) -- (3,0,0) (3,1,1) -- (3,0,1) 
        (0,1,0) -- (1,1,0) -- (2,1,0) -- (3,1,0) (0,1,1) -- (1,1,1) -- (2,1,1) -- (3,1,1) 
        (0,0,0) -- (1,0,0) -- (2,0,0) -- (3,0,0) (0,0,1) -- (1,0,1) -- (2,0,1) -- (3,0,1);
    \draw[double] (0,1,0) -- (0,1,1); \whitedot{0}{1}{0}; \blackdot{0}{1}{1}; 
    \draw[zig] (0,0,0) -- (0,0,1); \whitedot{0}{0}{0}; \blackdot{0}{0}{1};
    \draw[double] (1,1,0) -- (1,1,1); \whitedot{1}{1}{0}; \blackdot{1}{1}{1}; 
    \draw[zig] (1,0,0) -- (1,0,1); \whitedot{1}{0}{0}; \blackdot{1}{0}{1};
    \draw[double] (2,1,0) -- (2,1,1); \whitedot{2}{1}{0}; \blackdot{2}{1}{1}; 
    \draw[zig] (2,0,0) -- (2,0,1); \whitedot{2}{0}{0}; \blackdot{2}{0}{1};
    \draw[double] (3,1,0) -- (3,1,1); \whitedot{3}{1}{0}; \whitedot{3}{1}{1}; 
    \draw[zig] (3,0,0) -- (3,0,1); \whitedot{3}{0}{0}; \whitedot{3}{0}{1};

    \foreach \z in {0,1}{
        \draw[Latex-Latex,shorten >= 5pt, shorten <= 5pt,very thin,yshift=0.3mm] (0,1,\z) to [bend right] (0,0,\z);
        \draw[Latex-,shorten >= 5pt, shorten <= 5pt,very thin,yshift=0.7mm] (1,1,\z) to [bend left=10] (3,1,\z);
        \draw[Latex-,shorten >= 5pt, shorten <= 5pt,very thin,yshift=0.3mm] (3,1,\z) to [bend left] (3,0,\z);
        \draw[Latex-,shorten >= 5pt, shorten <= 5pt,very thin,yshift=0.7mm] (3,0,\z) to [bend right=10] (2,0,\z);
        \draw[Latex-,shorten >= 5pt, shorten <= 5pt,very thin,yshift=0.7mm] (2,0,\z) to [bend right=10] (1,0,\z);
        \draw[Latex-,shorten >= 5pt, shorten <= 5pt,very thin,yshift=0.3mm] (1,0,\z) to [bend left] (1,1,\z);
    }

    \node at (3.65,0.5,0) {$\xLongrightarrow{\tau_1}$};

    \draw[dashed] (4.6,1,0) -- (4.6,0,0) (4.6,1,1) -- (4.6,0,1) (5.6,1,0) -- (5.6,0,0) (5.6,1,1) -- (5.6,0,1)
        (6.6,1,0) -- (6.6,0,0) (6.6,1,1) -- (6.6,0,1) (7.6,1,0) -- (7.6,0,0) (7.6,1,1) -- (7.6,0,1) 
        (4.6,1,0) -- (5.6,1,0) -- (6.6,1,0) -- (7.6,1,0) (4.6,1,1) -- (5.6,1,1) -- (6.6,1,1) -- (7.6,1,1) 
        (4.6,0,0) -- (5.6,0,0) -- (6.6,0,0) -- (7.6,0,0) (4.6,0,1) -- (5.6,0,1) -- (6.6,0,1) -- (7.6,0,1);
    \draw[double] (4.6,1,0) -- (4.6,1,1); \blackdot{4.6}{1}{0}; \whitedot{4.6}{1}{1}; 
    \draw[zig] (4.6,0,0) -- (4.6,0,1); \blackdot{4.6}{0}{0}; \whitedot{4.6}{0}{1};
    \draw[double] (5.6,1,0) -- (5.6,1,1); \blackdot{5.6}{1}{0}; \whitedot{5.6}{1}{1}; 
    \draw[zig] (5.6,0,0) -- (5.6,0,1); \whitedot{5.6}{0}{0}; \blackdot{5.6}{0}{1};
    \draw[double] (6.6,1,0) -- (6.6,1,1); \whitedot{6.6}{1}{0}; \blackdot{6.6}{1}{1}; 
    \draw[zig] (6.6,0,0) -- (6.6,0,1); \whitedot{6.6}{0}{0}; \whitedot{6.6}{0}{1};
    \draw[double] (7.6,1,0) -- (7.6,1,1); \whitedot{7.6}{1}{0}; \blackdot{7.6}{1}{1}; 
    \draw[zig] (7.6,0,0) -- (7.6,0,1); \blackdot{7.6}{0}{0}; \blackdot{7.6}{0}{1};

    \foreach \y in {0,1}{
        \draw[-Latex,shorten >= 5pt, shorten <= 5pt,very thin,yshift=-0.7mm] (4.6,\y,1) to [bend right=10] (6.6,\y,1);
        \draw[-Latex,shorten >= 5pt, shorten <= 5pt,very thin,yshift=-0.6mm] (6.6,\y,1) to [bend right=10] (7.6,\y,1);
        \draw[-Latex,shorten >= 3pt, shorten <= 5pt,very thin] (7.6,\y,1) to [bend right] (7.6,\y,0);
        \draw[-Latex,shorten >= 5pt, shorten <= 5pt,very thin] (7.6,\y,0) to (5.6,\y,1);
        \draw[-Latex,shorten >= 5pt, shorten <= 5pt,very thin,yshift=0.6mm] (5.6,\y,1) to [bend right=10] (4.6,\y,1);
    }
}}
\vspace{5pt}\\
\hspace{-5pt}\scalebox{0.8}{\tikz{\Large
    \node at (-1,0.5,0) {$\xLongrightarrow{\tau_2}$};

    \draw[dashed] (0,1,0) -- (0,0,0) (0,1,1) -- (0,0,1) (1,1,0) -- (1,0,0) (1,1,1) -- (1,0,1)
        (2,1,0) -- (2,0,0) (2,1,1) -- (2,0,1) (3,1,0) -- (3,0,0) (3,1,1) -- (3,0,1) 
        (0,1,0) -- (1,1,0) -- (2,1,0) -- (3,1,0) (0,1,1) -- (1,1,1) -- (2,1,1) -- (3,1,1) 
        (0,0,0) -- (1,0,0) -- (2,0,0) -- (3,0,0) (0,0,1) -- (1,0,1) -- (2,0,1) -- (3,0,1);
    \draw[double] (0,1,0) -- (0,1,1); \blackdot{0}{1}{0}; \blackdot{0}{1}{1}; 
    \draw[zig] (0,0,0) -- (0,0,1); \whitedot{0}{0}{0}; \whitedot{0}{0}{1};
    \draw[double] (1,1,0) -- (1,1,1); \whitedot{1}{1}{0}; \whitedot{1}{1}{1}; 
    \draw[zig] (1,0,0) -- (1,0,1); \whitedot{1}{0}{0}; \whitedot{1}{0}{1};
    \draw[double] (2,1,0) -- (2,1,1); \blackdot{2}{1}{0}; \blackdot{2}{1}{1}; 
    \draw[zig] (2,0,0) -- (2,0,1); \blackdot{2}{0}{0}; \blackdot{2}{0}{1};
    \draw[double] (3,1,0) -- (3,1,1); \whitedot{3}{1}{0}; \whitedot{3}{1}{1}; 
    \draw[zig] (3,0,0) -- (3,0,1); \blackdot{3}{0}{0}; \blackdot{3}{0}{1};

    \foreach \z in {0,1}{
        \draw[Latex-Latex,shorten >= 5pt, shorten <= 5pt,very thin] (0,1,\z) to (2,0,\z);
        \draw[Latex-Latex,shorten >= 5pt, shorten <= 5pt,very thin] (2,1,\z) to (3,0,\z);
    }

    \node at (3.6,0.5,0) {$\xLongrightarrow{\tau_3}$};

    \draw[dashed] (4.6,1,0) -- (4.6,0,0) (4.6,1,1) -- (4.6,0,1) (5.6,1,0) -- (5.6,0,0) (5.6,1,1) -- (5.6,0,1)
        (6.6,1,0) -- (6.6,0,0) (6.6,1,1) -- (6.6,0,1) (7.6,1,0) -- (7.6,0,0) (7.6,1,1) -- (7.6,0,1) 
        (4.6,1,0) -- (5.6,1,0) -- (6.6,1,0) -- (7.6,1,0) (4.6,1,1) -- (5.6,1,1) -- (6.6,1,1) -- (7.6,1,1) 
        (4.6,0,0) -- (5.6,0,0) -- (6.6,0,0) -- (7.6,0,0) (4.6,0,1) -- (5.6,0,1) -- (6.6,0,1) -- (7.6,0,1);
    \draw[double] (4.6,1,0) -- (4.6,1,1); \whitedot{4.6}{1}{0}; \whitedot{4.6}{1}{1}; 
    \draw[zig] (4.6,0,0) -- (4.6,0,1); \whitedot{4.6}{0}{0}; \whitedot{4.6}{0}{1};
    \draw[double] (5.6,1,0) -- (5.6,1,1); \whitedot{5.6}{1}{0}; \whitedot{5.6}{1}{1}; 
    \draw[zig] (5.6,0,0) -- (5.6,0,1); \whitedot{5.6}{0}{0}; \whitedot{5.6}{0}{1};
    \draw[double] (6.6,1,0) -- (6.6,1,1); \whitedot{6.6}{1}{0}; \whitedot{6.6}{1}{1}; 
    \draw[zig] (6.6,0,0) -- (6.6,0,1); \whitedot{6.6}{0}{0}; \whitedot{6.6}{0}{1};
    \draw[double] (7.6,1,0) -- (7.6,1,1); \whitedot{7.6}{1}{0}; \whitedot{7.6}{1}{1}; 
    \draw[zig] (7.6,0,0) -- (7.6,0,1); \whitedot{7.6}{0}{0}; \whitedot{7.6}{0}{1};
}} 
                \end{center}
            \item $\alpha=1,\beta=2,\gamma\geq 1$ : This graph shows how to tackle $1$ A-card and $2$ B-cards with $1$ C-card. 
                \begin{center}
                    \scalebox{0.8}
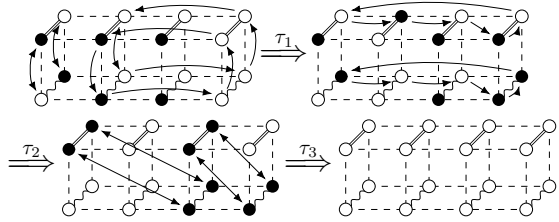
{\tikz{\Large
    \draw[dashed] (0,1,0) -- (0,0,0) (0,1,1) -- (0,0,1) (1,1,0) -- (1,0,0) (1,1,1) -- (1,0,1)
        (2,1,0) -- (2,0,0) (2,1,1) -- (2,0,1) (3,1,0) -- (3,0,0) (3,1,1) -- (3,0,1) 
        (0,1,0) -- (1,1,0) -- (2,1,0) -- (3,1,0) (0,1,1) -- (1,1,1) -- (2,1,1) -- (3,1,1) 
        (0,0,0) -- (1,0,0) -- (2,0,0) -- (3,0,0) (0,0,1) -- (1,0,1) -- (2,0,1) -- (3,0,1);
    \draw[double] (0,1,0) -- (0,1,1); \whitedot{0}{1}{0}; \blackdot{0}{1}{1}; 
    \draw[zig] (0,0,0) -- (0,0,1); \blackdot{0}{0}{0}; \whitedot{0}{0}{1};
    \draw[double] (1,1,0) -- (1,1,1); \whitedot{1}{1}{0}; \blackdot{1}{1}{1}; 
    \draw[zig] (1,0,0) -- (1,0,1); \whitedot{1}{0}{0}; \blackdot{1}{0}{1};
    \draw[double] (2,1,0) -- (2,1,1); \whitedot{2}{1}{0}; \blackdot{2}{1}{1}; 
    \draw[zig] (2,0,0) -- (2,0,1); \whitedot{2}{0}{0}; \blackdot{2}{0}{1};
    \draw[double] (3,1,0) -- (3,1,1); \whitedot{3}{1}{0}; \whitedot{3}{1}{1}; 
    \draw[zig] (3,0,0) -- (3,0,1); \whitedot{3}{0}{0}; \whitedot{3}{0}{1};

    \foreach \z in {0,1}{
    	\draw[Latex-Latex,shorten >= 5pt, shorten <= 5pt,very thin] (0,1,\z) to [bend right] (0,0,\z);
        \draw[-Latex,shorten >= 5pt, shorten <= 5pt,very thin,yshift=-0.3mm] (1,1,\z) to [bend right] (1,0,\z);
        \draw[-Latex,shorten >= 5pt, shorten <= 5pt,very thin,yshift=0.7mm] (1,0,\z) to [bend left=10] (3,0,\z);
        \draw[-Latex,shorten >= 5pt, shorten <= 5pt,very thin,yshift=0.3mm] (3,0,\z) to [bend right] (3,1,\z);
        \draw[-Latex,shorten >= 5pt, shorten <= 5pt,very thin,yshift=0.7mm] (3,1,\z) to [bend right=10] (1,1,\z);
    }

    \node at (3.65,0.5,0) {$\xLongrightarrow{\tau_1}$};

    \draw[dashed] (4.6,1,0) -- (4.6,0,0) (4.6,1,1) -- (4.6,0,1) (5.6,1,0) -- (5.6,0,0) (5.6,1,1) -- (5.6,0,1)
        (6.6,1,0) -- (6.6,0,0) (6.6,1,1) -- (6.6,0,1) (7.6,1,0) -- (7.6,0,0) (7.6,1,1) -- (7.6,0,1) 
        (4.6,1,0) -- (5.6,1,0) -- (6.6,1,0) -- (7.6,1,0) (4.6,1,1) -- (5.6,1,1) -- (6.6,1,1) -- (7.6,1,1) 
        (4.6,0,0) -- (5.6,0,0) -- (6.6,0,0) -- (7.6,0,0) (4.6,0,1) -- (5.6,0,1) -- (6.6,0,1) -- (7.6,0,1);
    \draw[double] (4.6,1,0) -- (4.6,1,1); \whitedot{4.6}{1}{0}; \blackdot{4.6}{1}{1}; 
    \draw[zig] (4.6,0,0) -- (4.6,0,1); \blackdot{4.6}{0}{0}; \whitedot{4.6}{0}{1};
    \draw[double] (5.6,1,0) -- (5.6,1,1); \blackdot{5.6}{1}{0}; \whitedot{5.6}{1}{1}; 
    \draw[zig] (5.6,0,0) -- (5.6,0,1); \whitedot{5.6}{0}{0}; \whitedot{5.6}{0}{1};
    \draw[double] (6.6,1,0) -- (6.6,1,1); \whitedot{6.6}{1}{0}; \blackdot{6.6}{1}{1}; 
    \draw[zig] (6.6,0,0) -- (6.6,0,1); \whitedot{6.6}{0}{0}; \blackdot{6.6}{0}{1};
    \draw[double] (7.6,1,0) -- (7.6,1,1); \whitedot{7.6}{1}{0}; \blackdot{7.6}{1}{1}; 
    \draw[zig] (7.6,0,0) -- (7.6,0,1); \blackdot{7.6}{0}{0}; \blackdot{7.6}{0}{1};

    \foreach \y in {0,1}{
    	\draw[-Latex,shorten >= 5pt, shorten <= 5pt,very thin,yshift=-0.6mm] (4.6,\y,0) to [bend right=10] (5.6,\y,0);
    	\draw[-Latex,shorten >= 5pt, shorten <= 5pt,very thin,yshift=-0.6mm] (5.6,\y,0) to [bend right=10] (6.6,\y,0);
    	\draw[-Latex,shorten >= 5pt, shorten <= 5pt,very thin] (6.6,\y,0) to (7.6,\y,1);
    	\draw[-Latex,shorten >= 5pt, shorten <= 5pt,very thin] (7.6,\y,1) to [bend right] (7.6,\y,0);
    	\draw[-Latex,shorten >= 5pt, shorten <= 5pt,very thin,yshift=0.7mm] (7.6,\y,0) to [bend right=10] (4.6,\y,0);
    }
}}
\vspace{5pt}\\
\hspace{-5pt}\scalebox{0.8}{\tikz{\Large
    \node at (-1,0.5,0) {$\xLongrightarrow{\tau_2}$};

    \draw[dashed] (0,1,0) -- (0,0,0) (0,1,1) -- (0,0,1) (1,1,0) -- (1,0,0) (1,1,1) -- (1,0,1)
        (2,1,0) -- (2,0,0) (2,1,1) -- (2,0,1) (3,1,0) -- (3,0,0) (3,1,1) -- (3,0,1) 
        (0,1,0) -- (1,1,0) -- (2,1,0) -- (3,1,0) (0,1,1) -- (1,1,1) -- (2,1,1) -- (3,1,1) 
        (0,0,0) -- (1,0,0) -- (2,0,0) -- (3,0,0) (0,0,1) -- (1,0,1) -- (2,0,1) -- (3,0,1);
    \draw[double] (0,1,0) -- (0,1,1); \blackdot{0}{1}{0}; \blackdot{0}{1}{1}; 
    \draw[zig] (0,0,0) -- (0,0,1); \whitedot{0}{0}{0}; \whitedot{0}{0}{1};
    \draw[double] (1,1,0) -- (1,1,1); \whitedot{1}{1}{0}; \whitedot{1}{1}{1}; 
    \draw[zig] (1,0,0) -- (1,0,1); \whitedot{1}{0}{0}; \whitedot{1}{0}{1};
    \draw[double] (2,1,0) -- (2,1,1); \blackdot{2}{1}{0}; \blackdot{2}{1}{1}; 
    \draw[zig] (2,0,0) -- (2,0,1); \blackdot{2}{0}{0}; \blackdot{2}{0}{1};
    \draw[double] (3,1,0) -- (3,1,1); \whitedot{3}{1}{0}; \whitedot{3}{1}{1}; 
    \draw[zig] (3,0,0) -- (3,0,1); \blackdot{3}{0}{0}; \blackdot{3}{0}{1};

    \foreach \z in {0,1}{
        \draw[Latex-Latex,shorten >= 5pt, shorten <= 5pt,very thin] (0,1,\z) to (2,0,\z);
        \draw[Latex-Latex,shorten >= 5pt, shorten <= 5pt,very thin] (2,1,\z) to (3,0,\z);
    }

    \node at (3.6,0.5,0) {$\xLongrightarrow{\tau_3}$};

    \draw[dashed] (4.6,1,0) -- (4.6,0,0) (4.6,1,1) -- (4.6,0,1) (5.6,1,0) -- (5.6,0,0) (5.6,1,1) -- (5.6,0,1)
        (6.6,1,0) -- (6.6,0,0) (6.6,1,1) -- (6.6,0,1) (7.6,1,0) -- (7.6,0,0) (7.6,1,1) -- (7.6,0,1) 
        (4.6,1,0) -- (5.6,1,0) -- (6.6,1,0) -- (7.6,1,0) (4.6,1,1) -- (5.6,1,1) -- (6.6,1,1) -- (7.6,1,1) 
        (4.6,0,0) -- (5.6,0,0) -- (6.6,0,0) -- (7.6,0,0) (4.6,0,1) -- (5.6,0,1) -- (6.6,0,1) -- (7.6,0,1);
    \draw[double] (4.6,1,0) -- (4.6,1,1); \whitedot{4.6}{1}{0}; \whitedot{4.6}{1}{1}; 
    \draw[zig] (4.6,0,0) -- (4.6,0,1); \whitedot{4.6}{0}{0}; \whitedot{4.6}{0}{1};
    \draw[double] (5.6,1,0) -- (5.6,1,1); \whitedot{5.6}{1}{0}; \whitedot{5.6}{1}{1}; 
    \draw[zig] (5.6,0,0) -- (5.6,0,1); \whitedot{5.6}{0}{0}; \whitedot{5.6}{0}{1};
    \draw[double] (6.6,1,0) -- (6.6,1,1); \whitedot{6.6}{1}{0}; \whitedot{6.6}{1}{1}; 
    \draw[zig] (6.6,0,0) -- (6.6,0,1); \whitedot{6.6}{0}{0}; \whitedot{6.6}{0}{1};
    \draw[double] (7.6,1,0) -- (7.6,1,1); \whitedot{7.6}{1}{0}; \whitedot{7.6}{1}{1}; 
    \draw[zig] (7.6,0,0) -- (7.6,0,1); \whitedot{7.6}{0}{0}; \whitedot{7.6}{0}{1};
}} 
                \end{center}
        \end{itemize}
    
\end{proof}

\begin{proof}[Proof of Lemma \ref{e35a2}]
    To ease the presentation, we say $\bm u, \bm v$ (or $\{\bm u, \bm v\}$) is a \emph{concurrent pair}, if $\bm u = \bm v^{\oplus r_1}$.
	
	The cycle transforming process is divided into following 4 stages: 

	\noindent\textbf{Stage I.}\quad In the first stage, we attempt to construct $\pi_0\in SC_{\zoton}^{(r_1)}$ such that $\sigma\pi_0$ contains an even cycle $\mathscr C_0$ of length no more than $4$.
	
	\begin{casepar}{Case 0}
	    Suppose there exists a 2-cycle in $\sigma$ already, then simply let $\pi_0 := \id$.
	\end{casepar}
	
	\begin{casepar}{Case 1}
	    Suppose there exist $\bm u, \bm v$ such that $\bm u_{r_1} = \bm v_{r_1} = \sigma(\bm u)_{r_1} = \sigma(\bm v)_{r_1}$. If $\sigma(\bm u) = \bm v$ (or $\sigma(\bm v) = \bm u$), perform 
	    \[
	        \pi_0 := (\bm u, \sigma(\bm v))(\bm u^{\oplus r_1}, \sigma(\bm v)^{\oplus r_1})
        \]
	    and a 2-cycle $\mathscr C_0 = (\bm v, \sigma(\bm v))$ will appear. Otherwise, perform 
	    \[
	        \pi_0' := \begin{array}{l}
	            (\bm u, \sigma(\bm u))(\bm u^{\oplus r_1}, \sigma(\bm u)^{\oplus r_1})\\
	            (\bm v, \sigma(\bm v))(\bm v^{\oplus r_1}, \sigma(\bm v)^{\oplus r_1})
            \end{array}            
	    \]
	    and 2 fix-points $\sigma(\bm u), \sigma(\bm v)$ will appear. Thus, 
	    \[
	        \pi_0 := \pi_0'\circ (\sigma(\bm u),\sigma(\bm v))(\sigma(\bm u)^{\oplus r_1},\sigma(\bm v)^{\oplus r_1})
        \] 
        is as required such that $\sigma\pi_0$ contains a 2-cycle $\mathscr C_0 = (\sigma(\bm u), \sigma(\bm v))$. 
	\end{casepar}
	
	\begin{casepar}{Case 2}
	    Suppose there exists $\bm u$ such that $\bm u_{r_1} = \sigma^2(\bm u)_{r_1}$, $\bm u_{r_1} \neq \sigma(\bm u)_{r_1}$, and $\bm u, \sigma(\bm u)^{\oplus r_1}, \sigma^2(\bm u)$ are distinct. Let 
	    \[
	        \pi_0 := (\bm u, \sigma^2(\bm u))(\bm u^{\oplus r_1}, \sigma^2(\bm u)^{\oplus r_1}).
        \] 
        Thus, $\sigma\pi_0$ will contain a 2-cycle $\mathscr C_0 = (\sigma(\bm u), \sigma^2(\bm u))$.
	\end{casepar}

    \begin{casepar}{Case 3}
        Suppose there exist fix-points $\bm u, \bm v$ such that $\bm u_{r_1} = \bm v_{r_1}$. Let 
        \[
            \pi_0 := (\bm u, \bm v)(\bm u^{\oplus r_1}, \bm v^{\oplus r_1}).
        \]
        Thus, $\sigma\pi_0$ will contain a 2-cycle $\mathscr C_0 = (\bm u, \bm v)$.
    \end{casepar}
    
    \begin{casepar}{Case 4}
        If none of the previous 3 cases holds, either there exists a 4-cycle containing two concurrent pairs, or there exist distinct $\bm u_1, \ldots, \bm u_6$ such that $(\ldots, \bm u_1, \bm u_2,\bm u_3,\bm u_4,\bm u_5, \bm u_6, \ldots)$ is in $\sigma$, $(\bm u_1)_{r_1} = (\bm u_3)_{r_1} = (\bm u_5)_{r_1}$ and $(\bm u_1, \bm u_2), (\bm u_3, \bm u_4), (\bm u_5, \bm u_6)$ are concurrent pairs. 
        Then let $\pi_0=\id$ for the first one; and $\pi_0 = (\bm u_1,\bm u_3,\bm u_5)(\bm u_2,\bm u_4,\bm u_6)$ for the second.
        
        \begin{figure}[H]
        	\centering
        	\scalebox{0.8}{\begin{tikzpicture}[>=stealth,dot/.style={circle,fill=black,inner sep=0pt,minimum size=5pt}]

        \draw[<->] (-1.5,0) -- (-1.5,1); 
        
        \node at (-1.7, 0.5) {$r_1$};
        
        \draw (-1.6,0) -- (-1.4,0);
        \draw (-1.6,1) -- (-1.4,1);

		\node[dot, label=below:{$\bm u_2$}] (a00) at (0,0) {};
		\node[dot, label=below:{$\bm u_4$}] (a01) at (1,0) {};
		\node[dot, label=below:{$\bm u_6$}] (a02) at (2,0) {};
		\node[dot, label={$\bm u_1$}] (a10) at (0,1) {};
		\node[dot, label={$\bm u_3$}] (a11) at (1,1) {};
		\node[dot, label={$\bm u_5$}] (a12) at (2,1) {};

		\draw[-Latex,shorten >= 1pt] (-0.5, 0.5) -- (a10);
        \draw[-Latex,shorten >= 1pt] (a10) -- (a00);
        \draw[-Latex,shorten >= 1pt] (a00) -- (a11);
        \draw[-Latex,shorten >= 1pt] (a11) -- (a01);
        \draw[-Latex,shorten >= 1pt] (a01) -- (a12);
        \draw[-Latex,shorten >= 1pt] (a12) -- (a02);
        \draw[-Latex,shorten >= 1pt] (a02) -- (2.5, 0.5);

        \node[dot, label=below:{$\bm u_2$}] (b00) at (4,0) {};
		\node[dot, label=below:{$\bm u_4$}] (b01) at (5,0) {};
		\node[dot, label={$\bm u_1$}] (b10) at (4,1) {};
		\node[dot, label={$\bm u_3$}] (b11) at (5,1) {};


        \draw[-Latex,shorten >= 1pt] (b10) -- (b00);
        \draw[-Latex,shorten >= 1pt] (b00) -- (b11);
        \draw[-Latex,shorten >= 1pt] (b11) -- (b01);
        \draw[-Latex,shorten >= 1pt] (b01) -- (b10);


\end{tikzpicture}}
        	\caption{One of the structures can be found in Case 4}
    	\end{figure}
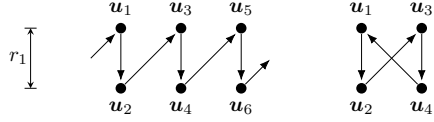
    \end{casepar}

	\noindent\textbf{Stage II.}\quad 
	In this stage, several concurrent swaps will be performed to eliminate most of the 3/5-cycles and keep $\mathscr C_0$ invariant. 
	The following operation will be iterated in several rounds. In round-$i$, $\pi_i \in SC_{\zoton}^{(r_1)}$ is performed. Let $S_{i,c}$ be the set of all $c$-cycle on which each vertex $\bm v$ satisfies $\bm v, \bm v^{\oplus r_1} \not\in \mathscr C_0$ in $\sigma_{i-1}$ ($\sigma_t:=\sigma\pi_0\pi_1\ldots\pi_t$ in the following).
	
	Denote $\zeta_{i}:=|S_{i,1}|+|S_{i,2}|+|S_{i,3}|+|S_{i,4}|+|S_{i,5}|$. If $S_{i-1,3} \cup S_{i-1,5}\neq\emptyset$, 
k an arbitrary cycle $\mathscr C_1$ from it. Since $\mathscr C_1$ is an odd cycle, there exists $\bm u\in\mathscr C_1$ such that $\bm v:=\bm u^{\oplus r_1}\notin\mathscr C_1$. Let $\mathscr C_2$ be the cycle where $\bm v$ belongs (by choice of $\mathscr C_1$, here $\mathscr C_2\neq\mathscr C_0$). Define
    $$
	    T:=\mathscr C_0\cup\mathscr C_1\cup 
	    \big\{\bm w\in\mathscr C_2\mid \mdist^{\sigma_{i-1}}(\bm v,\bm w)\leq 5\big\}.
	$$
	Note that $|T|\leq 4+5+11$. Since $n\geq8$, we can always find a concurrent pair $(\bm s, \bm t)$ that $\bm s,\bm t\not \in T$. Then, let 
	\[
	    \pi_i:=(\bm u, \bm t)(\bm v, \bm s)\in SC_{\zoton}^{(r_1)}.
	\]
	We will prove $\zeta_{i} < \zeta_{i - 1}$, by checking the following cases:
	\begin{casepar}{Case 1}
	    $\bm t, \bm s \notin \mathscr C_2$: 
			Swapping $\bm u, \bm t$ merges $\mathscr C_1$ with another cycle. And similarly when swapping $\bm v, \bm s$.
	\end{casepar}
	\begin{casepar}{Case 2}
	    $\bm t \notin \mathscr C_2, \bm s \in \mathscr C_2$: 
			Swapping $\bm u, \bm t$ merges $\mathscr C_1$ with another cycle. Then swapping $\bm v, \bm s$ splits $\mathscr C_2$ into two cycles; and the length of neither is smaller than $6$, which does not increase the number of short cycles.
	\end{casepar}
	\begin{casepar}{Case 3}
	    $\bm t \in \mathscr C_2, \bm s \notin \mathscr C_2$: 
			Swapping $\bm u, \bm t$ merges $\mathscr C_1$ with $\mathscr C_2$. Then swapping $\bm v, \bm s$ merges new $\mathscr C_2$ with another cycle.
	\end{casepar}
	\begin{casepar}{Case 4}
	    $\bm t, \bm s\in\mathscr C_2$:
			Swapping $\bm u, \bm t$ merges $\mathscr C_1$ with $\mathscr C_2$. Then swapping $\bm v, \bm s$ splits new $\mathscr C_2$ into two cycles; and the length of neither is smaller than $6$, which does not increase the number of short cycles.
	\end{casepar}
	Repeat until $S_{i,3} \cup S_{i,5}=\emptyset$. Suppose this process has $k$ rounds, then the permutation after Stage II is $\sigma_k=\sigma\pi_0\pi_1\ldots\pi_{k}$.
	\medskip

	\noindent\textbf{Stage III.}\quad 
	This stage is designed to remove remaining 3/5-cycles by a permutation $\pi_{k+1}$. Notice that after Stage I if $\mathscr  C_0$ is a $4$-cycle it must consist of 2 concurrent pairs, and in Stage II we exclude the cycles containing a vertex in $\{w, w^{\oplus r_1} \mid w\in \mathscr C_0\}$. Thus there are at most two 3/5-cycles in $\sigma_k$.
	\medskip
	
	\begin{casepar}{Case 1}
	    If there is no 3/5-cycle, simply let $\pi_{k + 1} := \text{id}$. Note that if $|\mathscr C_0| = 4$, it must be in Case 1.
	\end{casepar}
	\begin{casepar}{Case 2}
	    If there are two 3/5-cycles $\mathscr C_3, \mathscr C_4$, we can always find $\bm v_3\in \mathscr C_3, \bm v_4\in \mathscr C_4$ such that $\bm v_3^{\oplus r_1}, \bm v_4^{\oplus r_1}$ are in $\mathscr C_0$. Perform $\pi_{k + 1} := (\bm v_3, \bm v_4)(\bm v_3^{\oplus r_1}, \bm v_4^{\oplus r_1})$. If $(\bm v_3)_{r_1} = (\bm v_4)_{r_1}$, $\mathscr C_3, \mathscr C_4$ are merged into an even cycle and $\mathscr C_0$ becomes two fix-points. Otherwise, $\mathscr C_0, \mathscr C_3, \mathscr C_4$ are merged into an even cycle of length at most $12$. Let the new even cycle be $\mathscr C_0$.
	\end{casepar}
	\begin{casepar}{Case 3}
	    Suppose there is a unique 3/5-cycle $\mathscr C_3$.
	    \begin{casepar}{Case 3.1}
	        If $\mathscr C_3$ contains a vertex $\bm u'$ such that $\bm u'^{\oplus r_1} \not\in \mathscr C_0, \mathscr C_3$, perform another round of Stage II with $\mathscr C_1=\mathscr C_3, \bm u=\bm u'$; and construct a swap $\pi_{k + 1}$. 
	    \end{casepar}
	    \begin{casepar}{Case 3.2}
	        Otherwise, $\mathscr C_3$ contains a concurrent pair $(\bm u, \bm v)$. Attempt to find a concurrent pair $\bm s, \bm t$ where $\bm s, \bm t \not\in \mathscr C_0, \mathscr C_3$ are contained by different cycles and assume $\bm u_{r_1} = \bm t_{r_1}$. 
	        \begin{casepar}{Case 3.2.1}
	            If such $\bm s,\bm t$ exist, perform $\pi'_{k + 1} := (\bm u, \bm t)(\bm v, \bm s)$ which will merge 3 different cycles including $\mathscr C_3$ and leaves $\mathscr C_0$ invariant.
	        \end{casepar}
	        \begin{casepar}{Case 3.2.2}
	            Otherwise, let $\bm s \in \mathscr C_0$ such that $\bm s^{\oplus r_1}$, denoted by $\bm t$, is not in $\mathscr C_3$. In this case, such $\bm s$ must exist. Also, let the cycle containing $\bm t$ be $\mathscr C_4$; then $\mathscr C_4$ is of odd length. 
	            
	            Assume $\bm u_{r_1} = \bm t_{r_1}$. Thus, $\pi'_{k + 1} := (\bm u, \bm t)(\bm v, \bm s)$ merges $\mathscr C_0,\mathscr C_3,\mathscr C_4$ as an even cycle if $\mathscr C_0=\{\bm v_1,\bm v_2\}$ and $(\bm v_1)_{r_1}\neq (\bm v_2)_{r_1}$. Otherwise $(\bm v_1)_{r_1}=(\bm v_2)_{r_2}$, $\pi'_{k + 1}$ will merge $\mathscr C_3,\mathscr C_4$ as an even cycle and breaks $\mathscr C_0$ into two fix-points. Let the new even cycle be $\mathscr C_0$.
	            
	            Note that it is also the only possible case where the length of the smallest even cycle can be larger than 12. Define $W := \{\bm v_1,\bm v_1^{\oplus r_1},\bm v_2,\bm v_2^{\oplus r_2}\}$. 
	            In this case, every concurrent pair $(\bm s, \bm t)$, where $\bm s, \bm t\not\in W$, is contained by the same cycle in $\sigma_k\pi_{k+1}'$.
	        \end{casepar}
	    \end{casepar}
		If all 3/5-cycles are eliminated, let $\pi_{k + 1} = \pi_{k + 1}'$. But when the remaining $\mathscr C_3$ is a 3-cycle, $\pi_{k + 1}'$ may give a 5-cycle. Consider the (only) two bad instances: 
		
		\ding{91} $\mathscr C_3$ is merged with a 2-cycle in Case 3.1; 
		
		\ding{91} $\mathscr C_3$ is merged with two fix-points in Case 3.2.1. 
		
		\noindent In either bad instance, $\mathscr C_0$ is unchanged, all 3-cycles are eliminated and at most one 5-cycle is left. Try another round of Stage III with $\sigma_k\pi_{k + 1}'$ and get $\pi''_{k + 1}$. Then let $\pi_{k + 1} = \pi_{k + 1}'\pi_{k + 1}''$; and $\sigma_{k + 1}:=\sigma_k\pi_{k+1}$ is 3/5-cycle free.
	\end{casepar}

	\noindent{\textbf{Stage IV.}}\quad 
	After Stage III, $\sigma_{k + 1}$ is 3/5-cycle free, and contains an even cycle. If $\pi_0\pi_1\cdots\pi_{k+1}\in AC_{\zoton}^{(r_1)}$, simply let $\pi_{k + 2} := \text{id}$. If otherwise, we construct 
	$$
	\pi_{k + 2}\in SC_{\zoton}^{(r_1)}\backslash AC_{\zoton}^{(r_1)},
	$$ 
	which preserves an even cycle but forbids 3/5-cycle.
	
    \begin{casepar}{Case 1}
        If there exists a concurrent pair $\bm u, \bm v\notin\mathscr C_0$ contained by different cycles, $|\mathscr C_0|$ can not be greater than $12$ due to the analysis in Case 3.2.2 of Stage III. Let $\mathscr C_1$ and $\mathscr C_2$ be cycles that $\bm u \in \mathscr C_1, \bm v \in \mathscr C_2$. Define
			$$
				T := \mathscr C_0 \cup \{\bm w\mid \mdist^{\sigma_{k+1}}(\bm u,\bm w) \leq 5\}$$
				$$\cup \{\bm w\mid \mdist^{\sigma_{k+1}}(\bm v,\bm w) \leq 5\}.
			$$
		Note that $|T| \leq 34$. Since $n\geq8$ and $2^n\geq2|T|+1$, we can always find a concurrent pair $\bm t, \bm s \not\in T$ where $\bm t_{r_1} = \bm u_{r_1}$. Let $\pi_{k + 2} := (\bm u, \bm t)(\bm v, \bm s)$. Thus, $\sigma_{k + 2}$ still contains $\mathscr C_0$. With the same argument in Stage II, no new 3/5-cycle appears. 
    \end{casepar}
    \begin{casepar}{Case 2}
        Otherwise, consider the size of $\mathscr C_0$. If $|\mathscr C_0| \leq 12$, define $W = \{\bm w, \bm w^{\oplus r_1} \mid \bm w \in \mathscr C_0\}$. If $|\mathscr C_0|>12$, it must comes from Case 3.2.2 of Stage III; and we adopt the definition of $W$ from there. In either case, $|W|\leq24$.
        
        Now, each concurrent pair out of $W$ is contained in the same cycle. If there exist 3 concurrent pairs $\bm u_i, \bm v_i\not\in W,i\in[3]$ and $\bm u_1, \bm u_2, \bm u_3$ are contained in 3 distinct cycles. Let $\tau := (\bm u_1, \bm u_2, \bm u_3)(\bm v_1, \bm v_2, \bm v_3)$ (assuming $(\bm u_1)_{r_1}=(\bm u_2)_{r_1}=(\bm u_3)_{r_1}$). Then $\tau\in AC_{\zoton}^{(r_1)}$ and merges the 3 cycles. Repeat such merging operation until a large even cycle $\mathscr C_1$ of length $\ell\geq 2\times (21\times2+12+1)=110$ appears. Since $n\geq8$ and $2^n\geq|W|+2\ell$, this is inevitable. Let $\pi_{k + 2}'\in AC_{\zoton}^{(r_1)}$ as the merging process.
				
		Denote $\sigma_{k+1}\pi_{k+2}'$ by $\sigma'$ for convenience. Pick 3 distinct concurrent pairs $\bm u_i,\bm v_i\in\mathscr C_1,i\in[3]$ such that
			\begin{gather*}
				\mdist^{\sigma'}(\bm u_i, \bm u_j), \mdist^{\sigma'}(\bm v_i, \bm v_j),\\
				\mdist^{\sigma'}(\bm v_i, \bm u_j), \mdist^{\sigma'}(\bm u_i, \bm v_j) \geq 6
			\end{gather*}
		and $\bm u_i, \bm v_j \not\in W$ for all distinct $i,j \in [3]$. Let $\pi^{i,j} := (\bm u_i, \bm u_j)(\bm v_i, \bm v_j)$. The cycle pattern after $\pi^{i,j}$ is related to the order of the 4 vertices. 
		
		Since 
		$$
		\dist^{\sigma'}(\bm u_i, \bm v_i) = |\mathscr C_1|- \dist^{\sigma'}(\bm v_i, u_i),~\forall i\in[3]
		$$
		and $|\mathscr C_1|=\ell\geq110$,
		there exist distinct $\hat i,\hat j\in[3]$ such that 
		$$
		\dist^{\sigma'}(\bm u_{\hat i}, \bm v_{\hat i}) + \dist^{\sigma'}(\bm v_{\hat j}, \bm u_{\hat j})\geq6
		$$ 
		and 
		$$
		\dist^{\sigma'}(\bm v_{\hat i}, \bm u_{\hat i}) + \dist^{\sigma'}(\bm u_{\hat j}, \bm v_{\hat j})\geq6.
		$$
		
		Define a notation $a \rightsquigarrow b \rightsquigarrow c \rightsquigarrow d$ to represent that $\sigma_{k+1}$ contains a cycle in \fig{patternabcd}.
		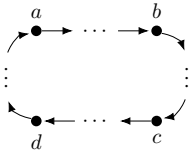
\begin{figure}[H]
    		\centering
    		\scalebox{0.8}{\begin{tikzpicture}[>=stealth,dot/.style={circle,fill=black,inner sep=0pt,minimum size=5pt}]
	\node[dot, label={$a$}] (a) at (0,1.5) {};
	\node[dot, label={$b$}] (b) at (2,1.5) {};
	\node[dot, label=below:{$c$}] (c) at (2,0) {};
	\node[dot, label=below:{$d$}] (d) at (0,0) {};

	\node (ab) at (1,1.5) {$\cdots$};
	\node (bc) at (2.5,0.75) {$\phantom{\vdots}$};
	\node (cd) at (1,0) {$\cdots$};
	\node (da) at (-0.5,0.75) {$\phantom{\vdots}$};
	\node at (2.5,0.825) {$\vdots$};
	\node at (-0.5,0.825) {$\vdots$};

    \draw[-Latex,shorten >= 1pt] (a) -- (ab);
    \draw[-Latex,shorten >= 1pt] (ab) -- (b);
    \draw[-Latex,shorten >= 1pt] (b) to [bend left] (bc);
    \draw[-Latex,shorten >= 1pt] (bc) to [bend left] (c);
    \draw[-Latex,shorten >= 1pt] (c) -- (cd);
    \draw[-Latex,shorten >= 1pt] (cd) -- (d);
    \draw[-Latex,shorten >= 1pt] (d) to [bend left] (da);
    \draw[-Latex,shorten >= 1pt] (da) to [bend left] (a);
\end{tikzpicture}}
    		\caption{Pattern $a \rightsquigarrow b \rightsquigarrow c \rightsquigarrow d$}\label{fig:patternabcd}
		\end{figure}
		
        Here we list possible orders of the 4 vertices.
        \begin{casepar}{\underline{Order 1}}
            $\bm u_i \rightsquigarrow \bm u_j \rightsquigarrow \bm v_j \rightsquigarrow \bm v_i$: Break into 3 cycles with the length of $\dist^{\sigma'}(\bm u_i, \bm u_j)$, $\dist^{\sigma'}(\bm v_j, \bm v_i)$ and $\dist^{\sigma'}(\bm u_j, \bm v_j) + \dist^{\sigma'}(\bm v_i, \bm u_i)$ respectively;
        \end{casepar}
        \begin{casepar}{\underline{Order 2}}
            $\bm u_i \rightsquigarrow \bm u_j \rightsquigarrow \bm v_i \rightsquigarrow \bm v_j$: Break into 3 cycles with the length of  $\dist^{\sigma'}(\bm u_i, \bm u_j)$, $\dist^{\sigma'}(\bm v_j, \bm v_i)$ and $\dist^{\sigma'}(\bm u_j, \bm v_i) + \dist^{\sigma'}(\bm v_j, \bm u_i)$ respectively;
        \end{casepar}
        \begin{casepar}{\underline{Order 3}}
            $\bm u_i \rightsquigarrow \bm v_i \rightsquigarrow \bm u_j \rightsquigarrow \bm v_j$: Remain a cycle of the same length.
        \end{casepar}
		Due to symmetry, other orders are not essentially different from these. 
		Then, let $\pi_{k + 2} := \pi_{k+2}'\pi^{\hat i,\hat j}$; we have $\pi_0\pi_1\cdots\pi_{k+2}\in AC_{\zoton}^{(r_1)}$ and $\sigma\pi_0\cdots\pi_{k+2}$ satisfies the desired properties.
    \end{casepar}
\end{proof}

\begin{proof}[Proof of Lemma \ref{evenh}]
    W.l.o.g, suppose $(1,...,2k)$ is an even cycle in $\sigma$. Define $h_0 \in S_{\zoton}$ as
    \[
        h_0(i) = \left\{\begin{array}{cl}
            i + 1 & i \in [2k - 1] \\
            1 & i = 2k \\
            i & \text{otherwise.}
        \end{array}\right.
    \] 
    It is easy to see that $h_0$ is odd and satisfies $h_0\sigma h_0^{-1}=\sigma$. Since $\sigma, \pi$ has the same cycle pattern, then there exists $h_1\in S_{\{0,1\}^n}$ such that $h_1\sigma h_1^{-1}=\pi$. If $h_1$ is odd, define $h := h_1 h_0$. Otherwise, define $h := h_1$. Thus, $h$ is even and satisfies $h\sigma h^{-1} = \pi$, which finishes the proof.
\end{proof}

\begin{proof}[proof of \lem{GcyclesEven1}]
    W.l.o.g, assume $r_1 = 1$ and $r_2 = 2$. There are at least 12 cycles $\mathscr C_1, \mathscr C_2, \ldots, \mathscr C_k$ with $|\mathscr C_i| \geq 2$ for all $i\in [k]$ in $\sigma$, which implies that there are at least 5 pairs of cycles $\{\mathscr C_1^{(1)}, \mathscr C_2^{(1)}\}, \ldots, \{\mathscr C_1^{(5)}, \mathscr C_2^{(5)}\}$ with the length of $\{a_1, b_1\}, \ldots, \{a_5, b_5\}$ respectively, such that $a_i + b_i$ is even and $\{a_i, b_i\} \neq \{2, 4\}$ for all $i \in [5]$. W.l.o.g, assume $a_1+b_1+a_2+b_2 \equiv a_3 + b_3 + a_4 + b_4 \equiv 0\mod 4$ and the selected 8 cycles are $\mathscr C_1, \ldots, \mathscr C_8$. Let $\ell_1 := a_1 + b_1 + a_2 + b_2$, $\ell_2 := a_3 + b_3 + a_4 + b_4$ and $\ell := \ell_1 + \ell_2$. Choose arbitrary $S \subseteq \{0,1\}^{n-2}$ with size of $\ell/4$ and define $T := \{0,1\}^{n-2} \setminus S$. Due to the fact that $\sigma$ is free of 3/5-cycle and a simple generalization of \prop{new2sum}, there exist $\pi_3 \in S_{\zoton}^{(r_1)}$ and $\tau_3 \in S_{\zoton}^{(r_2)}$ such that $\Supp(\pi_3), \Supp(\tau_3) \subseteq \zoto{2}\times T$ and $\pi_3\tau_3$ is a $|\mathscr C_9|, \ldots, |\mathscr C_k|$-cycle.
    
    In the remaining part of the proof, we provide 4 schemata to construct a $|\mathscr C_1|, \ldots |\mathscr C_8|$-cycle locally with parity-distinct $\pi^{(1)}, \pi^{(2)} \in AC_{\zoton}^{(r_1)}$ and $\tau^{(1)}, \tau^{(2)} \in AC_{\zoton}^{(r_2)}$. Thus, not so strictly speaking, we can adjust the parity of $\pi$ and $\tau$ as required and keep $\pi\tau$ being a $|\mathscr C_1|, \ldots, |\mathscr C_k|$-cycle.
    
    Divide $S = S_1 \sqcup S_2$ where $|S_1| = \ell_1/4$. Let $S_{1,1}, S_{1,2}$ be disjoint subsets of $S_1$ where $|S_{1,i}| = \left\lfloor (a_i + b_i)/4 \right\rfloor$ for $i \in [2]$. Consider the value of $(a_1 + b_1) \bmod 4$:
    \smallskip
    \begin{casepar}{Case 1}
        If $(a_1 + b_1) \equiv 2 \bmod 4$, call 
        \[
            \textsc{TPack}(r_1, r_2, a_1, b_1, a_2, b_2, \zoto{2}\times S_1)
        \]
        with $\pi_{1}, \tau_{1}$ as the outputs.
    \end{casepar}
    \begin{casepar}{Case 2}
        Otherwise, call
        \[
            \textsc{RPack}\left(a_i, b_i, r_1, r_2, \zoto{2}\times S_{1,i}\right)
        \]
        with $\pi_{1,i}, \tau_{1,i}$ as the outputs for $i \in [2]$. Define $\pi_1 = \pi_{1,1} \circ \pi_{1,2}$ and $\tau_1 = \tau_{1,1} \circ \tau_{1,2}$.
    \end{casepar}
    Since $\sigma$ is free of 3/5-cycle, $a_i, b_i$ is valid as inputs of \textsc{TPack} and \textsc{RPack}.

    The proof is based on the following observations: If we swap two pairs of consecutive nodes as shown in Figure \ref{SwapPairs}, then the resulted permutation will have the same cycle pattern with the original one, no matter the two pairs belong to the same cycle or not.
    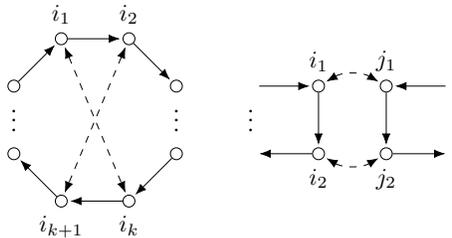
\begin{figure}[H]
        \centering
        \scalebox{0.9}{\begin{tikzpicture}[>=stealth,dot/.style={circle,draw=black,inner sep=0pt,minimum size=5pt}]
	\node[dot, label=above:{}] (p1) at (0,0.7) {};		
	\node[dot, label=below:{}] (p2) at (0,1.7) {};
	\node[dot, label=above:{$i_1$}] (p3) at (0.7,2.4) {};
	\node[dot, label=above:{$i_2$}] (p4) at (1.7,2.4) {};
	\node[dot, label=below:{}] (p5) at (2.4,1.7) {};
	\node[dot, label=below:{}] (p6) at (2.4,0.7) {};
	\node[dot, label=below:{$i_k$}] (p7) at (1.7,0) {};
	\node[dot, label=below:{$i_{k+1}$}] (p8) at (0.7,0) {};
	\draw[-Latex,shorten >= 1pt] (p2) -- (p3);
	\draw[-Latex,shorten >= 1pt] (p3) -- (p4);
	\draw[-Latex,shorten >= 1pt] (p4) -- (p5);
	\node at (2.4,1.3) {$\vdots$};
	\draw[-Latex,shorten >= 1pt] (p6) -- (p7);
	\draw[-Latex,shorten >= 1pt] (p7) -- (p8);
	\draw[-Latex,shorten >= 1pt] (p8) -- (p1);
	\node at (0,1.3) {$\vdots$};
	\draw[dashed, Latex-Latex,shorten >= 1pt,shorten <= 1pt] (p3) -- (p7);
	\draw[dashed, Latex-Latex,shorten >= 1pt,shorten <= 1pt] (p4) -- (p8);
			
	\begin{scope}[shift={(3.5,0.7)}]
	\node (q1) at (0,0) {};
	\node[dot, label=below:{$i_2$}] (q2) at (1,0) {};
	\node[dot, label=below:{$j_2$}] (q3) at (2,0) {};
	\node (q4) at (3,0) {};
	\node (q5) at (0,1) {};
	\node[dot, label=above:{$i_1$}] (q6) at (1,1) {};
	\node[dot, label=above:{$j_1$}] (q7) at (2,1) {};
	\node (q8) at (3,1) {};

	\draw[-Latex,shorten >= 1pt] (q5) -- (q6);
	\draw[-Latex,shorten >= 1pt] (q6) -- (q2);
	\draw[-Latex,shorten >= 0pt] (q2) -- (q1);
	\node at (3,0.6) {$\vdots$};

	\draw[-Latex,shorten >= 1pt] (q8) -- (q7);
	\draw[-Latex,shorten >= 1pt] (q7) -- (q3);
	\draw[-Latex,shorten >= 0pt] (q3) -- (q4);
	\node at (0,0.6) {$\vdots$};

	\draw[dashed, Latex-Latex,shorten >= 1pt,shorten <= 1pt] (q2) to [bend right] (q3);
	\draw[dashed, Latex-Latex,shorten >= 1pt,shorten <= 1pt] (q6) to [bend left] (q7);
	\end{scope}
\end{tikzpicture}}
        \caption{Swap two pairs of consecutive nodes}
        \label{SwapPairs}
    \end{figure}

    Formally, the following equations hold
    \begin{align*}
        &(i_1,i_2,...,i_k,i_{k+1},...)(i_1,i_k)(i_2,i_{k+1})\\
        =&(i_1,i_{k+1},i_3,...,i_k,i_2,i_{k+2},...)\\[3pt]
        &(i_1,..,i_k)(j_1,...,j_l)(i_1,j_1)(i_2,j_2)\\
        =&(i_1,j_2,i_3,...,i_k)(j_1,i_2,j_3,...,j_l).
    \end{align*}

 
    In order to change the concurrent parity of $\tau_{1}$, we simply perform a swap in proper position to the original construction. For example, when $a = b = 2k$, we can construct a $a, b$-cycle with $\pi'\tau'$ or $\pi'\tau''$, where $\tau'$ ($\id$) is concurrently even while $\tau''$ (a swap) is concurrently odd, as pictured in Figure \ref{even:part1split00x}. 
    

    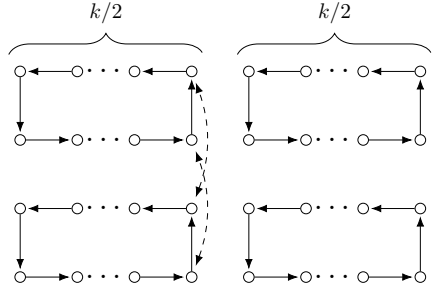
\begin{figure}[H]
        \centering
        \scalebox{0.8}{\begin{tikzpicture}[-Latex,shorten >= 1pt,scale=0.95,>=stealth,dot/.style={circle,draw=black,inner sep=0pt,minimum size=5pt}]
	\foreach \p in {0,1}{
        \foreach \x in {0,1,2,3}
            \foreach \z in {0,1}
                \node[dot] (\p\x\z) at (\x,1.2*\z+\p*2.4) {};

        \draw (\p11) -- (\p01); \draw (\p01) -- (\p00); \draw (\p00) -- (\p10);
        \draw (\p20) -- (\p30); \draw (\p30) -- (\p31); \draw (\p31) -- (\p21);

        \path (\p10) -- node {\Large$\cdots$} (\p20);
        \path (\p11) -- node {\Large$\cdots$} (\p21);
    }

    \draw[Latex-Latex,shorten <= 3pt,shorten >= 3pt,dashed] (030) to [bend right=20] (130);
    \draw[Latex-Latex,shorten <= 3pt,shorten >= 3pt,dashed] (031) to [bend right=20] (131);

    \foreach \p in {0,1}{
        \foreach \x in {0,1,2,3}
            \foreach \z in {0,1}
                \node[dot] (\p\x\z) at (\x+4,1.2*\z+\p*2.4) {};

        \draw (\p11) -- (\p01); \draw (\p01) -- (\p00); \draw (\p00) -- (\p10);
        \draw (\p20) -- (\p30); \draw (\p30) -- (\p31); \draw (\p31) -- (\p21);

        \path (\p10) -- node {\Large$\cdots$} (\p20);
        \path (\p11) -- node {\Large$\cdots$} (\p21);
    }

    \draw [-,decorate,decoration={brace,amplitude=10pt,raise=8pt},yshift=0pt]
        (0-0.2,3.6) -- (3+0.2,3.6) node [black,midway,yshift=28pt] {$k/2$};

    \draw [-,decorate,decoration={brace,amplitude=10pt,raise=8pt},yshift=0pt]
        (4-0.2,3.6) -- (7+0.2,3.6) node [black,midway,yshift=28pt] {$k/2$};
\end{tikzpicture}}
        \caption{$\tau''$ is concurrently odd (left); $\tau'$ is concurrently even (right).}\label{even:part1split00x}
    \end{figure}
    
    We can use the similar method to change the parity of $\tau_{1}$. One technique which need be emphasized is that, by arranging nodes in proper positions, we can ensure the existence of 2 proper consecutive node pairs such that the concurrent swap on them will not change the cycle pattern. As the result, $\tau_1'$ can be constructed, such that $\pi_1\tau_1'$ has the same cycle pattern with $\pi_1\tau_1$, but $\tau_1'$ has different concurrent parity with $\tau_1$. Furthermore, define $\pi_2, \tau_2$ and $\tau_2'$ for $a_3, b_3, a_4, b_4$ in the same way.
    

    Another essential ingredient is to ``rotate'' the constructed permutations in some way. Formally, we exchange $r_1, r_2$ dimensions by a permutation $\rho$, i.e, define $\rho_{r_1, r_2} \in S_{\{0,1\}^n}$ for $i < j$ as
    \[
        \rho_{i,j}: s_1 \ldots s_{i} \ldots s_{j} \ldots s_n \mapsto s_1 \ldots s_{j} \ldots s_{i} \ldots s_n,
    \]
    which maps $s$ to the string constructed by exchanging the $r_1$-th and $r_2$-th elements. Define $\text{switch}: S_{\{0,1\}^n} \to S_{\{0,1\}^n}$ as
    \[
        \text{switch}(\nu) := \rho_{r_1,r_2}^{-1}\circ \nu \circ \rho_{r_1,r_2}.
    \]
    Define 
    \begin{align*}
        \pi^{(1)} := \pi_1\circ\text{switch}(\tau_2)\circ\pi_3\\
        \pi^{(2)} := \pi_1\circ\text{switch}(\tau_2')\circ\pi_3\\
        \tau^{(1)} := \tau_1\circ\text{switch}(\pi_2)\circ\tau_3\\
        \tau^{(2)} := \tau_1'\circ\text{switch}(\pi_2)\circ\tau_3
    \end{align*}
    where $\pi^{(1)},\pi^{(2)} \in SC^{(r_1)}_{\zoton}$ have different concurrent parity, as well as $\tau^{(1)},\tau^{(2)} \in SC^{(r_2)}_{\zoton}$. Note the following facts:
    \begin{itemize}
        \item $\pi_1\circ\tau_1\circ\pi_2\circ\tau_2\circ\pi_3\circ\tau_3$ has the same cycle pattern with $\sigma$;
        \item $\text{switch}(\tau_2)\circ\text{switch}(\pi_2)$ is conjugated with $\pi_2\circ\tau_2$;
        \item $\pi' \circ \tau'$ is conjugated with $\tau' \circ \pi'$ for any $\pi', \tau'$;
        \item These permutations noted with different subscripts act on disjoint supports.
    \end{itemize}
    Thus, it can be shown that $\pi^{(i)}\circ \tau^{(j)}$ has the same cycle pattern with $\sigma$ for all $i,j \in [2]$, which finishes the proof.

\end{proof}

\begin{proof}[Proof of \lem{GcyclesEven2}]
    W.l.o.g, assume $r_1 = 1$ and $r_2 = 2$. Due the restriction of given $\sigma$, there exists cycles $\mathscr C_1, \mathscr C_2$ with the length of $a, b$ respectively, such that $a + b \equiv 0 \bmod 2$ and $a \geq 12$. Due to a similar argument to the one used in the proof of \lem{GcyclesEven1}, it suffices to prove there exist $\pi_1, \ldots, \pi_4 \in SC_{\zoton}^{(r_1)}$ and $\tau_1, \ldots, \tau_4 \in SC_{\zoton}^{(r_2)}$ such that
    \begin{itemize}
        \item $\pi_1, \pi_2, \tau_1, \tau_3$ are concurrently even;
        \item $\pi_3, \pi_4, \tau_2, \tau_4$ are concurrently odd;
        \item $\pi_i\tau_i$ is an $a,b$-cycle for all $i \in [4]$.
    \end{itemize}

    Next, we will construct $\pi', \pi'', \tau', \tau''$ for the following cases such that $\pi'\tau', \pi''\tau''$ are $a,b$-cycles, and $\pi', \pi''$ have different concurrent parity. Let $k := \lfloor a/2 \rfloor$ and $l := \lfloor b/2 \rfloor$.
    \smallskip
    
    \begin{casepar}{Case 1}
        $a, b$ are even and $a = b$:
        \vspace{0.5em}
        
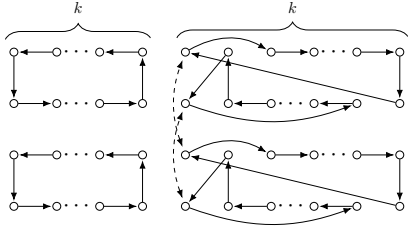
        \begin{center}\scalebox{0.6}{\begin{tikzpicture}[-Latex,shorten >= 1pt,scale=0.95,>=stealth,dot/.style={circle,draw=black,inner sep=0pt,minimum size=5pt}]
	\foreach \p in {0,1}{
        \foreach \x in {0,1,2,3}
            \foreach \z in {0,1}
                \node[dot] (\p\x\z) at (\x,1.2*\z+\p*2.4) {};

        \draw (\p11) -- (\p01); \draw (\p01) -- (\p00); \draw (\p00) -- (\p10);
        \draw (\p20) -- (\p30); \draw (\p30) -- (\p31); \draw (\p31) -- (\p21);
        \node at (1.5,\p*2.4) {\Large$\cdots$};
        \node at (1.5,1.2+\p*2.4) {\Large$\cdots$};
    }

	\draw [-,decorate,decoration={brace,amplitude=10pt,raise=8pt},yshift=0pt]
        (0-0.2,3.6) -- (3+0.2,3.6) node [black,midway,yshift=28pt] {$k$};

    \foreach \p in {0,1}{
        \foreach \x in {0,1,2,3,4,5}
            \foreach \z in {0,1}
                \node[dot] (\p\x\z) at (\x+4,1.2*\z+\p*2.4) {};
        \draw (\p00) to [bend right=20] (\p40);
        \draw (\p40) to (\p30);
        \path (\p30) -- node {\Large$\cdots$} (\p20);
        \draw (\p20) to (\p10);
        \draw (\p10) to (\p11);
        \draw (\p11) to (\p00);

        \draw (\p01) to [bend left] (\p21);
        \draw (\p21) to (\p31);
        \path (\p31) -- node {\Large$\cdots$} (\p41);
        \draw (\p41) to (\p51);
        \draw (\p51) to (\p50);
        \draw (\p50) to (\p01);
    }
    \draw[Latex-Latex,shorten <= 3pt,shorten >= 3pt,dashed] (000) to [bend left=20] (100);
    \draw[Latex-Latex,shorten <= 3pt,shorten >= 3pt,dashed] (001) to [bend left=20] (101);
    \draw [-,decorate,decoration={brace,amplitude=10pt,raise=8pt},yshift=1pt]
        (4-0.2,1.2*1+1*2.4) -- (9+0.2,1.2*1+1*2.4) node [black,midway,yshift=27pt] {$k$};
\end{tikzpicture}}\end{center}
    \end{casepar}
    \begin{casepar}{Case 2}
        $a, b$ are even and $a \neq b$:
        \vspace{0.5em}
        
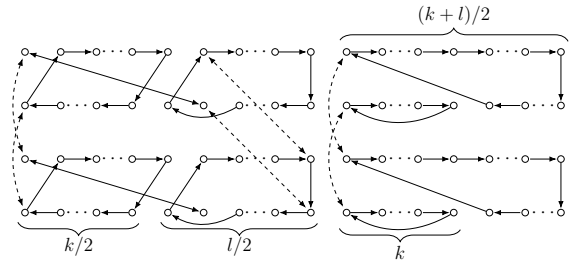
        \begin{center}\scalebox{0.5}{\begin{tikzpicture}[-Latex,shorten >= 1pt,scale=0.95,>=stealth,dot/.style={circle,draw=black,inner sep=0pt,minimum size=5pt}]
	\foreach \p in {0,1}{
        \foreach \x in {0,1,2,3,4,5,6,7,8}
            \foreach \z in {0,1}
                \node[dot] (\p\x\z) at (\x,1.5*\z+\p*3) {};
        \draw (\p00) to (\p11);
        \draw (\p11) to (\p21);
        \path (\p21) -- node {\Large$\cdots$} (\p31);
        \draw (\p31) -- (\p41);
        \draw (\p41) -- (\p30);
        \draw (\p30) -- (\p20);
        \path (\p20) -- node {\Large$\cdots$} (\p10);
        \draw (\p10) -- (\p00);
        
        \draw[Latex-Latex,shorten <= 1pt,shorten >= 1pt] (\p01) -- (\p50);

        \draw (\p40) -- (\p51);
        \draw (\p51) -- (\p61);
        \path (\p61) -- node {\Large$\cdots$} (\p71);
        \draw (\p71) -- (\p81);
        \draw (\p81) -- (\p80);
        \draw (\p80) -- (\p70);
        \path (\p70) -- node {\Large$\cdots$} (\p60);
        \draw (\p60) to [bend left] (\p40);
    }
    \draw[Latex-Latex,shorten <= 3pt,shorten >= 3pt,dashed] (000) to [bend left=20] (100);
    \draw[Latex-Latex,shorten <= 3pt,shorten >= 3pt,dashed] (001) to [bend left=20] (101);

    \draw[Latex-Latex,shorten <= 3pt,shorten >= 3pt,dashed] (080) to (150);
    \draw[Latex-Latex,shorten <= 3pt,shorten >= 3pt,dashed] (081) to (151);

    \draw [-,decorate,decoration={brace,amplitude=10pt,mirror,raise=12pt},yshift=5pt]
        (0-0.2,0) -- (3+0.2,0) node [black,midway,yshift=-30pt] {\Large$k/2$};
    \draw [-,decorate,decoration={brace,amplitude=10pt,mirror,raise=12pt},yshift=5pt]
        (4-0.2,0) -- (8+0.2,0) node [black,midway,yshift=-30pt] {\Large$l/2$};

    \begin{scope}[shift={(9,0)}]\foreach \p in {0,1}{
        \foreach \x in {0,1,2,3,4,5,6}
            \foreach \z in {0,1}
                \node[dot] (\p\x\z) at (\x,1.5*\z+\p*3) {};

        \draw[-Latex,shorten >= 1pt] (\p01) -- (\p11);
        \draw[Latex-,shorten <= 1pt] (\p31) -- (\p21);
        \draw[-Latex,shorten >= 1pt] (\p31) -- (\p41);
        \draw[Latex-,shorten <= 1pt] (\p61) -- (\p51);
        \draw[-Latex,shorten >= 1pt] (\p61) -- (\p60);
        \draw[Latex-,shorten <= 1pt] (\p40) -- (\p50);
        \draw[-Latex,shorten >= 1pt] (\p00) -- (\p10);
        \draw[Latex-,shorten <= 1pt] (\p30) -- (\p20);
        \draw[-Latex,shorten >= 1pt] (\p30) to [bend left] (\p00);
        \draw[-Latex,shorten >= 1pt] (\p40) -- (\p01);

        \foreach \x in {1.25,1.5,1.75}{
            \node at (\x,0+\p*3) {\Large$\cdot$};
            \node at (\x,1.5+\p*3) {\Large$\cdot$};
        }
        \foreach \x in {4.25,4.5,4.75}{
            \node at (\x,1.5+\p*3) {\Large$\cdot$};
        }
        \foreach \x in {5.25,5.5,5.75}{
            \node at (\x,0+\p*3) {\Large$\cdot$};
        }
    }
    \draw[Latex-Latex,shorten >= 3pt,shorten <= 3pt,dashed] (000) to [bend left] (100);
    \draw[Latex-Latex,shorten >= 3pt,shorten <= 3pt,dashed] (001) to [bend left] (101);
    \draw [-,decorate,decoration={brace,amplitude=10pt,mirror,raise=8pt},yshift=0pt]
        (6+0.2,1.5*1+1*3) -- (0-0.2,1.5*1+1*3) node [black,midway,yshift=28pt] {\Large$(k+l)/2$};
    \draw [-,decorate,decoration={brace,amplitude=10pt,mirror,raise=12pt},yshift=0pt]
        (-0.2,0) -- (3+0.2,0) node [black,midway,yshift=-30pt] {\Large$k$};
    \end{scope}

\end{tikzpicture}}\end{center}
        
        The constructions for the cases where $k$ or $l$ is not even are similar.
    \end{casepar}
    \begin{casepar}{Case 3}
        $a, b$ are odd and $a = b$:
        \vspace{0.5em}
        
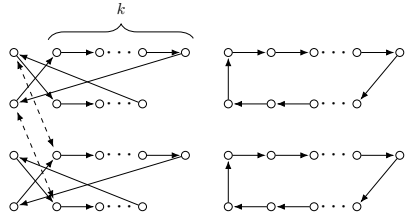
        \begin{center}\scalebox{0.6}{\begin{tikzpicture}[-Latex,shorten >= 1pt,scale=0.95,>=stealth,dot/.style={circle,draw=black,inner sep=0pt,minimum size=5pt}]
			\foreach \p in {0,1}{
		        \foreach \x in {0,1,2,3}
		            \foreach \z in {0,1}
		                \node[dot] (\p\x\z) at (\x,1.2*\z+\p*2.4) {};
		        \node [dot] (\p41) at (4,1.2+\p*2.4) {};

		        \draw (\p01) -- (\p10);
		        \draw (\p10) -- (\p20);
		        \path (\p20) -- node {\Large$\cdots$} (\p30);
		        \draw (\p30) -- (\p01);

		        \draw (\p00) -- (\p11);
		        \draw (\p11) -- (\p21);
		        \path (\p21) -- node {\Large$\cdots$} (\p31);
		        \draw (\p31) -- (\p41);
		        \draw (\p41) -- (\p00);
		    }
		    \draw[Latex-Latex,shorten <= 3pt,shorten >= 3pt,dashed] (010) to (100);
		    \draw[Latex-Latex,shorten <= 3pt,shorten >= 3pt,dashed] (011) to (101);

		    \draw [-,decorate,decoration={brace,amplitude=10pt,raise=8pt},yshift=0pt]
		        (1-0.2,3.6) -- (4+0.2,3.6) node [black,midway,yshift=28pt] {$k$};

		    \begin{scope}[shift={(5,0)}]
		    	\foreach \p in {0,1}{
			        \foreach \x in {0,1,2,3}
			            \foreach \z in {0,1}
			                \node[dot] (\p\x\z) at (\x,1.2*\z+\p*2.4) {};
			        \node [dot] (\p41) at (4,1.2+\p*2.4) {};

			        \draw (\p00) -- (\p01);
			        \draw (\p01) -- (\p11);
			        \draw (\p11) -- (\p21);
			        \path (\p21) -- node {\Large$\cdots$} (\p31);
			        \draw (\p31) -- (\p41);
			        \draw (\p41) -- (\p30);
			        \path (\p30) -- node {\Large$\cdots$} (\p20);
			        \draw (\p20) -- (\p10);
			        \draw (\p10) -- (\p00);
			    }
		    \end{scope}
		\end{tikzpicture}}\end{center}
    \end{casepar}
    \begin{casepar}{Case 4}
        $a, b$ are odd, $b \geq 7$ and $a \neq b$:
        \vspace{0.5em}
        
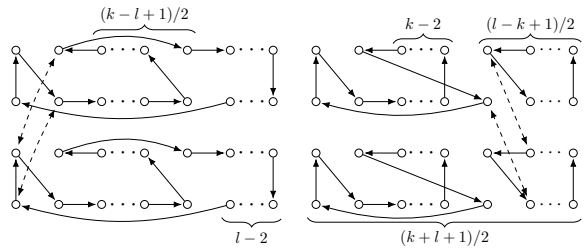
        \begin{center}\scalebox{0.6}{\begin{tikzpicture}[-Latex,shorten >= 1pt,scale=0.95,>=stealth,dot/.style={circle,draw=black,inner sep=0pt,minimum size=5pt}]
			\foreach \p in {0,1}{
		        \foreach \x in {0,1,2,3,4,5,6}
		            \foreach \z in {0,1}
		                \node[dot] (\p\x\z) at (\x,1.2*\z+\p*2.4) {};

		        \draw (\p00) -- (\p01);
		        \draw (\p01) -- (\p10);
		        \draw (\p10) -- (\p20);
		        \foreach \x in {2.25,2.5,2.75,5.25,5.5,5.75}{
		            \node at (\x,\p*2.4) {\Large$\cdot$};
		            \node at (\x,1.2+\p*2.4) {\Large$\cdot$};
		        }
		        \draw (\p30) -- (\p40);
		        \draw (\p40) -- (\p31);
		        \draw (\p21) -- (\p11);
		        \draw (\p11) to [bend left=20] (\p41);
		        \draw (\p41) -- (\p51);
		        \draw (\p61) -- (\p60);
		        \draw (\p50) to [bend left=15] (\p00);
		    }
		    \draw[Latex-Latex,shorten <= 3pt,shorten >= 3pt,dashed] (000) to [bend left=5] (110);
		    \draw[Latex-Latex,shorten <= 3pt,shorten >= 3pt,dashed] (001) to [bend left=5] (111);
		    \draw [-,decorate,decoration={brace,amplitude=6pt,raise=8pt},yshift=2pt]
		        (2-0.2,1.2*1+1*2.4) -- (4+0.2,1.2*1+1*2.4) node [black,midway,yshift=20pt] {$(k-l+1)/2$};
		    \draw [-,decorate,decoration={brace,amplitude=6pt,mirror,raise=12pt},yshift=4pt]
		        (5-0.2,0) -- (6+0.2,0) node [black,midway,yshift=-25pt] {$l-2$};

			\begin{scope}[shift={(7,0)}]
				\foreach \p in {0,1}{
			        \foreach \x in {0,1,2,3,4,5,6}
			            \foreach \z in {0,1}
			                \node[dot] (\p\x\z) at (\x,1.2*\z+\p*2.4) {};

			        \draw (\p00) -- (\p01);
			        \draw (\p01) -- (\p10);
			        \draw (\p10) -- (\p20);
			        \path (\p20) -- node {\Large$\cdots$} (\p30);
			        \draw (\p30) -- (\p31);
			        \path (\p31) -- node {\Large$\cdots$} (\p21);
			        \draw (\p21) -- (\p11);
			        \draw (\p11) -- (\p40);
			        \draw (\p40) to [bend left=15] (\p00);

			        \draw (\p41) -- (\p50);
			        \path (\p50) -- node {\large$\cdots$} (\p60);
			        \draw (\p60) -- (\p61);
			        \path (\p51) -- node {\large$\cdots$} (\p61);
			        \draw (\p51) -- (\p41);
			    }
			    \draw[Latex-Latex,shorten <= 3pt,shorten >= 3pt,dashed] (050) to (140);
			    \draw[Latex-Latex,shorten <= 3pt,shorten >= 3pt,dashed] (051) to (141);

			    \draw [-,decorate,decoration={brace,amplitude=6pt,raise=8pt},yshift=-4pt]
			        (2-0.2,3.6) -- (3+0.2,3.6) node [black,midway,yshift=20pt] {$k-2$};
			        
                \draw [-,decorate,decoration={brace,amplitude=6pt,raise=8pt},yshift=-4pt]
			        (4-0.2,3.6) -- (6+0.2,3.6) node [black,midway,yshift=20pt] {$(l-k+1)/2$};

			    \draw [-,decorate,decoration={brace,amplitude=6pt,mirror,raise=12pt},yshift=4pt]
			        (0-0.2,0) -- (6+0.2,0) node [black,midway,yshift=-25pt] {$(k+l+1)/2$};
			\end{scope}
		\end{tikzpicture}}\end{center}
        
        The construction for $k \equiv l \bmod 2$ is similar.
    \end{casepar}
    \begin{casepar}{Case 5}
        $a, b$ are odd and $b = 1$:
        \vspace{0.5em}
        
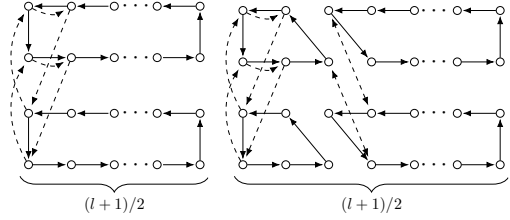
        \begin{center}\scalebox{0.6}{\begin{tikzpicture}[-Latex,shorten >= 1pt,scale=0.95,>=stealth,dot/.style={circle,draw=black,inner sep=0pt,minimum size=5pt}]
			\foreach \p in {0,1}{
		        \foreach \x in {0,1,2,3,4}
		            \foreach \z in {0,1}
		                \node[dot] (\p\x\z) at (\x,1.2*\z+\p*2.5) {};

		        \draw[-Latex,shorten >= 1pt] (\p00) -- (\p10);
		        \draw[-Latex,shorten >= 1pt] (\p10) -- (\p20);
		        \draw[Latex-,shorten <= 1pt] (\p40) -- (\p30);
		        \draw[-Latex,shorten >= 1pt] (\p40) -- (\p41);
		        \draw[-Latex,shorten >= 1pt] (\p41) -- (\p31);
		        \draw[Latex-,shorten <= 1pt] (\p11) -- (\p21);
		        \draw[-Latex,shorten >= 1pt] (\p11) -- (\p01);
		        \draw[-Latex,shorten >= 1pt] (\p01) -- (\p00); 

		        \foreach \x in {2.25,2.5,2.75}{
		            \node at (\x,0+\p*2.5) {\Large$\cdot$};
		            \node at (\x,1.2+\p*2.5) {\Large$\cdot$};
		        }
		    }
		    \draw[-Latex,shorten >= 3pt,dashed] (000) to [bend left] (100);
		    \draw[-Latex,shorten >= 3pt,dashed] (100) to [bend right] (110);
		    \draw[-Latex,shorten >= 3pt,dashed] (110) to (000);
		    \draw[-Latex,shorten >= 3pt,dashed] (001) to [bend left] (101);
		    \draw[-Latex,shorten >= 3pt,dashed] (101) to [bend right] (111);
		    \draw[-Latex,shorten >= 3pt,dashed] (111) to (001);
		    \draw [-,decorate,decoration={brace,amplitude=10pt,mirror,raise=12pt},yshift=5pt]
		        (-0.2,0) -- (4+0.2,0) node [black,midway,yshift=-30pt] {$(l+1)/2$};

			\begin{scope}[shift={(5,0)}]
				\foreach \p in {0,1}{
			        \foreach \x in {0,1,2,3,4,5,6}
			            \foreach \z in {0,1}
			                \node[dot] (\p\x\z) at (\x,1.2*\z+\p*2.4) {};

			        \draw (\p00) -- (\p10);
			        \draw (\p10) -- (\p20);
			        \draw (\p20) -- (\p11);
			        \draw (\p11) -- (\p01);
			        \draw (\p01) -- (\p00);

			        \draw (\p30) -- (\p40);
			        \path (\p40) -- node {\Large$\cdots$} (\p50);
			        \draw (\p50) -- (\p60);
			        \draw (\p60) -- (\p61);
			        \draw (\p61) -- (\p51);
			        \path (\p51) -- node {\Large$\cdots$} (\p41);
			        \draw (\p41) -- (\p31);
			        \draw (\p31) -- (\p21);
			        \draw (\p21) -- (\p30);
			    }
			    \draw[-Latex,shorten >= 3pt,dashed] (000) to [bend left] (100);
    		    \draw[-Latex,shorten >= 3pt,dashed] (100) to [bend right] (110);
    		    \draw[-Latex,shorten >= 3pt,dashed] (110) to (000);
    		    \draw[-Latex,shorten >= 3pt,dashed] (001) to [bend left] (101);
    		    \draw[-Latex,shorten >= 3pt,dashed] (101) to [bend right] (111);
    		    \draw[-Latex,shorten >= 3pt,dashed] (111) to (001);

			    \draw[Latex-Latex,shorten <= 3pt,shorten >= 3pt,dashed] (030) to (120);
			    \draw[Latex-Latex,shorten <= 3pt,shorten >= 3pt,dashed] (031) to (121);

			    \draw [-,decorate,decoration={brace,amplitude=10pt,mirror,raise=12pt},yshift=5pt]
			        (0-0.2,0) -- (6+0.2,0) node [black,midway,yshift=-30pt] {$(l+1)/2$};
			\end{scope}
		\end{tikzpicture}}\end{center}
        
        The construction for even $l$ is similar.
    \end{casepar}
    
    Furthermore, recalling the analysis in the proof of \lem{GcyclesEven1}, it is easy to verified that there exist concurrently odd $\rho', \rho'' \in SC_{\zoton}^{(r_2)}$ such that $\pi'\tau'$ has the same cycle pattern with $\pi'\tau'\rho'$, as well as $\pi''\tau''$ and $\rho''$, which finishes the proof.
\end{proof}

\begin{proof}[Proof of Lemma \ref{Odd4}]
We give a constructive proof when $n=3$, the construction can be easily embeded into higher dimension. For $n=3$, let
\begin{align*}
\pi=&(001,011)(101,111)\\
\tau_1=&(010,100,110)(011,101,111)\\
\tau_2=&(001,100,101)(011,110,111)\\
\tau_3=&(001,010,011)(101,110,111)\\
\tau_4=&(001,101,100)(011,111,110).
\end{align*}
For $n=4$, we simply padding $0$ to the string, that is, let $\pi=(0010,0110)(101,1110)$,
$\tau_1=(0100,1000,1100)(0110,1010,1110)$ and ditto for $n>4$.
\end{proof}

\ifCLASSOPTIONcaptionsoff
  \newpage
\fi

\end{document}